\documentclass[a4paper]{article}
\usepackage{amsfonts,amssymb,amsmath,amsthm,cite}
\usepackage{ucs} 
\usepackage[utf8x]{inputenc}
\usepackage{graphicx}
\usepackage[english]{babel}
\usepackage{slashed}
\usepackage{textcomp}
\usepackage{bm}
\usepackage{titlesec}
\usepackage{stackengine}

\usepackage{etoolbox}
\patchcmd{\thebibliography}{\section*}{\section}{}{}
\usepackage{adjustbox}
\usepackage{nccmath}
\usepackage{mathtools}
\usepackage{xcolor}
\usepackage{footnote}
\makesavenoteenv{tabular}
\mathtoolsset{showonlyrefs}
\evensidemargin=\oddsidemargin
\numberwithin{equation}{section}

\newtheorem{theorem}{Theorem}
\newtheorem{lemma}{Lemma}

\begin{document}
\vspace{10mm}
\begin{center}
	\large{\textbf{Three-loop singularity structure for a non-linear sigma model}}
\end{center}
\vspace{2mm}
\begin{center}
	\large{\textbf{P. V. Akacevich${}^{\star}$~~~A. V. Ivanov${}^{\dagger}$~~~I. V. Korenev${}^{\sharp}$}}
\end{center}
\begin{center}
	${}^{\star,\dagger,\sharp}$St. Petersburg Department
	of Steklov Mathematical Institute
	of Russian Academy of Sciences,\\
	27 Fontanka, St. Petersburg, 191023, Russia
\end{center}
\begin{center}
	${}^{\star,\dagger}$Saint Petersburg State University,\\ 	7/9 Universitetskaya Emb., St. Petersburg, 199034, Russia
\end{center}

\begin{center}
${}^{\star}$E-mail: pavel.akacevich@yandex.ru\,\,\,
${}^{\dagger}$E-mail: regul1@mail.ru\,\,\,
${}^{\sharp}$E-mail: jacepool332@gmail.com
\end{center}
\vspace{10mm}
\begin{flushright}
	\textbf{On the 85-th anniversary of PDMI RAS\footnote{The appropriate epigraph can be found in the Russian version.}}~~~
\end{flushright}
\vspace{10mm}

\textbf{Abstract.} The paper is devoted to the three-loop renormalization of the effective action for a two-dimensional non-linear sigma model using the background field method and a cutoff regularization in the coordinate representation. The coefficients of the renormalization constant and the necessary auxiliary vertices are found, as well as the asymptotic expansions of all three-loop diagrams, and their dependence on the type of regularizing function. A comparison is also made with the standard case of a cutoff in the momentum representation.

\vspace{2mm}

\newpage	

\tableofcontents 
\newpage

\section{Introduction}
\label{33:sec:int}
The study of quantum field models plays an important role in modern mathematics and theoretical physics \cite{3,9,10}. There is a large number of approaches suitable for studying various properties. Nevertheless, almost every option faces some mathematical difficulties. Formally, these approaches can be divided into two groups: nonperturbative and perturbative. If the methods of the first category, as a rule, border on fundamental problems in constructing a functional integral \cite{sk-1,sk-2,34-6,34-c-m}, then the second case deals more with computational challenges, which, in particular, include the appearance of divergences in the coefficients of series under consideration, see \cite{6,7,105}.

Despite the long history of research on perturbative decompositions, both physical and mathematical questions remain open within the framework of this concept. Moreover, the second class is much broader and includes, for example, a lack of understanding of the nature of power-law singularities or the dependence of spectral data on the renormalization process. Additional interest in perturbative methods has arisen in recent years due to the development of functorial quantum field theory \cite{sk-14,sk-3,sk-4,sk-5,sk-6,sk-7}. This is due to the fact that a functional integral can be defined by a formal series, the coefficients of which contain divergences, see \cite{sk-16,sksk}, and even more diverse in their appearance than those that were within the framework of the standard paradigm.

Thus, the study of singular densities and integrals, as well as methods of their regularization, is an important and relevant task. Let us pay attention to several popular approaches to regularization: dimensional regularization \cite{19,555}, regularization with the usage of higher covariant derivatives \cite{Bakeyev-Slavnov,29-st,AA-1,AA-2}, the implicit one \cite{chi-0,chi-1,chi-2}, Feynman \cite{FF-1,Bog-R} and Pauli--Villars \cite{Pauli-Villars} regularizations, and cutoff regularizations \cite{w6,w7,w8,Khar-2020,sk-b-19,Sh}. All the mentioned cases have both pros and cons. For example, the dimensional regularization is undisputed leader in multi-loop calculations. Nevertheless, it contains the operation of deforming the dimension of space, which leads to open fundamental questions. Further, the regularization with higher derivatives is an indispensable tool in supersymmetric theories. But at the same time, it is associated with the transition to higher-order operators, the properties of which, as a rule, are poorly studied. The cutoff regularizations are the most intuitively understood, but they can violate important internal symmetries (for example, the gauge one) of models.

This paper is devoted to the study of special cutoff regularization (a family of regularizations) in the coordinate representation, see \cite{34,Ivanov-Kharuk-2020,Ivanov-Kharuk-20222,Ivanov-Akac,Ivanov-Kharuk-2023,Iv-2024-1,Iv-Kh-2024,Kh-2024,Kh-25}, using the example of three-loop calculations for a two-dimensional non-linear sigma model (or principal chiral model), see \cite{HB,HB1,HB2,HB3,HB4,HB7,sig1,sig2,q-2,q-1,q-3,q-6,q-14}. A special feature of such regularization, see \cite{Iv-2024,Ivanov-2022,sk-b-20}, is a direct connection with the averaging operator, an explicit spectral representation, as well as the presence of additional degrees of freedom, which, depending on the model under consideration, can be fixed taking into account physical and computational features. Bypassing the stages of introducing regularization, we can say that it boils down to a special deformation of the Green's function, which in the main order in the two-dimensional case can be expressed by the transition
\begin{equation*}	
-\frac{\ln(|x|\sigma)}{2\pi}\longrightarrow
\frac{\mathbf{f}\big(\Lambda,|x|^2\Lambda^2\big)}{4\pi}+
	\frac{1}{2\pi}
	\begin{cases} 
		~~\ln(\Lambda/\sigma), & \mbox{for}\,\,\,|x|\leqslant1/\Lambda;\\
		-\ln(|x|\sigma), &\mbox{for}\,\,\,|x|>1/\Lambda.
	\end{cases}
\end{equation*}
Here, $\Lambda$ is a regularizing parameter, and $\mathbf{f}(\cdot)$ is an auxiliary\footnote{More detailed properties can be found in Section \ref{33:sec:ps}.} function. At the same time, the paper does not study the supersymmetric case, some results and references for which can be found in \cite{HB5,HB6}.

The choice of the non-linear sigma model is due not only to its important position in physics, see \cite{q-12,q-13,q-10,q-11,q-9,q-5,q-7,q-8,q-4,q-15,q-16,q-17,33-ma-1}, but also to the desire to study the structure of divergences for a theory in which the interaction part contains a derivative. Previously, similar calculations were performed for the Yang--Mills theory in two loops, see \cite{Ivanov-Kharuk-2020,Ivanov-Kharuk-20222}, and were a continuation of the works \cite{29-1-0,29-1-1} on the study of renormalization "scenario". However, the complexity of the four-dimensional theory leads to serious computational difficulties that are not yet possible to solve. Two-dimensional theory, as shown in this paper, can be studied explicitly. At the same time, additional interest arises due to the fact that previously the non-linear sigma model beyond the second correction was studied using only the dimensional regularization. Therefore, it remains unclear what renormalization should look like in the proposed regularization, as well as which additional vertices will appear.

The main results, which are described in detail in Section \ref{33:sec:r}, include the following six points.
\begin{enumerate}
	\item The asymptotic expansions for three-loop diagrams are calculated.
	\item Auxiliary vertices necessary for the renormalization process in three loops are found.
	\item The $\mathcal{R}$-operation is studied using the example of three-loop diagrams.
	\item The main logarithmic singularity $L^2$ for the coupling constant is calculated.
	\item A comparison is made with the case of the standard cutoff in the momentum representation.
	\item The case of quasi-local vertices and some of their properties are studied.
\end{enumerate}

\vspace{2mm}
\textbf{The work is organized as follows.}\\

In \textbf{Section \ref{33:sec:ps}}, we introduce such basic concepts for a two-dimensional non-linear sigma model as the quantum (effective) action, a regularized classical action, a deformed Green's function, and auxiliary vertices, as well as the well-known results for the first two loops are presented.
\textbf{Section \ref{33:sec:r}} presents the main results in the form of theorems, as well as some special cases. In addition, a comparison with the standard cutoff in the momentum representation is considered.
Further, in \textbf{Section \ref{33:sec:vsvs}}, the elements of diagram techniques and the rules for handling them are presented, as well as a number of auxiliary functionals are introduced in both analytical and diagrammatic form.
\textbf{Section \ref{33:sec:cl:1}} describes needed properties of the deformed Green's function, its asymptotic expansion near the diagonal, as well as a number of relations for special functions from the decomposition coefficients.
Then, in \textbf{Section \ref{33:sec:cl}}, we present calculations of asymptotic expansions with respect to the regularizing parameter for all strongly connected three-loop diagrams.
In \textbf{Sections \ref{33:sec:dop}} and \textbf{\ref{33:sec:kva}}, an additional method for checking the obtained decompositions is proposed, as well as a number of auxiliary quasi-local vertices have also been studied.
\textbf{Sections \ref{33:sec:zak}} and \textbf{\ref{33:sec:appl}} contain concluding comments, acknowledgements, and some auxiliary calculations.

\section{Problem statement}
\label{33:sec:ps}

Consider\footnote{In fact, we can consider an arbitrary compact Lie group.} a special unitary group $\mathrm{SU}(n)$ of degree $n\in\mathbb{N}$ and the corresponding Lie algebra $\mathfrak{su}(n)$, see \cite{2}. By symbols $t^a$, where $a\in\{1,\ldots,\mathrm{dim}\,\mathfrak{su}(n)\}$, we notate the generators of the Lie algebra satisfying two conditions
\begin{equation}\label{33-p-1}
[t^a,t^b]=f^{abc}t^c,\,\,\,
\mathrm{tr}(t^at^b)=-\frac{1}{2}\delta^{ab},
\end{equation}
where the coefficients $f^{abc}$ are totally antisymmetric structure constants for $\mathfrak{su}(n)$. The maps $[\,\cdot\,,\,\cdot\,]$ and $\mathrm{tr}(\,\cdot\,)$ denote the commutator and the Killing form, respectively. For convenience, we assume that the adjoint representation is used. Next, we can check that the structure constants satisfy the relations
\begin{equation}\label{33-p-2}
f^{abc}f^{aef}=f^{abf}f^{aec}-f^{acf}f^{aeb}
,\,\,\,
f^{abc}f^{abe}=c_2\delta^{ce}.
\end{equation}
Here, the real constant $c_2$ denotes the eigenvalue of the Casimir operator for the algebra $\mathfrak{su}(n)$. Next, we introduce the two-dimensional Euclidean space $\mathbb{R}^2$. As a rule, the elements of this set are notated by the symbols $x$, $y$, and $z$, and their individual components are represented by Greek letters. Also, unless otherwise specified, the Einstein convention on summation over repeated indices is used.

Consider a smooth mapping $g(\cdot)\in C^{\infty}(\mathbb{R}^2,\mathrm{SU}(n))$ and define a set of functions $C_\mu^a(x)$, where
$a\in\{1,\ldots,\mathrm{dim}\,\mathfrak{su}(n)\}$ and $\mu\in\{1,2\}$, by the equalities
\begin{equation}\label{33-p-3}
C_\mu^a(x)t^a=\big(\partial_{x^\mu}g(x)\big)g^{-1}(x).
\end{equation}
Then the classical action for a two-dimensional non-linear sigma model (or principal chiral model) can be written as
\begin{equation}\label{33-p-4}
S[C]=\int_{\mathbb{R}^2}\mathrm{d}^2x\,
C_\mu^a(x)C_\mu^a(x).
\end{equation}
Further, the effective (quantum) action is symbolically\footnote{At the moment, the functional integral method contains a number of significant open mathematical questions. The word "symbolically" means that the object is actually defined by a formal series.} represented by the following functional integral
\begin{equation}\label{33-p-5}
W=-\mathrm{ln}\bigg(\int_{\mathcal{H}}\mathcal{D}g\,e^{-S[C]/4\gamma^2}\bigg).
\end{equation}
Here, $\gamma$ denotes the coupling constant and will subsequently play the role of a small parameter. The set $\mathcal{H}$ denotes the "domain" of integration and contains functions with specified behavior\footnote{They should be determined based on physical considerations.} at infinity (boundary).

Moving on to a more specific formulation, let us use the background field method, see \cite{102,103,24,25,26,33-d-1}. To do this, first perform the shift $$g(x)=\exp(\gamma\phi(x))h(x)$$ in integral\footnote{It is assumed that the functional integral has the standard properties for change of variables valid for an ordinary integral. The new "measure", which may differ by the Jacobian of the transformation, will be denoted by $\mathcal{D}^\prime$. Thus, we get $\mathcal{D}g\to\mathcal{D}^\prime\phi$.} \eqref{33-p-5}. In this case, the set $\mathcal{H}$ is transferred to the new one $\mathcal{H}_0$. According to standard theory, the function $h(\cdot)$ is called the background field, while $\phi(\cdot)$ is a fluctuation field. Thus, we move from integration over the fields $g(\cdot)$ to integration over the fluctuations $\phi(\cdot)$ near the selected fixed field $h(\cdot)$. Next, we assume that the background field $h(\cdot)$ satisfies the quantum equation of motion, see \cite{Fa-1,Ba-1,Ba-2}, and the boundary conditions for functions from $\mathcal{H}$.

Let us introduce a few auxiliary notations. The symbols $B_\mu^a(x)$ denote the corresponding coefficients from the equality
\begin{equation}\label{33-p-6}
B_\mu^a(x)t^a=\big(\partial_{x^\mu}h(x)\big)h^{-1}(x).
\end{equation}
In further, they are also called background ones, since they are built using only the function $h(x)$. This is not confusing. Also, for convenience, the symbol $B_\mu^{ac}(x)$ with two upper subscripts denotes the construction $B_\mu^b(x)f^{abc}$. Next, we define the covariant derivative and the Laplace operator in local coordinates by the equations
\begin{equation}\label{33-p-7}
D_{\mu}^{ab}(x)=\delta^{ab}\partial_{x^\mu}-B_\mu^{ab}(x),
\end{equation}
\begin{equation}\label{33-o-3}
A^{ab}(x)=
-\frac{1}{4}\Big(D_\mu^{ab}(x)\partial_{x_\mu}+\partial_{x_\mu}D_\mu^{ab}(x)\Big)
=A_0(x)\delta^{ab}/2+V^{ab}(x)/2,
\end{equation}
where the auxiliary operators are represented as
\begin{equation}\label{33-o-4}
V^{ab}(x)=B^{ab}_\mu(x)\partial_{x_{\mu}}/2+\partial_{x_{\mu}}B^{ab}_\mu(x)/2,
\end{equation}
\begin{equation}\label{33-o-31}
A_0(x)=-\partial_{x_\mu}\partial_{x^\mu}
=-\partial_{x_1}\partial_{x^1}-\partial_{x_2}\partial_{x^2}.
\end{equation}
The Green's function for the operator $A^{ab}(x)$ is notated by the symbol $G^{bc}(x,y)$. By construction, it satisfies the standard equation
\begin{equation}\label{33-p-15}
A^{ab}(x)G^{bc}(x,y)=\delta^{ac}\delta(x-y).
\end{equation}

Next, decomposing the exponential $\exp(\gamma\phi(x))$ into a series in powers of the coupling constant, we can check the following relation for the classical action
\begin{equation}\label{33-p-8}
\frac{S[C]}{4\gamma^2}=
\frac{S[B]}{4\gamma^2}+
\frac{1}{2\gamma}\Gamma_1[\phi]+
\frac{1}{2}S_2[\phi]-\frac{1}{2}
\sum_{k=3}^{+\infty}\frac{\gamma^{k-2}}{k!}\Gamma_{k}[\phi],
\end{equation}
where
\begin{equation}\label{33-p-9}
S_2[\phi]=\int_{\mathbb{R}^2}\mathrm{d}^2x\,\phi^a(x)A^{ab}(x)\phi^b(x),\,\,\,
\Gamma_1[\phi]=-\int_{\mathbb{R}^2}\mathrm{d}^2x\,\phi^a(x)\partial_{x_\mu}B^a_\mu(x),
\end{equation}
\begin{equation}\label{33-p-10}
\Gamma_{k+2}[\phi]=-\int_{\mathbb{R}^2}\mathrm{d}^2x\,
\big(\partial_{x_\mu}\phi^{a_1}(x)\big)\bigg(\prod_{i=1}^{k}\phi^{a_ia_{i+1}}(x)\bigg)D^{a_{k+1}b}_\mu(x)\phi^b(x),\,\,\,\mbox{}\,\,\,k\geqslant1.
\end{equation}
In the latter equation, we have used the notation $\phi^{ab}(x)=f^{acb}\phi^c(x)$. Note that in the case of odd indices, the covariant derivative $D^{ab}_\mu(x)$ can be replaced by the background field $-B_\mu^{ab}(x)$, which is a consequence of the antisymmetry of the structure constants. Next, we use the standard method of transition to a formal decomposition, see \cite{3,sk-2,Vas-98}.
To do this, consider\footnote{Prior to the introduction of any regularization, it is formal in nature.} the auxiliary partition function
\begin{equation}\label{33-p-11}
Z[j]=\int_{\mathcal{H}_0}\mathcal{D}^\prime\phi\,\exp\bigg(-\frac{1}{2}S_2[\phi]+\int_{\mathbb{R}^2}\mathrm{d}^2x\,\phi^a(x)j^a(x)\bigg)=\frac{e^{g(G,j)}}{\sqrt{\det(A)}}=e^{g(G,j)}\sqrt{\det(G)}
,
\end{equation}
where the auxiliary quadratic form can be written as
\begin{equation}\label{33-p-18}
g(G,j)=\frac{1}{2}\int_{\mathbb{R}^2}\mathrm{d}^2x\int_{\mathbb{R}^2}\mathrm{d}^2y\,
j^a(x)G^{ab}(x,y)j^b(y).
\end{equation}
The functional $Z[j]$ depends on the smooth field $j^a(\cdot)$. Thus, if we introduce a notation for the sum of the vertices
\begin{equation}\label{33-p-12}
\Theta[\gamma,\phi]=-
\frac{1}{2\gamma}\Gamma_1[\phi]+\frac{1}{2}
\sum_{k=3}^{+\infty}\frac{\gamma^{k-2}}{k!}\Gamma_{k}[\phi],
\end{equation}
then representation \eqref{33-p-5} for the quantum action is rewritten as
\begin{equation}\label{33-p-13}
W=\frac{S[B]}{4\gamma^2}-\frac{1}{2}\ln\det(G)-\mathrm{ln}\bigg(e^{\Theta[\gamma,\delta_j]}e^{g(G,j)}\Big|_{j=0}\bigg),
\end{equation}
where the symbol $\delta_j$ denotes the variational derivative with respect to the auxiliary field, see \cite{Vas-98}, which is defined by the equality
\begin{equation}\label{33-p-14}
\delta_{j^a(x)}j^b(y)=\frac{\delta j^b(y)}{\delta j^a(x)}=\delta^{ab}\delta(x-y).
\end{equation}
The second term in \eqref{33-p-13} is a formal series with respect to the coupling constant $\gamma$.
Let us make two important notes. First, the quantum action is actually a functional that depends on the background field, that is, $W=W[B]$. Secondly, representation \eqref{33-p-13} contains one well-known problem, see \cite{sk-b-19,Ivanov-Akac}, which consists in the presence of divergences. As a result, it is necessary to introduce a regularization, that is, to make the transition
\begin{equation}\label{33-p-19}
W[B]\xrightarrow{reg.} W_{\mathrm{reg}}[B,\Lambda],
\end{equation}
where the last object in each order of $\gamma$ is finite. In the framework of this work, we use a cutoff regularization in the coordinate representation, which boils down to smoothing fields in the classical action. Let $\omega(\cdot)$ be the kernel of an integral operator $\mathrm{O}^\Lambda_{(\cdot)}$, where $\Lambda>0$, acting from the left, using the selected variable, on a set of continuous functions according to the rule
\begin{equation}\label{33-p-21}
\mathrm{O}_x^\Lambda f(x)=\int_{\mathbb{R}^2}\mathrm{d}^2y\,\omega(|y|)f(x+y/\Lambda),
\end{equation}
and satisfies the conditions
\begin{equation}\label{33-p-20}
\mathrm{supp}\big(\omega(|\cdot|)\big)\subset\{y\in\mathbb{R}^2:|y|\leqslant1/2\}
,\,\,\,\int_{\mathbb{R}^2}\mathrm{d}^2y\,\omega(|y|)=1.
\end{equation}
We also require that the function $\ln(|x|)$ transformed using this integral operator in the following way
\begin{equation}\label{33-p-22}
s(x)=
\mathrm{O}_x^\Lambda\mathrm{O}_x^\Lambda\ln(|x|)
=
\int_{\mathbb{R}^2}\mathrm{d}^2z_1\int_{\mathbb{R}^2}\mathrm{d}^2z_2\,
\omega(|z_1|)\omega(|z_2|)\ln(|x+z_1/\Lambda+z_2/\Lambda|)
\end{equation} 
belongs to the class $s(\cdot)\in C^2(\mathbb{R}^2,\mathbb{R})$. Then we introduce the regularized classical action \eqref{33-p-8} for the two-dimensional non-linear sigma model as
\begin{equation}\label{33-p-23}
\frac{S[B]}{4\gamma^2}+
\frac{1}{2}S_2[\mathrm{O}^\Lambda\phi]-\Theta[\gamma,\mathrm{O}^\Lambda\phi]+
\int_{\mathbb{R}^2}\mathrm{d}^2x\,\Big(\phi^a(x)A_0(x)\phi^a(x)-\big(\mathrm{O}_x^\Lambda\phi^a(x)\big)A_0(x)\big(\mathrm{O}_x^\Lambda\phi^a(x)\big)\Big).
\end{equation}
This is equivalent to smoothing all fields $\phi$, except those that enter with the operator $A_0(x)$. From the construction it is clear that we can remove the regularization using the limit $\Lambda\to+\infty$. 

Such a transition move the Laplace operator in the quadratic form to the new one\footnote{Here, the operator $\mathrm{O}$ acts on all functions that are to the right.}
\begin{equation}\label{33-p-24}
A^{ab}_\Lambda(x)=\frac{1}{2}\Big(A_0(x)\delta^{ab}+\mathrm{O}^\Lambda_xV^{ab}(x)\mathrm{O}^\Lambda_x\Big).
\end{equation}
The corresponding new Green's function we notate by $G_{\mathrm{reg}}^{ab}(x,y)$. In addition, we introduce a deformed Green's function, which is obtained after additional application of the averaging operators
\begin{equation}\label{33-p-25}
G_\Lambda^{ab}(x,y)=\mathrm{O}^\Lambda_x\mathrm{O}^\Lambda_yG_{\mathrm{reg}}^{ab}(x,y)=
2G^{\Lambda,\mathbf{f}}(x-y)\delta^{ab}+\ldots,
\end{equation}
where the following notation was used\footnote{The possibility of using the second equality is related to the averaging properties of free fundamental solutions, see \cite{Ivanov-2022,Iv-2024}.}
\begin{equation}\label{33-o-55-1}	
G^{\Lambda,\mathbf{f}}(x)=
-\frac{\mathrm{O}_x^\Lambda\mathrm{O}_x^\Lambda\ln(|x|\sigma)}{2\pi}
=\,\frac{\mathbf{f}\big(\Lambda,|x|^2\Lambda^2\big)}{4\pi}+
\frac{1}{2\pi}
\begin{cases} 
	~~~~~~L, & \mbox{for}\,\,\,|x|\leqslant1/\Lambda;\\
		-\ln(|x|\sigma), &\mbox{for}\,\,\,|x|>1/\Lambda.
	\end{cases}
\end{equation}
Here $L=\ln(\Lambda/\sigma)$, the parameter $\sigma>0$ is an auxiliary one and is used to make the argument dimensionless. The function $G_\Lambda^{ab}(x,y)$ does not depend \footnote{This parameter appears when adding and subtracting auxiliary singular functions, which are more convenient for specific calculations.} on it.
Then, considering the fairness of the equality
\begin{equation}\label{33-p-26}
F[\mathrm{O}^\Lambda\delta_j]e^{g(G_{\mathrm{reg}},j)}\Big|_{j=0}=F[\delta_j]e^{g(G_{\Lambda},j)}\Big|_{j=0}
\end{equation}
for an arbitrary auxiliary functional decomposable in monomials, the regularized quantum action is defined by the formula
\begin{equation}\label{33-p-13-1}
W_{\mathrm{reg}}[B,\Lambda]=
\frac{S[B]}{4\gamma^2}-\frac{1}{2}\big(\ln\det(G_{\mathrm{reg}})-\varkappa_0\big)-\bigg[\mathrm{ln}\bigg(e^{\Theta[\gamma,\delta_j]}e^{g(G_{\Lambda},j)}\Big|_{j=0}\bigg)-\sum_{k=1}^{+\infty}\gamma^{2k}\varkappa_k\bigg],
\end{equation}
where the auxiliary\footnote{The occurrence of such quantities is due to the presence of integrals over a region of infinite volume. Such integrals do not depend on the background field, so such an adjustment is a shift by a constant and does not affect the physical model.} constants $\{\varkappa_k\}_{k=0}^{+\infty}$, which do not depend on the background field, subtract singular densities and tend to infinity when the regularization is removed. Their explicit form is unimportant, since they are responsible only for the normalization of the partition function and do not carry any physical information.

The regularized quantum action $W_{\mathrm{reg}}[B,\Lambda]$ allows a fairly convenient representation in the form of a sum over Feynman diagrams. Explicit definitions for the elements of the diagram technique are presented in Section \ref{33:sec:int:1}, but for now it is enough to know only two facts. 
\begin{enumerate}
	\item Each functional $\Gamma_n[\phi]$ from \eqref{33-p-12} is assigned a vertex (picture) with $n$ external lines. We notate such the vertex with the same symbol $\Gamma_n$, but without the argument.
	\item A line corresponds to the function $G_\Lambda$.
\end{enumerate}
Next, on the set of vertices, we can introduce a linear operator $\mathbb{H}_i^{c(sc)}$, see \cite{I-R}, where $i\in\mathbb{N}\cup\{0\}$, which transforms non-linear combinations of vertices into connected (strongly connected) diagrams with $i$ external\footnote{If $i=0$, then the diagram is called a vacuum one.} lines. It is clear that due to linearity, it is sufficient to define the operator only on monomials. Indeed, let us define a product with non-negative powers
\begin{equation}\label{33-p-16}
\Gamma=\Gamma_1^{k_1}\prod_{i=3}^{+\infty}\Gamma_i^{k_i},\,\,\,
k(\Gamma)=k_1+\sum_{i=3}^{+\infty}k_i<+\infty,
\end{equation}
then $\mathbb{H}_i^{c(sc)}(\Gamma)$ is equal to the sum of all possible connected (strongly connected) diagrams $\{\mathrm{D}_j\}$, which can be obtained by connecting (pairing) $k(\Gamma)-i$ outer ends belonging to the set of vertices $\Gamma$. Thus, the following decomposition holds
\begin{equation}\label{33-p-17}
\mathbb{H}_i^{c(sc)}(\Gamma)=\sum_{j=1}^{n(\Gamma)}\mathrm{D}_j.
\end{equation}
Clearly, the number $n(\Gamma)$ is finite. For convenience, we supplement the definition with the equality $\mathbb{H}_i^{c(sc)}(1)=\delta_{i0}$. In this case, each diagram corresponds to a single analytical expression, which can be obtained by substituting  expressions \eqref{33-p-10} and \eqref{33-p-25} instead of the corresponding elements of the diagrammatic technique.

Finally, using logarithmic properties of the partition function and the background field, see \cite{Vas-98,I-R},  formula \eqref{33-p-13-1} for the regularized quantum action is rewritten as
\begin{equation}\label{33-p-27}
W_{\mathrm{reg}}[B,\Lambda]=
\frac{S[B]}{4\gamma^2}-\frac{1}{2}\big(\ln\det(G_{\mathrm{reg}})-\hat{\varkappa}_0\big)-\bigg[\mathbb{H}_0^{\mathrm{sc}}\exp\bigg(
\sum_{k=3}^{+\infty}\frac{\gamma^{k-2}}{k!2}\Gamma_{k}\bigg)
-\sum_{k=1}^{+\infty}\gamma^{2k}\varkappa_k\bigg].
\end{equation}
This representation is the most transparent and convenient to study. In particular, it can be used as a definition of a functional representation for the regularized quantum action in the form of a formal series with respect to the coupling constant. 

Object \eqref{33-p-27} does not contain ultraviolet divergences. Nevertheless, each coefficient of the series with respect to the coupling constant is a singular function with respect to the regularizing parameter $\Lambda$. Such coefficients diverge when the regularization is removed. In this regard, a renormalization procedure is applied to the regularized quantum action, which reduces (subtracts) the singular components. Thus, the transition is made
\begin{equation}\label{33-p-19-1}
W_{\mathrm{reg}}[B,\Lambda]\xrightarrow{ren.} W_{\mathrm{ren}}[B,\Lambda].
\end{equation}
The coefficients of the resulting series $W_{\mathrm{ren}}[B,\Lambda]$ have a finite limit at $\Lambda\to+\infty$. The procedure itself for converting the regularized action into the renormalized one, depending on the model under consideration and the type of regularization, may have its own characteristics and, thus, differ from a purely multiplicative approach, which consists in redefining constants. In this paper, the process of renormalization is understood as the transition from \eqref{33-p-27} to the functional
\begin{equation}\label{33-p-27-1}
W_{\mathrm{ren}}[B,\Lambda]=
\frac{S[B]}{4\gamma^2_0}-\frac{1}{2}\big(\ln\det(G_{\mathrm{reg}})-\varkappa_0\big)-\bigg[\mathbb{H}_0^{\mathrm{sc}}\exp\bigg(
\sum_{k=3}^{+\infty}\frac{\gamma^{k-2}_0}{k!2}\Big(\Gamma_{k}+\Gamma_{r,k-2}\Big)\bigg)
-\sum_{k=1}^{+\infty}\gamma^{2k}_0\hat{\varkappa}_k\bigg],
\end{equation}
during which we have replaced the coupling constant $\gamma\to\gamma_0=\gamma_0(\Lambda)$,  added  auxiliary vertices $\{\Gamma_{r,k-2}\}_{k=3}^{+\infty}$, and changed\footnote{Such a change is purely formal, since the explicit form of the constants is unimportant.} the constants $\{\varkappa_i\}_{i=0}^{+\infty}\to\{\hat{\varkappa}_i\}_{i=0}^{+\infty}$. Note that in the new vertices, the second index corresponds to the power of $\gamma_0$ with which they enter the formula. 
Moreover, we assume that $\Gamma_{r,k-2}$ is a finite linear combination of vertices, each of which can contain no more than $k$ external lines. Further, the new coupling constant $\gamma_0^2$ is a formal series in powers of $\gamma^2$ with coefficients depending on the regularizing parameter $\Lambda$, that is
\begin{equation}\label{33-p-32}
\frac{1}{\gamma^2_0}=\frac{1}{\gamma^2}-\sum_{i=0}^{+\infty}a_i\gamma^{2i},
\end{equation}
where $a_i=a_i(\Lambda)$.
Note that the last series is formal, so the functions\footnote{For example, exponentiation.} of such a series can be obtained using the formal calculation of coefficients for identical powers of the constant $\gamma$. Calculating the coefficients from \eqref{33-p-32} is a sequential procedure, since each step uses previously obtained data. The same applies to the vertices. Note that when calculating contributions proportional to $\gamma^{2k}$, where $k\in\mathbb{N}$, the vertices $\Gamma_{r,2k-3}$ and $\Gamma_{r,2k-2}$ are determined. At the same time, knowing the two-loop results for the non-linear sigma model, see \cite{Ivanov-Akac,sk-b-19,Iv-244}, we immediately set $\Gamma_{r,1}=0$ for convenience.

After substituting series \eqref{33-p-32} into \eqref{33-p-27-1}, we can perform a summation, thereby obtaining a new series by the constant $\gamma^2$. The coefficients of such a series, if the model is renormalizable, are no longer contain singular components, and their values become finite in the limit of removing the regularization. Let us rewrite the renormalized action as
\begin{equation}\label{33-p-28}
	W_{\mathrm{ren}}[B,\Lambda]=
	\frac{S[B]}{4\gamma^2}+W_0[B,\Lambda]+\gamma^2W_1[B,\Lambda]+\gamma^4W_2[B,\Lambda]+\ldots,
\end{equation}
and write out the first three coefficients explicitly. To do this, we note that series \eqref{33-p-32} can be rewritten as
\begin{equation}\label{33-p-32-1}
\gamma^2_0
=\gamma^2+a_0\gamma^4+\ldots,
\end{equation}
then the functions are represented by the following formulas
\begin{equation}\label{33-p-29}
W_0[B,\Lambda]=-\frac{1}{2}\big(\ln\det(G_{\mathrm{reg}})-\hat{\varkappa}_0\big)-\frac{a_0}{4}S[B],
\end{equation}
\begin{equation}\label{33-p-30}
W_1[B,\Lambda]=-\frac{\mathbb{H}_0^{\mathrm{sc}}(\Gamma_3^2)}{8(3!)^2}
-\frac{\mathbb{H}_0^{\mathrm{sc}}(\Gamma_4+\Gamma_{r,2})}{2(4!)}-\frac{a_1}{4}S[B]+\hat{\varkappa}_1,
\end{equation}
\begin{align}\label{33-p-31}
W_2[B,\Lambda]=&-\frac{\mathbb{H}_0^{\mathrm{sc}}(\Gamma_3^4)}{16(3!)^44!}
-\frac{\mathbb{H}_0^{\mathrm{sc}}(\Gamma_3^2\Gamma_4^{\phantom{1}})}{16(3!)^24!}
-\frac{\mathbb{H}_0^{\mathrm{sc}}(\Gamma_3\Gamma_5)}{4(3!5!)}
-\frac{\mathbb{H}_0^{\mathrm{sc}}(\Gamma_4^2)}{8(4!)^2}
-\frac{\mathbb{H}_0^{\mathrm{sc}}(\Gamma_6)}{2(6!)}
\\\nonumber&-\frac{\mathbb{H}_0^{\mathrm{sc}}(\Gamma_3^2\Gamma_{r,2}^{\phantom{1}})}{16(3!)^24!}
-\frac{\mathbb{H}_0^{\mathrm{sc}}(\Gamma_4\Gamma_{r,2})}{4(4!)^2}
-\frac{\mathbb{H}_0^{\mathrm{sc}}(\Gamma_{r,2}^2)}{8(4!)^2}
-\frac{\mathbb{H}_0^{\mathrm{sc}}(\Gamma_3\Gamma_{r,3})}{4(3!5!)}
-\frac{\mathbb{H}_0^{\mathrm{sc}}(\Gamma_{r,4})}{2(6!)}
\\\nonumber&+a_0\Big(W_1[B,\Lambda]+a_1S[B]/4\Big)-\frac{a_2}{4}S[B]+\hat{\varkappa}_2.
\end{align}
Next, the symbol $\stackrel{\mathrm{s.p.}}{=}$ denotes the equality of singular components with respect to the regularizing parameter. Thus, the renormalizability condition leads to a set of relations
\begin{equation}\label{33-p-33}
W_i[B,\Lambda]\stackrel{\mathrm{s.p.}}{=}0,
\end{equation}
where $i\in\mathbb{N}\cup\{0\}$. Let us consider in more detail the first two relations, which were previously studied for the proposed type of regularization.\\

\noindent\underline{The first loop.} For convenience, we introduce $L_1=L+\mathbf{f}(0)/2=2\pi G^{\Lambda,\mathbf{f}}(0)$. This choice was made for convenience reasons, since in the higher loops a large number of diagrams contain the Green's function with equal to each other arguments (on the diagonal). At the same time, such a shift is a new fixation of the initial conditions in the renormalization process and does not impose additional restrictions on the parameter $\sigma$, which, as before, obeys the condition $\sigma>0$. Taking into account the notation, it can be shown that relation \eqref{33-p-29} leads to equality
\begin{equation}\label{33-p-34}
-\frac{1}{2}\big(\ln\det(G_{\mathrm{reg}})-\hat{\varkappa}_0\big)\stackrel{\mathrm{s.p.}}{=}-\frac{c_2L_1}{16\pi}S[B],
\end{equation}
and thus, we obtain the well-known answer
\begin{equation}\label{33-p-35}
a_0=-\frac{c_2L_1}{4\pi}.
\end{equation}

\noindent\underline{The second loop.} In this case, we use the results obtained in the works \cite{Ivanov-Akac,Iv-244} on the two-loop coefficient. Let us introduce a few additional objects: three auxiliary numbers
\begin{equation}\label{33-p-37}
\alpha_1(\mathbf{f})=A_0(x)\mathbf{f}(|x|^2)\big|_{x=0},
\end{equation}
\begin{equation}\label{33-p-38}
\frac{\theta_1}{2\pi}=
\int_{\mathbb{R}^{2}}\mathrm{d}^2x\,
\Big(A_0(x)G^{1,\mathbf{f}}(x)\Big)
\Big(G^{1,\mathbf{f}}(x)-
G^{1,\mathbf{f}}(0)\bigg),
\end{equation}
\begin{equation}\label{33-p-38-1}
\frac{\theta_2}{2\pi}=
\int_{\mathbb{R}^{4}}\mathrm{d}^2x\mathrm{d}^2y\,
\Big(A_0(x)G^{1,\mathbf{f}}(x)\Big)
\Big(A_0(y)G^{1,\mathbf{f}}(y)\Big)
\Big(G^{1,\mathbf{f}}(x-y)-
G^{1,\mathbf{f}}(x)\bigg),
\end{equation}
as well as the functional
\begin{equation}\label{33-p-36}
\mathrm{R}_0[\phi]=\int_{\mathbb{R}^2}\mathrm{d}^2x\,\phi^a(x)\phi^a(x),
\end{equation}
which corresponds to the vertex $\mathrm{R}_0$ with two external lines. Then it can be shown that
\begin{equation}\label{33-p-39}
-\frac{\mathbb{H}_0^{\mathrm{sc}}(\Gamma_3^2)}{8(3!)^2}
-\frac{\mathbb{H}_0^{\mathrm{sc}}(\Gamma_4)}{2(4!)}+\hat{\varkappa}_1\stackrel{\mathrm{s.p.}}{=}
-\frac{\Lambda^2c_2\alpha_1(\mathbf{f})}{96\pi}\mathbb{H}_0^{\mathrm{sc}}(\mathrm{R}_0)-\frac{L_1c_2^2}{12(4\pi)^2}S[B](3\theta_1+2\theta_2).
\end{equation}
Thus, comparing with formula \eqref{33-p-30}, we get the vertex and the coefficient
\begin{equation}\label{33-p-40}
\Gamma_{r,2}=-\frac{\Lambda^2c_2\alpha_1(\mathbf{f})}{2\pi}\mathrm{R}_0,\,\,\,
a_1=-\frac{L_1c_2^2}{3(4\pi)^2}(3\theta_1+2\theta_2).
\end{equation}

\noindent\underline{The purpose of the work.} It is necessary to obtain asymptotic expansions with respect to the regularizing parameter for all diagrams from formula \eqref{33-p-31}. Find the type of auxiliary vertices $\Gamma_{r,3}$ and $\Gamma_{r,4}$, and calculate the leading logarithmic singularity. Special cases need to be considered for the formulas obtained.

\section{Results}
\label{33:sec:r}
\subsection{The main statements}
\label{33:sec:osn}
Let us introduce a number of auxiliary functions. First, in addition to the values from \eqref{33-p-38} and \eqref{33-p-38-1}, we define a set of numbers
\begin{equation}\label{33-p-38-2}
	\frac{\theta_{k+1}}{2\pi}=
	\int_{\mathbb{R}^{2\times k}}\mathrm{d}^2y_1\ldots\mathrm{d}^2y_k\,
	\Big(A_0(x)G^{1,\mathbf{f}}(y_1)\Big)\cdot\ldots\cdot
	\Big(A_0(y)G^{1,\mathbf{f}}(y_k)\Big)
	\Big(G^{1,\mathbf{f}}(y_{k-1}-y_k)-
	G^{1,\mathbf{f}}(y_{k-1})\bigg),
\end{equation}
where $k>1$, and also\footnote{The functionals $\alpha_9(\mathbf{f})$ and $\alpha_{11}(\mathbf{f})$ are calculated in Section \ref{33:sec:sp}, see property 4. Although the relation $\alpha_{9}(\mathbf{f})=0$ is fulfilled, the answers for diagrams in Section \ref{33:sec:cl:5} still include this functional. We believe, it is useful for additional verification when restoring calculations that are different from the calculation method of the integral itself.}  
\begin{equation}\label{33-r-24-9}
	\alpha_6(\mathbf{f})=\int_{\mathbb{R}^2}\mathrm{d}^2x\,
	\Big(A_0(x)G^{1,\mathbf{f}}(x)\Big)^2,
\end{equation}
\begin{equation}\label{33-r-24-1-9}
	\alpha_7(\mathbf{f})=\frac{1}{2}\int_{\mathbb{R}^2}\mathrm{d}^2x\,
	\Big(A_0(x)G^{1,\mathbf{f}}(x)\Big)^2\mathbf{f}(|x|^2),
\end{equation}
\begin{equation}\label{33-9-24-52}
	\alpha_8(\mathbf{f})=\pi\int_{\mathbb{R}^2}\mathrm{d}^2x\,
	\Big(A_0(x)G^{1,\mathbf{f}}(x)\Big)^2|x|^2,
\end{equation}
\begin{equation}\label{33-9-24-53}
	\alpha_9(\mathbf{f})=\pi\int_{\mathbb{R}^2}\mathrm{d}^2x\,
	A_0(x)G^{1,\mathbf{f}}(x)x^\mu\partial_{x^\mu}
	\bigg(G^{1,\mathbf{f}}(x)-\int_{\mathbb{R}^2}\mathrm{d}^2y\,G^{1,\mathbf{f}}(x-y)A_0(y)G^{1,\mathbf{f}}(y)\bigg)=0,
\end{equation}
\begin{equation}\label{33-9-24-54}
	\alpha_{10}(\mathbf{f})=\int_{\mathbb{R}^2}\mathrm{d}^2x\,
	\Big(A_0(x)G^{1,\mathbf{f}}(x)\Big)\Big(\partial_{x_\mu}G^{1,\mathbf{f}}(x)\Big)\Big(\partial_{x^\mu}G^{1,\mathbf{f}}(x)\Big),
\end{equation}
\begin{equation}\label{33-9-24-55}
	\alpha_{11}(\mathbf{f})=\pi\int_{\mathbb{R}^2}\mathrm{d}^2x\,
	\Big(A_0(x)G^{1,\mathbf{f}}(x)\Big)\Big(x^\mu\partial_{x^\mu}G^{1,\mathbf{f}}(x)\Big)=-\frac{1}{4}.
\end{equation}
Note that the latter functionals depend only on the regularizing function $\mathbf{f}(\cdot)$. At the same time, they do not depend on the regularizing and auxiliary parameters $\Lambda$ and $\sigma$. Secondly, we define a set of auxiliary vertices
\begin{equation}\label{33-d-28-11}
\mathrm{R}_1[\phi]=\int_{\mathbb{R}^2}\mathrm{d}^2x\,\phi^e(x)B_\mu^{eh}(x)B_\mu^{ha}(x)\phi^a(x),
\end{equation}
\begin{equation}\label{33-d-29-1}
\mathrm{R}_2[\phi]=\int_{\mathbb{R}^2}\mathrm{d}^2x\,
	\Big(f^{abc}\phi^b(x)\phi^d(x)f^{cde}\Big)B_\mu^{eh}(x)B_\mu^{ha}(x),
\end{equation}
\begin{equation}\label{33-w-5-1}
\mathrm{R}_3[\phi]=\int_{\mathbb{R}^2}\mathrm{d}^2x\,
	\Big(f^{a_1ga_2}f^{a_1bc}\phi^b(x)\phi^d(x)f^{cde_1}f^{e_2ge_1}\Big)B_\mu^{e_2h}(x)B_\mu^{ha_2}(x),
\end{equation}
\begin{equation}\label{33-d-28-1-1}
\mathrm{R}_4[\phi]=\int_{\mathbb{R}^2}\mathrm{d}^2x\,\phi^e(x)B_\mu^{ea}(x)\partial_{x_\mu}\phi^a(x),
\end{equation}
\begin{equation}\label{33-z-11}
	\hat{\Gamma}_3[\phi]=\int_{\mathbb{R}^2}\mathrm{d}^2x\,f^{abc}\phi^b(x)f^{ced}\phi^e(x)\Big(\partial_{x_\mu}B^{dg}_{\mu}(x)\Big)f^{gha}\phi^h(x),
\end{equation}
\begin{equation}\label{33-d-28-1-2}
	\widetilde{\Gamma}_4[\phi]=\int_{\mathbb{R}^2}\mathrm{d}^2x\,
	f^{a_1b_1a_2}\phi^{b_1}(x)f^{a_2b_2a_3}\phi^{b_2}(x)
	f^{a_3b_3a_4}\phi^{b_3}(x)f^{a_4b_4a_1}\phi^{b_4}(x).
\end{equation}

\begin{theorem}\label{33-t}
Let all the conditions stated above be fulfilled. Then it is possible to exclude nonlocal (disproportionate to the classical action) singular contributions from the three-loop correction \eqref{33-p-31} using the following choice of auxiliary renormalization vertices
\begin{equation*}
\Gamma_{r,3}=-\frac{L_1}{3\pi}\hat{\Gamma}_3,
\end{equation*}
\begin{equation*}
\Gamma_{r,4}=\Lambda^2\bigg(\frac{\alpha_1}{2\pi}-\frac{5\alpha_6}{2}\bigg)
\widetilde{\Gamma}_4
-\frac{15\Lambda^2c_2^2}{2\pi}\mathrm{R}_0\tau_5-
\frac{5L_1}{32\pi^2}\Big(c_2^2\mathrm{R}_1\tau_1+c_2\mathrm{R}_2\tau_2+\mathrm{R}_3\tau_3+c_2^2\mathrm{R}_4\tau_4\Big),
\end{equation*}
where auxiliary finite numbers are defined by the equalities
\begin{align*}
\tau_1=&\,4\theta_3-10\theta_2-2\theta_1+26\alpha_{11},\\
\tau_2=&\,2\alpha_8+34\theta_3-154\theta_2/5+8\theta_1+14\alpha_{11},\\
\tau_3=&\, 8\alpha_8-4\theta_3-164\theta_2/5+8\theta_1-6\alpha_{11},\\
\tau_4=&\,10\alpha_8-60\theta_2+52\theta_1+44\alpha_{11},\\
\tau_5=&-\mathbf{f}(0)/2+\alpha_7-4\alpha_{10}.
\end{align*}
They depend only on the deforming function $\mathbf{f}(\cdot)$. Singular terms proportional to the classical action $S[B]$ or independent of the background field at all can be reduced by a suitable choice of the coefficients $a_2$ and $\hat{\varkappa}_2$. In this case, the structure of the value $a_2$ has the following form
\begin{equation}\label{33-i-7}
a_2=\frac{a_0a_1}{2}-\frac{c_2^3L_1^2}{16(4!)^2\pi^3}\big(2\tau_1-2\tau_2+2\tau_3+\tau_4\Big)+L_1t,
\end{equation}
where $t\in\mathbb{R}$ and it depends only on the regularizing function. $\mathbf{f}(\cdot)$.
\end{theorem}

\begin{theorem}\label{33-tt} Let the assumptions of Theorem \ref{33-t} be fulfilled. Let $\{\mathrm{K}_i(\cdot)\}_{i=1}^4$ be a set of continuous functions such that the relations hold
	\begin{equation*}
		\mathrm{supp}(\mathrm{K}_i)\subset\mathrm{B}_1,\,\,\,
		\int_{\mathbb{R}^2}\mathrm{d}^2x\,\mathrm{K}_i(x)=\tau_i,
	\end{equation*}
	for all index values $i$. Next, using the functionals from Section \ref{33:sec:kva}, we define four vertices
	\begin{equation*}
		\hat{\mathrm{V}}_i[\phi]=\mathrm{V}_i[\mathrm{K}_i,\phi],
	\end{equation*}
	where $i\in\{1,\ldots,4\}$. Then, as an auxiliary vertex $\Gamma_{r,4}$ from Theorem \ref{33-t}, we can choose a functional of the form
	\begin{multline}\label{33-i-2}
		\Lambda^2\bigg(\frac{\alpha_1}{2\pi}-\frac{5\alpha_6}{2}\bigg)
		\widetilde{\Gamma}_4
		-\frac{15\Lambda^2c_2^2}{2\pi}\mathrm{R}_0\tau_5-
		\frac{5L_1}{32\pi^2}\Big(c_2^2\mathrm{R}_1\tau_1\upsilon_1+c_2\mathrm{R}_2\tau_2\upsilon_2+\mathrm{R}_3\tau_3\upsilon_3+c_2^2\mathrm{R}_4\tau_4\upsilon_4\Big)-\\-\frac{5}{32\pi}
		\Big(c_2\hat{\mathrm{V}}_1\vartheta_1+c_2\hat{\mathrm{V}}_2\vartheta_2+\hat{\mathrm{V}}_3\vartheta_3+c_2\hat{\mathrm{V}}_4\vartheta_4\Big),
	\end{multline}
	in which the coefficients $\{\upsilon_i\}_{i=1}^4$ and $\{\vartheta\}_{i=1}^4$ are subject to the following conditions
	\begin{align*}
		\upsilon_1-\vartheta_1+\vartheta_2-\vartheta_3=&\,1,\\
		\upsilon_2-\vartheta_2+\vartheta_4/2=&\,1,\\
		\upsilon_3-\vartheta_3=&\,1,\\
		\upsilon_4+\vartheta_4=&\,1.
	\end{align*}
	In this case, the third coefficient $a_2$ of the renormalization constant becomes equal to
	\begin{equation}\label{33-i-6}
		a_2\to a_2=\frac{a_0a_1}{2}-\frac{c_2^3L_1^2}{16(4!)^2\pi^3}\big(2\tau_1\upsilon_1-2\tau_2\upsilon_2+2\tau_3\upsilon_3+\tau_4\upsilon_4\Big)+L_1t_1,
	\end{equation}
	where the number $t_1\in\mathbb{R}$ depends on the regularizing function $\mathbf{f}(\cdot)$, the set of coefficients $\{\vartheta_i\}_{i=1}^4$, and the set of kernels $\{\mathrm{K}_i(\cdot)\}_{i=1}^4$, and it does not depend on $\{\upsilon_i\}_{i=1}^4$. In particular, fixing the values as follows $\upsilon_1=\upsilon_2=\upsilon_3=\upsilon_4=0$, we obtain coefficients for quasi-local vertices in the form
	\begin{equation*}
		\vartheta_1=-1/2,\,\,\,
		\vartheta_2=-1/2,\,\,\,
		\vartheta_3=-1,\,\,\,
		\vartheta_4=1.
	\end{equation*}
	In this case, the coefficient for the renormalization constant takes the form
	\begin{equation}\label{33-i-1}
		a_2\to a_2=\frac{a_0a_1}{2}+L_1t_2,
	\end{equation}
	where $t_2\in\mathbb{R}$ depends on the regularizing function $\mathbf{f}(\cdot)$ and the set of kernels $\{\mathrm{K}_i(\cdot)\}_{i=1}^4$.
\end{theorem}

\noindent\textbf{Additional comments on Theorem \ref{33-t}.}
\begin{enumerate}
	\item The three-loop coefficient $a_2$ for the coupling constant in the main order is proportional to the logarithm squared, which is completely consistent with the general theory.
	\item The coefficient $a_2$ does not contain power-law singularities.
	\item In addition to the functionals $\theta_i$ containing the differences of deformed Green's functions, the answer also contains integrals of the special form $\alpha_8$ and $\alpha_{11}$. It can be noted that $\alpha_{11}$ is equal to $-1/4$ and does not depend on the regularizing function. At the same time, $\alpha_8$ is the result of deformation and is somewhat artificial in nature. It would be zero in the absence of regularization. Indeed, we have
	 \begin{equation*}
	 \pi\int_{\mathbb{R}^2}\mathrm{d}^2x\,
	 \Big(A_0(x)G^{1,\mathbf{f}}(x)\Big)^2|x|^2\longrightarrow
	 \pi\int_{\mathbb{R}^2}\mathrm{d}^2x\,
	 \Big(A_0(x)\frac{\ln(|x|\sigma)}{2\pi}\Big)^2|x|^2=
	 \pi\int_{\mathbb{R}^2}\mathrm{d}^2x\,
	 \Big(|x|\delta(x)\Big)^2=0.
	 \end{equation*} 
 \item Auxiliary (counter) vertices do not contain a singularity of the form $\Lambda^2L$.
 \item The triple vertex $\hat{\Gamma}_3$ is an intermediate (auxiliary) object. It can be noted that it is proportional to $\partial_{x_\mu}B^a_\mu(x)$, and thus is involved in solving the classical equation of motion. It is assumed that in subsequent renormalization orders, additional terms will arise, which in total will lead to a functional equal to zero on the set of solutions of the quantum equation of motion. Thus, this amount will be zero due to the choice of the background field.
 \item The part of the coefficient $a_2$ from \eqref{33-i-7}, proportional to $L_1^2$, consists of two parts. The first one is standard and appears in every renormalization scheme. It can be called a "reaction" of the renormalization process on the first two loops. The second term is a regularization feature and a consequence of the requirement for locality of auxiliary vertices $\mathrm{R}_i[\,\cdot\,]$. It depends on the regularizing function $\mathbf{f}(\cdot)$. Note that in dimensional regularization, such a term is zero.
\end{enumerate}

\noindent\textbf{Additional comments on Theorem \ref{33-tt}.}
\begin{enumerate}
	\item The abandonment of the locality of auxiliary vertices and the transition to quasi-local functionals helped solve the problem of the presence of a second term in \eqref{33-i-7}.
	\item The introduced nonlocal vertices $\mathrm{\hat{V}}_i[\,\cdot\,]$ are not the only option. In particular, they depend on the kernels of $\mathrm{K}_i$. Moreover, we can enter vertices that differ significantly from the existing ones, see Section \ref{33:sec:zak}.
	\item After adjusting the coefficient $a_2$, the renormalized coupling constant can be used to calculate the $\beta$-function. Indeed, we can rewrite formula \eqref{33-p-32} as
	\begin{equation*}
\gamma_0^2=\gamma^2+a_0^{\phantom{1}}\gamma^4+(a_0^2+a_1^{\phantom{1}})\gamma^6+(a_0^3+2a_0^{\phantom{1}}a_1^{\phantom{1}}+a_2^{\phantom{1}})\gamma^8+\ldots,
	\end{equation*}
	then for the equation
	\begin{equation*}
\Lambda\frac{\mathrm{d}\gamma^2_0}{\mathrm{d}\Lambda}=\beta_1^{\phantom{1}}\gamma^4_0+\beta_2^{\phantom{1}}\gamma^6_0+\beta_3^{\phantom{1}}\gamma^8_0+\ldots
	\end{equation*}
	we obtain the following coefficients
	\begin{equation*}
		\beta_1=-\frac{c_2}{4\pi}
		,\,\,\,
		\beta_2=-\frac{c_2^2}{3(4\pi)^2}(3\theta_1+2\theta_2)
		,\,\,\,
		\beta_3=t_2
		.
	\end{equation*}
	\item The problem of calculating $t_2$ for the proposed deformation remains open.
\end{enumerate}

\subsection{Proofs}
\noindent\textbf{Proof of Theorem \ref{33-t}.}
Let us consider formula \eqref{33-p-31} for the functional $W_2[B,\Lambda]$ and substitute into it explicit asymptotic expansions with respect to the regularizing parameter $\Lambda$ for each individual term. Calculations of such decompositions are carried out in Section \ref{33:sec:cl}, therefore we provide only references to the corresponding decompositions:
\begin{align*}
\mathbb{H}_0^{\mathrm{sc}}(\Gamma_3\Gamma_5)&\longrightarrow
\mbox{see Lemma \ref{33-lem-3} in Section \ref{33:sec:cl:2}},
\\
\mathbb{H}_0^{\mathrm{sc}}\big(\Gamma_3^4\big)&\longrightarrow
\mbox{see Lemma \ref{33-lem-1} in Section \ref{33:sec:cl:3}},
\\
\mathbb{H}_0^{\mathrm{sc}}(\Gamma_3^2\Gamma_4^{\phantom{1}})&\longrightarrow
\mbox{see Lemma \ref{33-lem-13} in Section \ref{33:sec:sum1}},
\\
\mathbb{H}_0^{\mathrm{sc}}(\Gamma_4^2)&\longrightarrow
\mbox{see Lemma \ref{33-lem-14} in Section \ref{33:sec:sum2}},
\\
\mathbb{H}_0^{\mathrm{sc}}(\Gamma_6)&\longrightarrow
\mbox{see Lemma \ref{33-lem-10} in Section \ref{33:sec:cl:6}}.
\end{align*}
Since the statements are proved for the more general case, the free parameters $\sigma_1$--$\sigma_5$ should be chosen such that equalities hold $\ln(\Lambda/\sigma_i)=L_1$ for all $i\in\{1,\ldots,5\}$.
Additionally, we use the decomposition for the two-loop diagram:
\begin{equation}\label{33-m-6}
\mathbb{H}_0^{\mathrm{sc}}\big(\Gamma_4\big)\longrightarrow
\mbox{see Lemma \ref{33-lem-15} in Section \ref{33:sec:sum2}}.
\end{equation}
It has been studied in the works \cite{Ivanov-Akac,Iv-244}. Note that the finite parts were preserved in the last relation, since the diagram is included in $W_1[B,\Lambda]$ and, thus, multiplied by the renormalization coefficient $a_0$. Further, after additional summation, we come to the relation
\begin{align}\nonumber
W_2[B,\Lambda]\stackrel{\mathrm{s.p.}}{=}&\,
	\bigg(\frac{\Lambda^2\alpha_1}{4(6!)\pi}-\frac{\Lambda^2\alpha_6}{4!4!}\bigg)\mathbb{H}_0^{\mathrm{sc}}(\widetilde{\Gamma}_4)-\frac{c_2^2}{96\pi}\Lambda^2J_5[B]\tau_5-\frac{L_1}{12(3!5!)\pi}\mathbb{H}_0^{\mathrm{sc}}(\Gamma_3\hat{\Gamma}_3)
	\\\label{33-i-4}&-
	\frac{L}{8(4!)^2\pi^2}\Big(c_2^2J_1[B]\tau_1+c_2J_2[B]\tau_2+J_3[B]\tau_3+c_2^2J_4[B]\tau_4\Big)
	\\\nonumber
	&-\frac{\mathbb{H}_0^{\mathrm{sc}}(\Gamma_3\Gamma_{r,3})}{4(3!5!)}
	-\frac{\mathbb{H}_0^{\mathrm{sc}}(\Gamma_{r,4})}{2(6!)}+\hat{\upsilon}_1S[B]+\hat{\upsilon}_2,
\end{align}
where $\hat{\upsilon}_1$ and $\hat{\upsilon}_2$ are singular coefficients with respect to the regularizing parameter $\Lambda$. They are independent of the background field.
Next, we note that the $J$-functionals, the definitions for which are in Section \ref{33:sec:dia}, can be replaced via auxiliary vertices as follows
\begin{equation}\label{33-i-3}
J_5[B]\longrightarrow\frac{1}{2}\mathbb{H}_0^{\mathrm{sc}}(\mathrm{R}_0),\,\,\,
J_i[B]\longrightarrow\frac{1}{2}\mathbb{H}_0^{\mathrm{sc}}(\mathrm{R}_i),
\end{equation}
where $i=1,2,3,4$. However, such a substitution can only lead to a shift in the coefficients $\hat{\upsilon}_1$ and $\hat{\upsilon}_2$. Therefore, by choosing the vertices $\Gamma_{r,3}$ and $\Gamma_{r,4}$ as suggested in the statement, all singular nonlocal components are reduced. It is clear that the part proportional to the classical action $S[B]$ can be removed by a suitable choice of the coefficient $a_2$ from \eqref{33-p-32}, see also \eqref{33-p-31}. Let us clarify the structure of the value $a_2$.

Note that during the calculation of the asymptotics in Section \ref{33:sec:cl}, it was shown that all logarithmic singularities are constructed using monomials $L_1^k$, where $k>0$. Thus, the main logarithmic contribution can be calculated by differentiating by the auxiliary parameter $\sigma$. Indeed, we calculate the derivative in equation \eqref{33-p-31}, and then use equation \eqref{33-p-33} for $i=1$, the absence of dependence of the vertices and the deformed Green's function on the parameter $\sigma$, as well as the fact that such differentiation does not lead to towards new singularities. Note that only the last five terms in \eqref{33-p-31} can make a non-zero contribution. For convenience, the operator $-\sigma\partial_{\sigma}$ we notate by a dot. Then we get the equality
\begin{align}\label{33-i-5}
-\sigma\frac{\mathrm{d}}{\mathrm{d}\sigma}&W_2[B,\Lambda]=-\sigma\frac{\mathrm{d}}{\mathrm{d}\sigma}
\bigg(-\frac{\mathbb{H}_0^{\mathrm{sc}}(\Gamma_3\Gamma_{r,3})}{4(3!5!)}
-\frac{\mathbb{H}_0^{\mathrm{sc}}(\Gamma_{r,4})}{2(6!)}+a_0W_1[B,\Lambda]-\frac{a_2-a_0a_1}{4}S[B]+\hat{\varkappa}_2\bigg)\\\nonumber&
=-\frac{\mathbb{H}_0^{\mathrm{sc}}(\Gamma_3\hat{\Gamma}_3)}{12(3!5!)\pi}+
\frac{1}{16(4!)^2\pi^2}\mathbb{H}_0^{\mathrm{sc}}\big(c_2^2\mathrm{R}_1\tau_1+c_2\mathrm{R}_2\tau_2+\mathrm{R}_3\tau_3+c_2^2\mathrm{R}_4\tau_4\Big)
-\frac{\dot{a}_2-\dot{a}_0a_1}{4}S[B]+\dot{\hat{\varkappa}}_2.
\end{align}
Next, note that the diagram $\mathbb{H}_0^{\mathrm{sc}}(\Gamma_3\hat{\Gamma}_3)$ is finite, while the remaining diagrams can be calculated explicitly
\begin{equation*}
\mathbb{H}_0^{\mathrm{sc}}(\mathrm{R}_1)\stackrel{\mathrm{s.p.}}{=}-\frac{L_1c_2}{\pi}S[B]
,\,\,\,
\mathbb{H}_0^{\mathrm{sc}}(\mathrm{R}_2)\stackrel{\mathrm{s.p.}}{=}\frac{L_1c_2^2}{\pi}S[B]
,\,\,\,
\mathbb{H}_0^{\mathrm{sc}}(\mathrm{R}_3)\stackrel{\mathrm{s.p.}}{=}-\frac{L_1c_2^3}{\pi}S[B]
,\,\,\,
\mathbb{H}_0^{\mathrm{sc}}(\mathrm{R}_4)\stackrel{\mathrm{s.p.}}{=}-\frac{L_1c_2}{2\pi}S[B].
\end{equation*}
Therefore, the requirement $\dot{W}_2[B,\Lambda]\stackrel{\mathrm{s.p.}}{=}0$ leads to a first-order differential equation
\begin{equation*}
\dot{a}_2\stackrel{\mathrm{s.p.}}{=}-\frac{c_2^3L_1}{8(4!)^2\pi^3}\big(2\tau_1-2\tau_2+2\tau_3+\tau_4\Big)
+a_0\dot{a}_1,
\end{equation*}
which can be integrated explicitly and leads to the answer stated in the theorem. The proof is finished.\\

\noindent\textbf{Proof of Theorem \ref{33-tt}.} Let us again consider the representation from \eqref{33-i-4} for the singular part of the third loop correction $W_2[B,\Lambda]$. However, instead of replacing \eqref{33-i-3}, let us use the calculations obtained in Section \ref{33:sec:kva} and perform the substitutions taking into account the fairness of the following transitions
\begin{align*}
	c_2^2\tau_1\upsilon_1\mathbb{H}_0^{\mathrm{sc}}(\mathrm{R}_1)
	+c_2\pi L_1^{-1}\vartheta_1\mathbb{H}_0^{\mathrm{sc}}(\mathrm{\hat{V}}_1)\longrightarrow&\,
	\tau_1\Big(2c_2^2\tau_1\upsilon_1J_1-2c_2^2\vartheta_1J_1\Big),
	\\
	c_2\tau_2\upsilon_2\mathbb{H}_0^{\mathrm{sc}}(\mathrm{R}_2)
	+c_2\pi L_1^{-1}\vartheta_2\mathbb{H}_0^{\mathrm{sc}}(\mathrm{\hat{V}}_2)\longrightarrow&\,
	\tau_2\Big(2c_2\tau_2\upsilon_2J_2+2c_2\vartheta_2(c_2J_1-J_2)\Big),
	\\
	\tau_3\upsilon_3\mathbb{H}_0^{\mathrm{sc}}(\mathrm{R}_3)
	+\pi L_1^{-1}\vartheta_3\mathbb{H}_0^{\mathrm{sc}}(\mathrm{\hat{V}}_3)\longrightarrow&\,
	\tau_3\Big(2\tau_3\upsilon_3J_3+2\vartheta_3(-c_2^2J_1-J_3)\Big),
	\\
	c_2^2\tau_4\upsilon_4\mathbb{H}_0^{\mathrm{sc}}(\mathrm{R}_4)
	+c_2\pi L_1^{-1}\vartheta_4\mathbb{H}_0^{\mathrm{sc}}(\mathrm{\hat{V}}_4)\longrightarrow&\,
	\tau_4\Big(2c_2^2\tau_4\upsilon_4J_4+c_2\vartheta_4(2c_2J_4-J_2)\Big),
\end{align*}
in which the terms proportional to the classical action are omitted.
In this case, the auxiliary coefficients must obey the conditions from Theorem \ref{33-tt}, which is easy to verify after adding up all the relations and comparing the linear combinations from \eqref{33-i-2} and \eqref{33-i-4}. Thus, the possibility of using a vertex of the form \eqref{33-i-2} has been proved.

Turning to obtaining the renormalization coefficient, we use the differentiation by the parameter $\sigma$ and the corresponding formula from \eqref{33-i-5}. Note that the auxiliary quasi-local vertices are independent of the parameter $\sigma$. Therefore, after differentiation, we get the second line from \eqref{33-i-5} with the replacement $\tau_i\to\tau_i\upsilon_i$. Therefore, the third coefficient $a_2$ is obtained by using the same substitution. Next, it is easy to make sure that the diagrams $L_1\mathbb{H}_0^{\mathrm{sc}}(\mathrm{R}_i)$ for $i\in\{1,2,3,4\}$ does not lead to terms proportional to the classical action and the first degree of the logarithm $L_1$, which implies the absence of dependence on the set $\{\upsilon_i\}_{i=1}^4$. The theorem has been proved.

\subsection{Special cases}
\label{33:sec:spec}
Let us look at some explicit examples. To do this, let us pay attention to the constants $\tau_i$ from Theorem \ref{33-t}. They are linear combinations of functionals depending on the deformed free Green's function, and in particular contain $\theta$-functionals. Let us study them in more detail and start with $\theta_1$, see \eqref{33-p-38},
\begin{equation}\label{33-p-38-11}
	\frac{\theta_1}{2\pi}=
	\int_{\mathbb{R}^{2}}\mathrm{d}^2x\,
	\Big(A_0(x)G^{1,\mathbf{f}}(x)\Big)
	\Big(G^{1,\mathbf{f}}(x)-
	G^{1,\mathbf{f}}(0)\bigg).
\end{equation}
Next, we note that the regularization from Section \ref{33:sec:ps} leads to the deformation of the Green's function, which in the momentum representation is reduced to multiplying the function $|k|^{-2}$ by the regularizing factor $\rho(|k|/\Lambda)$. The explicit form of such a factor is not important yet, we only assume that the function $\rho(|k|/\Lambda)/|k|^2$ is bounded and integrable in the exterior of a ball of any radius centered at the origin, and also has the properties
\begin{equation*}
\rho(0)=1,\,\,\,0\leqslant\rho(s)\leqslant1\,\,\,\mbox{for almost all}\,\,\,s>0.
\end{equation*}
Let us switch to \eqref{33-p-38-11} in the momentum representation. Despite the fact that $1/|k|^2$ is a nonintegrable function in two-dimensional space, the difference is well defined
\begin{equation*}
\theta_1=\frac{1}{2\pi}
\int_{\mathbb{R}^{2}}\mathrm{d}^2k\,|k|^{-2}
\Big(\rho^2(|k|/\Lambda)-\rho(|k|/\Lambda)\Big).
\end{equation*}
It instantly follows from the properties of the function $\rho(\cdot)$ that $\theta_1\leqslant0$. In this case, the equality\footnote{This property will be useful later when discussing the standard cutoff in the momentum representation. Thus, the relation is true for a wider class of deformations.} $\theta_1=0$ can be fulfilled if and only if the mapping $\rho(|\cdot|)$ takes values from $\{0,1\}$ almost everywhere. The same reasoning holds for all $\theta_k$, $k>0$.
\begin{lemma}\label{33-ll-1}
Taking into account the type of regularization, the relations $\theta_k\leqslant0$ are correct for $k>0$. The constraint $a_1\geqslant0$ for the two-loop renormalization coefficient is also valid.
\end{lemma}

Note that the deformation of the Green's function, introduced in Section \ref{33:sec:ps}, in the main decomposition order near the diagonal was reduced to the deformation of the free fundamental solution in some neighborhood of zero (in a ball of radius $1/\Lambda$). This was important for using the averaging operator in the formulation of regularization. In particular, this led to the fact that the support of the function $A_0(x)G^{\Lambda,\mathbf{f}}(x)$ was contained in the ball $\mathrm{B}_{1/\Lambda}$. However, technically, this restriction was not applied in all functionals (diagrams).

Let us try to compare the deformation used with the standard cutoff in the momentum representation. In other words, let us proceed formally to the limiting case. Namely, choose \footnote{In the general case, we can consider a broader class of sets of finite measure. In particular, we can choose $|k|<a\in\mathbb{R}_+$, which eventually reduces to scaling the regularizing parameter.} $\rho(|k|)=\chi(|k|<1)$, where $\chi(\cdot)$ is the characteristic function of the specified set in $\mathbb{R}^2$. Note that in this situation, the function
\begin{equation*}
P(x-y)=\frac{1}{4\pi^2}\int_{\mathbb{R}^2}\mathrm{d}^2k\,e^{ik_\mu(x-y)^{\mu}}\chi(|k|<1)
\end{equation*}
is the kernel of an integral representation for a projector. Thus, the relation is correct
\begin{equation*}
\int_{\mathbb{R}^2}\mathrm{d}^2z\,P(x-z)P(z-y)=P(x-y),
\end{equation*}
which should be understood in the sense of generalized functions, see \cite{Gelfand-1964,Vladimirov-2002}. Therefore, as a deformed combination of $A_0(x)G_{\mathrm{def}}^\Lambda(x)$, we can choose the function $P(x)$, that is
\begin{equation*}
A_0(x)G_{\mathrm{def}}(x)=P(x)=\frac{1}{4\pi^2}\int_{\mathbb{R}^2}\mathrm{d}^2k\,e^{ik_\mu x^{\mu}}\chi(|k|<\Lambda)=\frac{\Lambda J_1(|x|\Lambda)}{2\pi|x|},
\end{equation*}
where in the last equation the definition for the Bessel function was used. This choice instantly leads to the equality $\theta_k=0$ for all $k>0$. And this fact is a positive feature of the standard cutoff in the momentum representation. It minimizes the value of $-\theta_i$. 

However, when trying to calculate, for example, the functional from \eqref{33-9-24-52}, a significant problem appears. Indeed, direct substitution leads to a divergent integral
\begin{equation*}
\alpha_8\Big|_{G^{\Lambda,\mathbf{f}}\to G_{\mathrm{def}}^\Lambda}=\lim_{N\to+\infty}\pi\int_{\mathrm{B}_N}\mathrm{d}^2x\,
\Big(A_0(x)G_{\mathrm{def}}^\Lambda(x)\Big)^2|x|^2=\lim_{N\to+\infty}\frac{1}{2}\int_0^N\mathrm{d}s\,sJ_1^2(s)=+\infty.
\end{equation*}
To understand what the number $\alpha_8$ should be equal to and which singularity it corresponds to, it is necessary to return to the process of obtaining it. When calculating, the functional was determined from the equality
\begin{equation*}
\frac{L\alpha_8}{2\pi}\stackrel{\mathrm{s.p.}}{=}\pi\int_{\mathrm{B}_{1/\sigma}}\mathrm{d}^2x\,
\Big(A_0(x)G^{\Lambda,\mathbf{f}}(x)\Big)^2G^{\Lambda,\mathbf{f}}(x)|x|^2\equiv\mathrm{Y}(G^{\Lambda,\mathbf{f}}),
\end{equation*}
as a coefficient about the logarithmic singularity. Thus, consider $\mathrm{Y}(G_{\mathrm{def}}^\Lambda)$.
For these purposes, we calculate the explicit form of the function $G_{\mathrm{def}}^\Lambda(x)$. Representing the Laplace operator in polar coordinates, we obtain
\begin{equation*}
G_{\mathrm{def}}^\Lambda(x)=\int_{1/\sigma}^{|x|}\mathrm{d}s\,\frac{J_{0}(s\Lambda)-1}{2\pi s}+\mathcal{O}\big(1/\Lambda\big),
\end{equation*}
where the integration used the fact that removing the regularization $\Lambda\to+\infty$ should result in the function $-\ln(|x|\sigma)/2\pi$. The correction term does not depend on the variable $x$. Therefore, after scaling, we have
\begin{equation*}
\mathrm{Y}(G_{\mathrm{def}}^\Lambda)=\int_{0}^{\Lambda/\sigma}\mathrm{d}t\,t
J_1^2(t)\int^{\Lambda/\sigma}_{t}\mathrm{d}s\,\frac{1-J_{0}(s)}{4\pi s}\\
=
\int^{\Lambda/\sigma}_{0}\mathrm{d}s\,\frac{1-J_{0}(s)}{4\pi s}
\int_{0}^{s}\mathrm{d}t\,t
J_1^2(t).
\end{equation*}
Let us differentiate and consider the asymptotic expansion for large $\Lambda$
\begin{align*}
\frac{\mathrm{d}}{\mathrm{d}\Lambda}\mathrm{Y}(G_{\mathrm{def}}^\Lambda)=&\,
\frac{1-J_{0}(\Lambda/\sigma)}{4\pi\Lambda}
\int_{0}^{\Lambda/\sigma}\mathrm{d}t\,t
J_1^2(t)\\=&\,\frac{1}{4\pi^2\sigma}
+\frac{\cos(2\Lambda/\sigma)}{8\pi^2\Lambda}
+\mathcal{O}\big(\Lambda^{-3/2}\big).
\end{align*}
Thus, the function $\mathrm{Y}(G_{\mathrm{def}}^\Lambda)$ has a power-law behavior at high momentum and, therefore, leads to significant differences between the two types of cutoffs. Of course, singular power contributions should not depend on the auxiliary parameter $\sigma$, and therefore such terms should cancel each other out with similar contributions from other functionals. However, further study of the cutoff in the momentum representation is beyond the scope of this work.

\section{Additional definitions}
\label{33:sec:vsvs}

\subsection{Diagram technique}
\label{33:sec:int:1}

Most of the calculations and results contain cumbersome formulas, which makes it convenient to use special-type diagram techniques. Speaking of purely group quantities, two main objects can be distinguished: the Kronecker symbol $\delta^{ab}$ and the totally antisymmetric structure constant $f^{abc}$. The following elements of diagram technique correspond to them
\begin{equation}
	\delta^{ab}=a\ 
	{
		\centering
		\adjincludegraphics[width = 0.7 cm, valign=c]{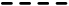}
	}
	\ b,
	\qquad
	f^{abc}=
	\hspace{10pt}
	\mathstrut
	\put(-7, -15){$a$}
	\put(12, 18){$b$}
	\put(30, -15){$c$}
	\begin{minipage}[h]{1.2 cm}
		{
			\centering
			\adjincludegraphics[width = 1 cm, valign=c]{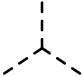}
		}
	\end{minipage}
	\hspace{10pt}.
	\label{33-eq:group_symbols}
\end{equation}
In the second case, the ordering of the indices is important: the sequence $a\to b\to c$ is organized clockwise. Note that well-known identities such as convolutions with two and three indices, as well as the Jacobi identity for structure constants, see \eqref{33-p-2}, can be represented by diagrammatic equalities
\begin{equation}
	{
		\centering
		\adjincludegraphics[width = 1.5 cm, valign=c]{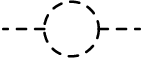}
	}
	=-c_2\
	{
		\centering
		\adjincludegraphics[width = 0.7 cm, valign=c]{fig/Line_1.eps}
	}
	,\qquad
	{
		\centering
		\adjincludegraphics[width = 1 cm, valign=c]{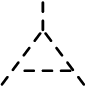}
	}
	=-\frac{c_2}{2}
	{
		\centering
		\adjincludegraphics[width = 1 cm, valign=c]{fig/Star.eps}
	},
	\label{33-eq:tr_S}
\end{equation}
\begin{equation}
	{
		\centering
		\adjincludegraphics[width = 1.2 cm, valign=c]{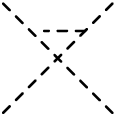}
	}
	=
	{
		\centering
		\adjincludegraphics[width = 1.2 cm, valign=c]{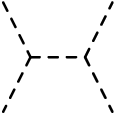}
	}
	-
	{
		\centering
		\adjincludegraphics[width = 1.2 cm, valign=c]{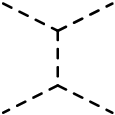}
	}\,\,.
	\label{33-eq:Jacobi}
\end{equation}
The next type of values are objects that, in addition to group indices, also contain spatial arguments. The main ones are the deformed Green's function $G^{ab}_{\Lambda}(x, y)$, as well as its nonlocal part $PS^{ab}_{\Lambda}(x, y)$. They are notated as follows
\begin{equation}
	G_{\Lambda}(x, y)=
	x\ 
	{
		\centering
		\adjincludegraphics[width = 1 cm, valign=c]{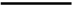}
	}
	\ y,\qquad PS_{\Lambda}(x, y)
	=x\ 
	{
		\centering
		\adjincludegraphics[width = 1 cm, valign=c]{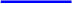}
	}
	\ y,
	\label{33-eq:G_symbols}
\end{equation}
which also implies the presence of group indexes. To avoid unnecessary cumbersomeness, instead of the arguments $x$ and $y$, we use dots (small circles) of the appropriate color, that is
\begin{equation}
	G_{\Lambda}({\color{red}\bullet}, \bullet)=
	{
		\centering
		\adjincludegraphics[width = 1 cm, valign=c]{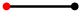}
	},\qquad PS_{\Lambda}({\color{red}\bullet}, \bullet)=
	{
		\centering
		\adjincludegraphics[width = 1 cm, valign=c]{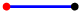}
	}.
	\label{33-eq:G_symbols_color}
\end{equation}
The objects above may have additional structures in the form of the derivative $\partial_{x^\mu}$, the background field $B_\mu(x)$, and the covariant derivative $D_\mu(x)$. Let us relate these symbols according to the following table
\begin{align*}
	\partial_{x^\mu}&\to\,\,\mid,\\
	B_\mu(x)&\to\,\,>,\\
	D_\mu(x)&\to\,\bullet.
\end{align*}
Given the fact that the last operators in calculations occur in pairs that depend on one variable and are summed by a free index, such objects in a pair will have the same color corresponding to the selected variable and different from the colors of the other pairs. For example, the representations are equivalent
\begin{align}
	&\partial_{x_\mu}\overrightarrow{D}_{\mu}(x) G_{\Lambda}(x, y)\longleftrightarrow
	{
		\centering
		\adjincludegraphics[width = 1 cm, valign=c]{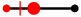}
	},\\&\partial_{x^\mu} PS_{\Lambda}(x, y) \overleftarrow{D}_{\nu}(y)\longleftrightarrow
	{
		\centering
		\adjincludegraphics[width = 1 cm, valign=c]{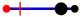}
	},\\&\big( \partial_{x^\mu} PS_{\Lambda}(x, y) \big)_{y = x} B_{\mu}(x)\longleftrightarrow
	{
		\centering
		\adjincludegraphics[width = 1 cm, valign=c]{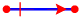}
	}.
\end{align}
At the same time, if the point indicating the argument is internal to the diagram (that is, inside), then it also includes summation by indices and corresponding integration over all possible values of the argument. So, for example, we have
\begin{equation*}
	\int_{\mathbb{R}^2}\mathrm{d}^2x\,\big( \partial_{x^\mu} PS_{\Lambda}^{ab}(x, y) \big)_{y = x} B_{\mu}^{ba}(x)=
	{
		\centering
		\adjincludegraphics[width = 0.8 cm, valign=c]{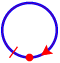}
	}.
\end{equation*}
The next type of objects are vertices. Taking into account the definition from \eqref{33-p-10}, we present them without additional explanations:
\begin{align}
	\label{33-g-3}
	&\Gamma_3=
	f^{abc}
	\int_{\mathbb{R}^2}
	\mathrm{d}^2x\,
	(\partial_\mu \phi^a)\phi^b
	B^{cd}_\mu \phi^d
	=
	{
		\centering
		\adjincludegraphics[width = 1.2 cm, valign=c]{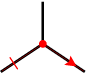}
	},
	\\\label{33-g-4}&\Gamma_4=-f^{abc} f^{cde} \int_{\mathbb{R}^2}
	\mathrm{d}^2x\,(\partial_\mu \phi^a) \phi^b \phi^d D_\mu^{eg} \phi^g
	=-{
		\centering
		\adjincludegraphics[width = 1.8 cm, valign=c]{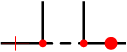}
	},
	\\\label{33-g-5}&\Gamma_5=f^{abc} f^{cde} f^{egh}\int_{\mathbb{R}^2}\mathrm{d}^2x\,
	(\partial_\mu \phi^a) \phi^b \phi^d \phi^g B_\mu^{hp} \phi^p=
	{
		\centering
		\adjincludegraphics[width = 2.4 cm, valign=c]{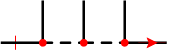}
	},
	\\\label{33-g-6}&\Gamma_6=-f^{abc} f^{cde} f^{egh} f^{hpq}
	\int_{\mathbb{R}^2}\mathrm{d}^2x\,
	(\partial_\mu \phi^a) \phi^b \phi^d \phi^g \phi^pD_\mu^{qr} \phi^r
	=-{
		\centering
		\adjincludegraphics[width = 3 cm, valign=c]{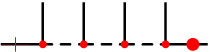}
	}.
\end{align}
Some auxiliary vertices may appear in the proofs as needed. Note that similar designations were proposed in \cite{13} and used in \cite{Ivanov-Kharuk-2020,Ivanov-Kharuk-20222,Ivanov-Akac}.
\subsection{Functionals}
\label{33:sec:dia}
Taking into account the definitions for the elements of the diagrammatic technique, we present a set of auxiliary functionals. However, they can be divided into two types. The first type contains the nonlocal component $PS_\Lambda(x,y)$ squared, as well as derivatives and background fields. Due to their bulkiness, they are presented only in the diagrammatic form:
\begin{equation}\label{33-s-8}
	\mathrm{A}_1=
	{\centering\adjincludegraphics[width = 0.6 cm, valign=c]{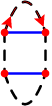}}
	,\,\,\,\mathrm{A}_2=
	{\centering\adjincludegraphics[width = 0.6 cm, valign=c]{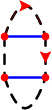}}
	,\,\,\,\mathrm{A}_3=
	{\centering\adjincludegraphics[width = 0.6 cm, valign=c]{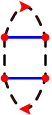}}
	,\,\,\,\mathrm{A}_4=
	{\centering\adjincludegraphics[width = 0.6 cm, valign=c]{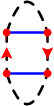}},
\end{equation}
\begin{equation}\label{33-d-107}
	\mathrm{A}_5=
	{\centering\adjincludegraphics[width = 0.6 cm, valign=c]{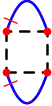}},\,\,\,
	\mathrm{A}_6=
	{\centering\adjincludegraphics[width = 0.6 cm, valign=c]{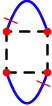}},\,\,\,
	\mathrm{A}_7=
	{\centering\adjincludegraphics[width = 0.6 cm, valign=c]{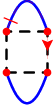}},\,\,\,
	\mathrm{A}_8=
	{\centering\adjincludegraphics[width = 0.6 cm, valign=c]{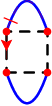}},
\end{equation}
\begin{equation}\label{33-s-9}
	\mathrm{B}_1=
	{\centering\adjincludegraphics[width = 0.8 cm, valign=c]{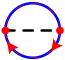}}
	,\,\,\,\mathrm{B}_2=
	{\centering\adjincludegraphics[width = 0.8 cm, valign=c]{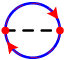}}
	,\,\,\,\mathrm{B}_3=
	{\centering\adjincludegraphics[width = 1.9 cm, valign=c]{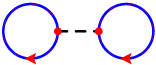}},
\end{equation}
\begin{equation}\label{33-d-98}
	\mathrm{B}_4={\centering\adjincludegraphics[width = 1.9 cm, valign=c]{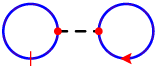}},\,\,\,
	\mathrm{B}_5={\centering\adjincludegraphics[width = 1.9 cm, valign=c]{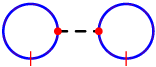}},
\end{equation}
\begin{equation}\label{33-d-106}
\mathrm{B}_6={\centering\adjincludegraphics[width = 0.8 cm, valign=c]{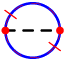}},\,\,\,
\mathrm{B}_7={\centering\adjincludegraphics[width = 0.8 cm, valign=c]{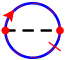}},\,\,\,
\mathrm{B}_8={\centering\adjincludegraphics[width = 0.8 cm, valign=c]{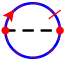}}.
\end{equation}
The second type of functionals can be represented analytically
\begin{equation}\label{33-d-28}
	J_1[B]=\int_{\mathbb{R}^2}\mathrm{d}^2x\,PS_\Lambda^{ae}(x,x)B_\mu^{eh}(x)B_\mu^{ha}(x),
\end{equation}
\begin{equation}\label{33-d-29}
	J_2[B]=\int_{\mathbb{R}^2}\mathrm{d}^2x\,
	\Big(f^{abc}PS_\Lambda^{bd}(x,x)f^{cde}\Big)B_\mu^{eh}(x)B_\mu^{ha}(x),
\end{equation}
\begin{equation}\label{33-w-5}
	J_3[B]=\int_{\mathbb{R}^2}\mathrm{d}^2x\,
	\Big(f^{a_1ga_2}f^{a_1bc}PS_\Lambda^{bd}(x,x)f^{cde_1}f^{e_2ge_1}\Big)B_\mu^{e_2h}(x)B_\mu^{ha_2}(x),
\end{equation}
\begin{equation}\label{33-d-28-1}
	J_4[B]=\int_{\mathbb{R}^2}\mathrm{d}^2x\,B_\mu^{ea}(x)\Big(\partial_{x_\mu}PS_\Lambda^{ae}(x,y)\Big)\Big|_{y=x},
\end{equation}
\begin{equation}\label{33-d-28-3}
	J_5[B]=\int_{\mathbb{R}^2}\mathrm{d}^2x\,PS_\Lambda^{aa}(x,x),
\end{equation}
\begin{equation}\label{33-d-123}
	J_6[B]=\int_{\mathbb{R}^2}\mathrm{d}^2x\,
	\Big(f^{a_1ga_2}f^{a_1bc}PS_\Lambda^{bd}(x,x)f^{cde_1}f^{e_2ge_1}\Big)PS_\Lambda^{e_2a_2}(x,x),
\end{equation}
\begin{equation}\label{33-d-123-1}
	J_7[B]=\int_{\mathbb{R}^2}\mathrm{d}^2x\,
	\Big(f^{a_1bc}PS_\Lambda^{bd}(x,x)f^{cde_1}\Big)PS_\Lambda^{e_1a_1}(x,x).
\end{equation}

%$G(x)=-\ln(|x|\sigma)/2\pi$

\section{Green's function and its properties}
\label{33:sec:cl:1}

As noted above, an unregularized quadratic form is defined by the Laplace operator, which can be written out in local coordinates by equation \eqref{33-o-3}. It is known that the Green's function for the free Laplace operator is equal to $-\ln(|x|\mu)/2\pi$. Here $\mu$ is a formal parameter to make the argument dimensionless, which is fixed taking into account the boundary conditions.
Further, in the process of regularization, all fluctuations are smoothed, except those in quadratic form with the operator $A_0$. 
Thus, the final operator \eqref{33-o-3} in the quadratic form passes into the operator from equation \eqref{33-p-24}.
It is for such an operator that the perturbative decomposition for the Green's function is constructed.
Using the formal application of the second resolvent identity, the Green's function $G_{\mathrm{reg}}^{ab}(x,y)$ for the regularized operator from \eqref{33-p-24} can be written as a perturbative series over the background field.
Finally, the desired deformed Green's function $G_{\Lambda}^{ab}(x,y)$ from equality \eqref{33-p-25}, taking into account the rules of diagram technique, is obtained by applying averaging operators for both variables and, according to previously obtained results, can be written in the following form
\begin{equation}\label{33-o-5}
	G_{\Lambda}^{ab}(x,y)=
	\mathrm{O}(x)\mathrm{O}(y)\big(A^{-1}_{\Lambda}\big)^{ab}(x,y)
	=
	2G^{\Lambda,\mathbf{f}}(x-y;\mu)\delta^{ab}+2\sum_{k=1}^3\eta^{ab}_k(x,y),
\end{equation}
where 
\begin{equation}\label{33-o-55}
	-\frac{\ln(|x|\mu)}{2\pi}\longrightarrow
	G^{\Lambda,\mathbf{f}}(x;\mu)=\,\frac{\mathbf{f}\big(\Lambda,|x|^2\Lambda^2\big)}{4\pi}+
	\frac{1}{2\pi}
	\begin{cases} 
		~~\ln(\Lambda/\mu), & \mbox{for}\,\,\,|x|\leqslant1/\Lambda;\\
		-\ln(|x|\mu), &\mbox{for}\,\,\,|x|>1/\Lambda,
	\end{cases}
\end{equation}
\begin{equation}\label{33-o-6}
	\eta_1^{ab}(x,y)=-\int_{\mathbb{R}^2}\mathrm{d}^2x_1\,
	G^{\Lambda,\mathbf{f}}(x-x_1;\mu)V^{ab}(x_1)G^{\Lambda,\mathbf{f}}(x_1-y;\mu),
\end{equation}
\begin{equation}\label{33-o-7}
	\eta_2^{ab}(x,y)=\int_{\mathbb{R}^{2\times2}}\mathrm{d}^2x_1\mathrm{d}^2x_2\,
	G^{\Lambda,\mathbf{f}}(x-x_1;\mu)V^{ac}(x_1)G^{\Lambda,\mathbf{f}}(x_1-x_2;\mu)
	V^{cb}(x_2)G^{\Lambda,\mathbf{f}}(x_2-y;\mu),
\end{equation}
\begin{multline}
	\eta_3^{ab}(x,y)=\sum_{k=3}^{+\infty}(-1)^k
	\int_{\mathbb{R}^{2\times k}}\mathrm{d}^2x_1\ldots\mathrm{d}^2x_k\,
	G^{\Lambda,\mathbf{f}}(x-x_1;\mu)\times\\\times V^{ac_1}(x_1)G^{\Lambda,\mathbf{f}}(x_1-x_2;\mu)\cdot\ldots\cdot
	V^{c_{k-1}b}(x_k)G^{\Lambda,\mathbf{f}}(x_k-y;\mu).
\end{multline}
Next, we assume that the kernel of the averaging operator (deformation) is chosen such that the following equality holds
\begin{equation}\label{33-r-31}
	\partial_{x^\mu}G^{\Lambda,\mathbf{f}}(x)\big|_{x=0}=0.
\end{equation}
It is worth noting that, despite the obvious dependence, the representation from \eqref{33-o-5} is very cumbersome, so it is more convenient to use the modified decomposition in real calculations. To do this, we introduce a dimensional fixed parameter $\sigma>0$ and fix the notation $G^{\Lambda,\mathbf{f}}(x)$ for an auxiliary function $G^{\Lambda,\mathbf{f}}(x;\mu)$, then, see \cite{29,vas1,vas2,33,30-1-1}, we can write
\begin{multline}\label{33-r-1}
	G_{\Lambda}^{ab}(x,y)=2G^{\Lambda,\mathbf{f}}(x-y)\delta^{ab}+
	G_{1}^\mu(x-y)\Big(B^{ab}_\mu(x)+B^{ab}_\mu(y)\Big)+\\+
	G_{11}^{\nu\mu}(x-y)\Big(\partial_{x^\nu}B^{ab}_\mu(x)+\partial_{y^\nu}B^{ab}_\mu(y)\Big)+\\+
	G_{2}^{\nu\mu}(x-y)\Big(B^{ac}_\mu(x)B^{cb}_\mu(x)+B^{ac}_\mu(y)B^{cb}_\mu(y)\Big)
	+2PS^{ab}_\Lambda(x,y),
\end{multline}
when $|x-y|<1/\sigma$, and also\footnote{It is assumed that when dividing the support, a unit partition with a suitable smooth function is used, so that no singular functions appear when differentiating.} $G_{\Lambda}^{ab}(x,y)=2PS^{ab}_\Lambda(x,y)$ for $|x-y|\geqslant1/\sigma$. In this case, the function $PS_\Lambda$ has two finite derivatives before and after removing the regularization. We can assume that $PS_\Lambda$ is determined by the latter equality, while the deformed function $G_{\Lambda}^{ab}(x,y)$ is independent of $\sigma$.
Note that the functions $\{G_{1}^\mu, G_{11}^{\nu\mu}, G_{2}^{\nu\mu}\}$ have an arbitrary definition, which consists in the possibility of adding a smooth function. For convenience, we fix and write out their explicit form, and also note a number of properties that they possess. We have
\begin{equation}\label{33-r-21}
	G_{1}^\mu(x)=-
	\int_{\mathrm{B}_{1/\hat{\sigma}}}\mathrm{d}^2y\,
	G^{\Lambda,\mathbf{f}}(x-y)
	\partial_{y_{\mu}}G^{\Lambda,\mathbf{f}}(y),
\end{equation}
\begin{equation}\label{33-r-22}
	G_{11}^{\nu\mu}(x)=-
	\int_{\mathrm{B}_{1/\sigma}}\mathrm{d}^2y\,
	G^{\Lambda,\mathbf{f}}(x-y)
	\big(\delta_{\nu\mu}/2+x^\nu\partial_{y_{\mu}}\big)G^{\Lambda,\mathbf{f}}(y),
\end{equation}
\begin{equation}\label{33-r-23}
	G_{2}^{\nu\mu}(x)=
	\int_{\mathrm{B}_{1/\sigma}}\mathrm{d}^2y\int_{\mathrm{B}_{1/\sigma}}\mathrm{d}^2z\,
	G^{\Lambda,\mathbf{f}}(x-y)
	G^{\Lambda,\mathbf{f}}(y-z)
	\partial_{z_{\nu}}\partial_{z_{\mu}}G^{\Lambda,\mathbf{f}}(z)+\frac{\delta_{\mu\nu}}{16\pi\sigma^2},
\end{equation}
where the parameter $\hat{\sigma}=\hat{\sigma}(\sigma)$ selected in such a way that the equality is satisfied
\begin{equation}\label{33-r-29}
	\partial_{x^\mu}G_{1}^\mu(x)\Big|_{x=0}=
	\int_{\mathrm{B}_{1/\hat{\sigma}}}\mathrm{d}^2y\,
	\Big(\partial_{y^{\mu}}G^{\Lambda,\mathbf{f}}(y)\Big)\Big(
	\partial_{y_{\mu}}G^{\Lambda,\mathbf{f}}(y)\Big)=
	G^{\Lambda,\mathbf{f}}(0)=\frac{2L+\mathbf{f}(0)}{4\pi}\equiv\frac{L_1}{2\pi}
	,
\end{equation}
where $L=\ln(\Lambda/\sigma)$.
Also, in formula \eqref{33-r-23}, the subtracted number was selected based on the following integral representation
\begin{multline}\label{33-r-30}
	\int_{\mathrm{B}_{1/\sigma}}\mathrm{d}^2y\int_{\mathrm{B}_{1/\sigma}}\mathrm{d}^2z\,
	G(y)
	G(y-z)
	\partial_{y_{\nu}}\partial_{y_{\mu}}G(y)=\\=-\frac{\delta_{\mu\nu}}{2}
	\int_{\mathrm{B}_{1/\sigma}}\mathrm{d}^2y\,
	G^2(y)=-\frac{\delta_{\mu\nu}}{4\pi}
	\int_0^{1/\sigma}\mathrm{d}s\,s\big(\ln(s\sigma)\big)^2=-\frac{\delta_{\mu\nu}}{16\pi\sigma^2}
	.
\end{multline}
Additionally, we introduce a number of auxiliary objects that appear in the calculations quite frequently. They have the form
\begin{equation}\label{33-r-29-1}
	\frac{L_2}{2\pi}=
	\int_{\mathrm{B}_{1/\Lambda}}\mathrm{d}^2y\,
	\Big(A_0(y)G^{\Lambda,\mathbf{f}}(y)\Big)G^{\Lambda,\mathbf{f}}(y)=
	\frac{L_1-\ln(\sigma/\hat{\sigma})}{2\pi},
\end{equation}
and for $k>1$
\begin{equation*}\label{33-r-29-2}
	\frac{L_{k+1}}{2\pi}=
	\int_{\mathbb{R}^{2\times k}}\mathrm{d}^2y_1\ldots\mathrm{d}^2y_k\,
	\Big(A_0(y)G^{\Lambda,\mathbf{f}}(y_1)\Big)
	\Big(A_0(y)G^{\Lambda,\mathbf{f}}(y_1-y_2)\Big)
	\ldots
	\Big(A_0(y)G^{\Lambda,\mathbf{f}}(y_{k-1}-y_k)\Big)
	G^{\Lambda,\mathbf{f}}(y_k).
\end{equation*}
Given the definition of numbers \eqref{33-p-38}, \eqref{33-p-38-1}, and \eqref{33-p-38-2}, it is possible to check the identities
\begin{equation}\label{33-bb-1}
	L_{k+1}=L_k+\theta_k\,\,\,\mbox{for}\,\,\,k\geqslant1,\,\,\,
	\ln(\sigma/\hat{\sigma})=-\theta_1,
\end{equation}
as well as
\begin{equation}\label{33-bb-2}
	\int_{\mathrm{B}_{1/\Lambda}}\mathrm{d}^2y\,
	\Big(A_0(y)G^{\Lambda,\mathbf{f}}(y)\Big)\partial_{y^\mu}G^{\mu}_1(y)=
	\frac{L_3-\theta_1}{2\pi}=\frac{L_1+\theta_2}{2\pi}.
\end{equation}
We note several useful properties that follow from spherical symmetry and explicit calculation of integrals. In addition to the functionals from Section \ref{33:sec:osn}, we define the value\footnote{The notations for $\alpha_2$--$\alpha_4$ are omitted, as they were used in other works. In particular, in \cite{Iv-244} it was introduced $\alpha_3(\mathbf{f})=2\theta_2/3$.}
\begin{equation}\label{33-r-24-4}
	\alpha_5(\mathbf{f})=-\frac{3}{16\pi}\bigg(\int_0^1\mathrm{d}s\,\mathbf{f}(s)-1\bigg).
\end{equation}
\begin{lemma}\label{33-lem-7}
	Taking into account all of the above, the following relations are true
	\begin{equation}\label{33-r-24}
		G_{1}^\mu(0)=0,\,\,\,G_{11}^{\nu\mu}(0)=0,\,\,\,
		\partial_{x_\rho}G_{11}^{\nu\mu}(x-y)\big|_{y=x}=0,\,\,\,
		\partial_{x_\rho}G_{2}^{\nu\mu}(x-y)\big|_{y=x}=0,
	\end{equation}
	\begin{equation}\label{33-r-25}
		G_{2}^{\nu\mu}(0)=\frac{L\alpha_{5}(\mathbf{f})}{\Lambda^2}\delta_{\mu\nu}+\mathcal{O}\big(1/\Lambda^2\big),
	\end{equation}
	\begin{equation}\label{33-r-28}
		\partial_{x^\mu}G_{\Lambda}^{ab}(x,y)\big|_{y=x}=
		\frac{L_1}{2\pi}B_\mu^{ab}(x)+2\partial_{x^\mu}PS_{\Lambda}^{ab}(x,y)\big|_{y=x}
		+\mathcal{O}\big(L/\Lambda^2\big),
	\end{equation}
	\begin{equation}\label{33-r-32}
		D_{\mu}^{ac}(x)G_{\Lambda}^{cb}(x,y)\big|_{y=x}=
		-\frac{L_1}{2\pi}B_\mu^{ab}(x)+2D_{\mu}^{ac}(x)PS_{\Lambda}^{cb}(x,y)\big|_{y=x}
		+\mathcal{O}\big(L/\Lambda^2\big),
	\end{equation}
	\begin{equation}\label{33-r-26}
		\partial_{x^\mu}G_{\Lambda}^{ab}(x,y)\big|_{y=x}+
		\partial_{y^\mu}G_{\Lambda}^{ab}(x,y)\big|_{y=x}=
		2\partial_{x^\mu}PS_{\Lambda}^{ab}(x,y)\big|_{y=x}+
		2\partial_{y^\mu}PS_{\Lambda}^{ab}(x,y)\big|_{y=x}+\mathcal{O}\big(L/\Lambda^2\big),
	\end{equation}
	\begin{equation}\label{33-r-27}
			D_{\mu}^{ac}(x)G_{\Lambda}^{cb}(x,y)\big|_{y=x}+
			D_{\mu}^{bc}(y)G_{\Lambda}^{ac}(x,y)\big|_{y=x}=
		2	D_{\mu}^{ac}(x)PS_{\Lambda}^{cb}(x,y)\big|_{y=x}+
		2	D_{\mu}^{bc}(y)PS_{\Lambda}^{ac}(x,y)\big|_{y=x}+\mathcal{O}\big(L/\Lambda^2\big).
	\end{equation}
\end{lemma}
\begin{proof} All relations except \eqref{33-r-25} are a consequence of the spherical symmetry of the densities and the explicit decomposition from \eqref{33-r-1}. The remaining equality  from \eqref{33-r-25} is presented in Section \ref{33:sec:appl:2}.
\end{proof}
\begin{lemma}\label{33-lem-8} Let us consider three auxiliary matrix-valued operators
	\begin{equation*}
		M_1^{ab}(y)=\int_{\mathbb{R}^2}\mathrm{d}^2x\,
		\Big(\partial_{y_\mu}G_\Lambda^{ad}(y,x)\Big)
		\Big(A^{de}(x)G_\Lambda^{eb}(x,y)\Big),
	\end{equation*}
	\begin{equation*}
		M_2^{ab}(y)=\int_{\mathbb{R}^2}\mathrm{d}^2x\,
		\Big(D_{\mu}^{ac}(y)G_\Lambda^{cd}(y,x)\Big)
		\Big(A^{de}(x)G_\Lambda^{eb}(x,y)\Big),
	\end{equation*}
	\begin{equation*}
		M_3^{ab}(y)=\int_{\mathbb{R}^2}\mathrm{d}^2x\,
		\Big(A^{ac}(y)G_\Lambda^{cd}(y,x)\Big)
		\Big(A^{de}(x)G_\Lambda^{eb}(x,y)\Big).
	\end{equation*}
	Then, taking into account all the above, the following relations are true
	\begin{equation*}
		M^{ab}_1(y)=\frac{L_1+2\theta_2}{2\pi}B_\mu^{ab}
		+\partial_{y_\mu}PS_\Lambda^{ad}(y,x)\big|_{x=y}+\mathcal{O}(L/\Lambda),
	\end{equation*}
	\begin{equation*}
		M_2^{ab}(y)=-\frac{L_1-2\theta_2+2\theta_1}{2\pi}B_\mu^{ab}
		+D_{\mu}^{ac}(y)PS_\Lambda^{ad}(y,x)\big|_{x=y}+\mathcal{O}(L/\Lambda),
	\end{equation*}
	\begin{equation*}
		\frac{1}{2}\Big(M_3^{ab}(y)+M_3^{ba}(y)\Big)=\Lambda^2\alpha_6(\mathbf{f})
		-\frac{\theta_3+\theta_2}{4\pi}B_\mu^{ac}(y)B_\mu^{cb}(y)
		+\mathcal{O}(L/\Lambda).
	\end{equation*}
\end{lemma}
\begin{proof}
	Decompositions follow from substitutions of \eqref{33-o-5} and \eqref{33-r-1} into integrals and additional integration by parts.
\end{proof}

\section{Calculation of asymptotic expansions}
\label{33:sec:cl}

\subsection{The case $\mathbb{H}_0^{\mathrm{sc}}(\Gamma_3\Gamma_5)$}
\label{33:sec:cl:2}
Let us calculate the contributions of interest in $\mathbb{H}_0^{\mathrm{sc}}(\Gamma_3\Gamma_5)$ containing singularities according to the regularizing parameter $\Lambda$. Note that all strongly connected diagrams for the product of the vertices $\Gamma_3$ and $\Gamma_5$ contain exactly one loop (Green's function on the diagonal). Therefore, the relation is true
\begin{equation}\label{33-s-10}
\mathbb{H}_0^{\mathrm{sc}}(\Gamma_3\Gamma_5)=\mathbb{H}_0^{\mathrm{sc}}(\Gamma_3\mathbb{H}_3^{\mathrm{c}}(\Gamma_5)).
\end{equation}
Without taking into account the group and operator elements available in \eqref{33-s-10}, the general view of the diagrams can be represented using the figure
\begin{equation}\label{33-s-31}
	{\centering\adjincludegraphics[width = 2 cm, valign=c]{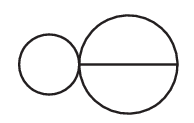}}.
\end{equation}
Next, we note that the diagram group from \eqref{33-s-10} can be divided into two parts
\begin{enumerate}
	\item The first part contains diagrams in which the derivative acts on the Green's function in the loop. The sum of such diagrams is denoted by the symbol $\mathrm{H}_{35}^1$.
	\item The second (remaining) part of the diagrams contains a loop without a derivative and will be indicated by the symbol $\mathrm{H}_{35}^2$.
\end{enumerate}
Thus, we get the following decomposition
\begin{equation}\label{33-w-1}
\mathbb{H}_0^{\mathrm{sc}}(\Gamma_3\mathbb{H}_3^{\mathrm{c}}(\Gamma_5))=\mathrm{H}_{35}^1+
\mathrm{H}_{35}^2.
\end{equation}

\noindent\underline{\textbf{Part} $\mathrm{H}_{35}^1$.} In this case, the diagrams contain four Green's functions, two derivatives, and two integration operators. The intermediate representation can be written as follows
\begin{equation}\label{33-s-11}
\mathrm{H}_{35}^1=\mathbb{H}_0^{\mathrm{sc}}(\Gamma_3\Gamma_{5,1}),
\end{equation}
where
\begin{equation}\label{33-s-1}
\Gamma_{5,1}=\,{\centering\adjincludegraphics[width = 2.1 cm, valign=c]{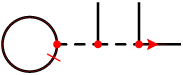}}-
	{\centering\adjincludegraphics[width = 1.5 cm, valign=c]{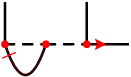}}+
	{\centering\adjincludegraphics[width = 1.5 cm, valign=c]{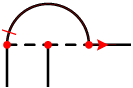}}+
	{\centering\adjincludegraphics[width = 1 cm, valign=c]{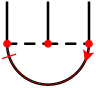}}\,.
\end{equation}
Further, taking into account the main distinguishing feature (the derivative in the loop) of the contributions under consideration, we note that the integral of the product of three Green's functions (without taking into account the loop), where only one of them contains one derivative, does not include nonintegrable densities. Therefore, the only singular component in the diagrams from $\mathrm{H}_{35}^1$ is related to the presence of a singular component in the loop. Thus, using the decomposition for the Green's function near the diagonal, see Lemma \ref{33-lem-7}, the following transition can be made
\begin{equation}\label{33-s-13}
{\centering\adjincludegraphics[width = 1.3 cm, valign=c]{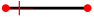}}\,=
\partial_{x^\mu}G_\Lambda(x,y)\big|_{y=x}
\longrightarrow
\frac{L}{2\pi}B_\mu(x)=\frac{L}{2\pi}
\,\,{\centering\adjincludegraphics[width = 1.1 cm, valign=c]{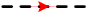}}\,,
\end{equation}
Taking this observation into account, it is advisable to make the replacement
\begin{equation}\label{33-s-2}
\Gamma_{5,1}\longrightarrow\frac{L}{2\pi}
\Bigg(\,{\centering\adjincludegraphics[width = 2.1 cm, valign=c]{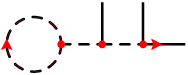}}+
{\centering\adjincludegraphics[width = 1.5 cm, valign=c]{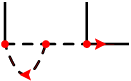}}+
{\centering\adjincludegraphics[width = 1.5 cm, valign=c]{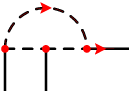}}-
{\centering\adjincludegraphics[width = 1 cm, valign=c]{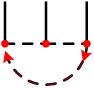}}\,\,\Bigg).
\end{equation}
It is easy to see that the first two diagrams are identically zero due to the antisymmetry of the structure constants, and the last two, using property \eqref{33-eq:Jacobi}, are summarily reduced to the vertex $-\widetilde{\Gamma}_3$, where
\begin{equation}\label{33-s-3}
\widetilde{\Gamma}_3=\,{\centering\adjincludegraphics[width = 1.3 cm, valign=c]{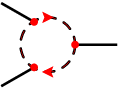}}\,.
\end{equation}
Next, note that by permuting the elements of the last vertex relative to the horizontal axis, we can get the same vertex with a minus sign. This means that $\widetilde{\Gamma}_3=0$, and the final result for the first part can be written as
\begin{equation}\label{33-s-12}
	\mathrm{H}_{35}^1\stackrel{\mathrm{s.p.}}{=}-\frac{L}{2\pi}\mathbb{H}_0^{\mathrm{sc}}(\Gamma_3\widetilde{\Gamma}_3)=0.
\end{equation}

\noindent\underline{\textbf{Part} $\mathrm{H}_{35}^2$.} In the second situation, the loop does not contain a derivative, so the Green's function on the diagonal splits into three parts
\begin{equation}\label{33-s-14}
G_\Lambda^{ab}(x,x)=\frac{L_1}{\pi}\delta^{ab}+2PS_\Lambda^{ab}(x,x)+\mathcal{O}\big(L/\Lambda\big).
\end{equation}
where
\begin{equation}\label{33-s-14-1}
L_1=L+\mathbf{f}(0)/2.
\end{equation}
The latter component does not play an important role, since when regularization is removed, it decreases faster than any degree of the logarithm. In turn, the first two parts lead to non-trivial answers. Let us divide the calculation process into two stages.  With the symbol $\mathrm{H}_{35}^{2,1}$/$\mathrm{H}_{35}^{2,2}$ we notate the diagrams corresponding to the first/second term of \eqref{33-s-14}. Thus, the relation is valid
\begin{equation}\label{33-s-15}
\mathrm{H}_{35}^{2}\stackrel{\mathrm{s.p.}}{=}
\mathrm{H}_{35}^{2,1}+\mathrm{H}_{35}^{2,2},
\end{equation}
Let us start with the first situation. To do this, write out all the possible components of $\mathrm{H}_{35}^{2,1}$, they have the form
\begin{equation}\label{33-s-4}
	\frac{L_1}{\pi}
	\Bigg(\,{\centering\adjincludegraphics[width = 2 cm, valign=c]{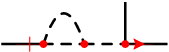}}-
	{\centering\adjincludegraphics[width = 1.7 cm, valign=c]{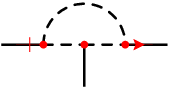}}+
	{\centering\adjincludegraphics[width = 1.7 cm, valign=c]{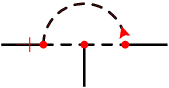}}+
	{\centering\adjincludegraphics[width = 2 cm, valign=c]{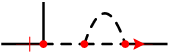}}-
	{\centering\adjincludegraphics[width = 2 cm, valign=c]{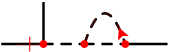}}+
	{\centering\adjincludegraphics[width = 2 cm, valign=c]{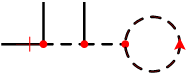}}\,\,\Bigg).
\end{equation}
Next, we define an auxiliary triple vertex $\overline{\Gamma}_3$, which in diagrammatic language can be written as follows
\begin{equation}\label{33-s-5}
	\overline{\Gamma}_3={\centering\adjincludegraphics[width = 1.3 cm, valign=c]{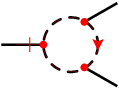}}\,.
\end{equation}
Then, using the properties for the structure constants from \eqref{33-eq:group_symbols}, it is easy to make sure that the third diagram from \eqref{33-s-4} is equal to $-c_2\Gamma_3/2-\overline{\Gamma}_3$, while the rest are all proportional to $\Gamma_3$. Summing up the numerical coefficients, we get the following answer
\begin{equation}\label{33-s-16}
\mathrm{H}_{35}^{2,1}=-\frac{9c_2L_1}{2\pi}\mathbb{H}_0^{\mathrm{sc}}(\Gamma_3\Gamma_3)-\frac{L_1}{\pi}\mathbb{H}_0^{\mathrm{sc}}(\Gamma_3\overline{\Gamma}_3).
\end{equation}
Additionally, we pay attention to the auxiliary decomposition for the vertex $\overline{\Gamma}_3$. To do this, we define another vertex with three external lines of the form
\begin{equation}\label{33-z-1}
\hat{\Gamma}_3=\int_{\mathbb{R}^2}\mathrm{d}^2x\,f^{abc}\phi^b(x)f^{ced}\phi^e(x)\Big(\partial_{x_\mu}B^{dg}_{\mu}(x)\Big)f^{gha}\phi^h(x).
\end{equation}
Then, using integration by parts and property \eqref{33-eq:Jacobi}, we can verify the equality
\begin{equation}\label{33-z-2}
\overline{\Gamma}_3=-\frac{1}{3}\hat{\Gamma}_3-\frac{c_2}{3}\Gamma_3
\end{equation}
and, therefore, rewrite \eqref{33-s-16} as follows
\begin{equation}\label{33-s-16-1}
	\mathrm{H}_{35}^{2,1}=-\frac{25c_2L_1}{6\pi}\mathbb{H}_0^{\mathrm{sc}}(\Gamma_3\Gamma_3)+\frac{L_1}{3\pi}\mathbb{H}_0^{\mathrm{sc}}(\Gamma_3\hat{\Gamma}_3).
\end{equation}
Note that in the second term, the combination $\mathbb{H}_0^{\mathrm{sc}}(\Gamma_3\hat{\Gamma}_3)$ is finite.

Moving on to the diagrams from $\mathrm{H}_{35}^{2,2}$, it is worth additionally dividing the process of finding contributions into two parts. Note that in this case, the loop is replaced by the function $2PS_\Lambda$, so the singularity will follow not because of the Green's function on the diagonal, but because of the appearance of a nonintegrable density. Indeed, the density contains three Green's functions, which are under two derivatives. This fact leads to logarithmic singularities. At the same time, if any Green's function with a derivative is replaced by the function $PS_\Lambda$, then the singularity disappears. Therefore, only the Green's function without derivatives can be replaced with $PS_\Lambda$. Let us first describe the process of calculating contributions containing two functions $PS_\Lambda$. And then we present the answer for the case of a single function $PS_\Lambda$.\\

\noindent\underline{Stage 1.} From the vertex $\Gamma_5$, we construct vertices with three external lines $\mathbb{H}_3^{\mathrm{c}}(\Gamma_5)$ and keep only those that contain a loop without a derivative. Next, in the remaining diagrams, we replace the Green's function $G_\Lambda$ in the loop with the function $2PS_\Lambda$. The result is denoted by the vertex $2\Gamma_{5,1}$, which consists of six diagrams\\

\noindent\underline{Stage 2.} At the vertex $2\Gamma_{5,1}$, we select one of the ends without a derivative and replace its color from black to blue. This is how we transform all diagrams and all external lines. The result is denoted by the vertex $2^{3/2}\Gamma_{5,2}$, which consists of twelve parts and has the following explicit form
\begin{equation}\label{33-s-6}
\Gamma_{5,2}=\,
{\centering\adjincludegraphics[width = 2 cm, valign=c]{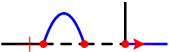}}+
{\centering\adjincludegraphics[width = 2 cm, valign=c]{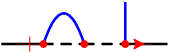}}+
{\centering\adjincludegraphics[width = 1.7 cm, valign=c]{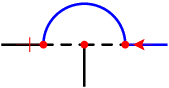}}+
{\centering\adjincludegraphics[width = 1.7 cm, valign=c]{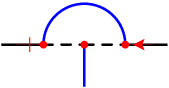}}+
{\centering\adjincludegraphics[width = 1.7 cm, valign=c]{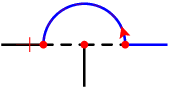}}+
{\centering\adjincludegraphics[width = 1.7 cm, valign=c]{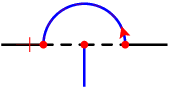}}\,+
\end{equation}
\begin{equation}\label{33-s-7}+
{\centering\adjincludegraphics[width = 2 cm, valign=c]{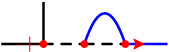}}+
{\centering\adjincludegraphics[width = 2 cm, valign=c]{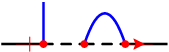}}-
{\centering\adjincludegraphics[width = 2 cm, valign=c]{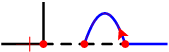}}-
{\centering\adjincludegraphics[width = 2 cm, valign=c]{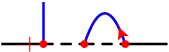}}+
{\centering\adjincludegraphics[width = 2.1 cm, valign=c]{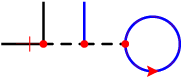}}+
{\centering\adjincludegraphics[width = 2.1 cm, valign=c]{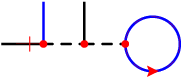}}\,.
\end{equation}
\noindent\underline{Stage 3.} In the vertex $\Gamma_3$, we similarly select the outer line without a derivative and repaint it blue. As a result, we obtain the sum of the two diagrams
\begin{equation}\label{33-s-17}
2^{1/2}\bigg(\,{\centering\adjincludegraphics[width = 1 cm, valign=c]{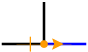}}+
{\centering\adjincludegraphics[width = 1 cm, valign=c]{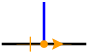}}\,\bigg).
\end{equation}
Next, we note that at the second vertex, it is possible to integrate in parts, as well as transfer the background field from one line to another. At the same time, if we exclude vertices containing the blue line with a derivative or a background field with a derivative (such vertices do not lead to singular contributions), then we can make sure that the second vertex can be replaced by half of the first one. As a result, we come to the diagram $2^{-1/2}3\,\Gamma_{3,1}$, where
\begin{equation}\label{33-s-18}
\Gamma_{3,1}=\,{\centering\adjincludegraphics[width = 1 cm, valign=c]{fig/A3-013.eps}}\,.
\end{equation}
Note that the color of the variable in this case is different from the color in $\Gamma_{5,2}$.\\

\noindent\underline{Stage 4.} Connect vertices from $2^{3/2}\Gamma_{5,2}$ with vertices from $2^{-1/2}3\,\Gamma_{3,1}$ in all possible ways, observing the color matching rule (blue-blue, black-black). It is clear that 24 diagrams can be obtained in this way. However, given the fact that the derivative can be shifted from one black line to another, discarding convergent (without singularities) diagrams in which the derivative acts on the blue line or the background field, we get the rule: connect the black line with/without the derivative only with its own kind, and double the result.\\

\noindent\underline{Stage 5.} Having obtained the final set of 12 diagrams, we use Lemma \ref{33-lem-2} to find the singular component, which, in the language of diagrammatic technique, boils down to the following transformations in each diagram: 
\begin{itemize}
	\item remove derivatives (dashes) from the diagram;
	\item replace black solid lines with black dashed ones;
	\item make the variables and background fields the same color;
	\item add the multiplier\footnote{In this case, it is taken into account that $G_\Lambda(x,y)$ contains in the decomposition $2G^{\Lambda,\mathbf{f}}(x-y)$.} $L/\pi$.
\end{itemize}
Finally, using relations \eqref{33-eq:tr_S} and \eqref{33-eq:Jacobi} and a less trivial equality of the form
\begin{equation}\label{33-s-19}
{\centering\adjincludegraphics[width = 1.7 cm, valign=c]{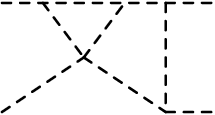}}=-\frac{c_2}{2}\,
{\centering\adjincludegraphics[width = 0.95 cm, valign=c]{fig/Cross_2.eps}}\,-\,
{\centering\adjincludegraphics[width = 1.7 cm, valign=c]{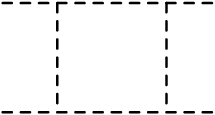}}\,,
\end{equation}
it can be shown that the singular part of $\mathrm{H}_{35}^{2,2}$, containing contributions with two functions $PS_\Lambda$, can be represented as a linear combination of diagrams from Section \ref{33:sec:dia}.
Moreover, it can also be established that the remaining singular component $\mathcal{J}_{35}$ is proportional to the first power of the function $PS_\Lambda(x,x)$ and the second power of the background field $B_\mu(x)$. Thus, the intermediate equality can be written as follows
\begin{equation}\label{33-s-20}
\mathrm{H}_{35}^{2,2}\stackrel{\mathrm{s.p.}}{=}
\frac{30L}{\pi}\Big(\mathrm{A}_1+2\mathrm{A}_2-2\mathrm{A}_3-\mathrm{A}_4\Big)
-\frac{15c_2L}{\pi}
\Big(\mathrm{B}_1+2\mathrm{B}_2+2\mathrm{B}_3\Big)+\mathcal{J}_{35}.
\end{equation}
Let us explain the calculation process using the example of connecting the third vertex from \eqref{33-s-6} and the vertex $\Gamma_{3,1}$ from \eqref{33-s-18}. Indeed, taking into account the second property from \eqref{33-eq:tr_S}, we have\footnote{It is assumed here that the operator $\mathbb{H}_0^{\mathrm{sc}}$ takes into account the rules for connecting colored lines.}
\begin{equation}\label{33-s-21}
\mathbb{H}_0^{\mathrm{sc}}\bigg(\,{\centering\adjincludegraphics[width = 1.7 cm, valign=c]{fig/A3-003.eps}}\,\,
{\centering\adjincludegraphics[width = 1 cm, valign=c]{fig/A3-013.eps}}\,\bigg)\longrightarrow
2\times\frac{L}{\pi}\times\bigg(-\frac{c_2}{2}\bigg)\times
\,{\centering\adjincludegraphics[width = 0.8 cm, valign=c]{fig/A4-006.eps}}\,.
\end{equation}

Next, we complete the analysis of this group of diagrams by searching for the contribution $\mathcal{J}_{35}$ proportional to the function $PS_{\Lambda}$. For these purposes, we introduce the vertex $\Gamma_{5,3}$, which is equal to half of the vertex $\Gamma_{5,2}$, if all the outer lines in the latter are black. Thus, considering that $\Gamma_{5,3}$ contains a single function $PS_{\Lambda}$, it can be argued that the contribution of $\mathcal{J}_{35}$ is equal to the singular part of the diagrams
\begin{equation}\label{33-w-2}
2\mathbb{H}_0^{\mathrm{sc}}(\Gamma_{5,3}\Gamma_3)\Big|_{G_\Lambda(\,\cdot\,,\,\cdot\,)\to
2G^{\Lambda,\mathbf{f}}(\,\cdot\,-\,\cdot\,)}.
\end{equation}
Here, using the last transition, all Green's functions are replaced by the main parts of the expansion from \eqref{33-r-1}. Next, we transform the vertex $\Gamma_3$. If we fix the outer lines, then all possible variants for connecting such a vertex are equal
\begin{equation}\label{33-w-3}
{\centering\adjincludegraphics[reflect, angle=180, width = 1 cm, valign=c]{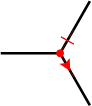}}-
{\centering\adjincludegraphics[ width = 1 cm, valign=c]{fig/A7-007-01.eps}}+
{\centering\adjincludegraphics[width = 1 cm, valign=c]{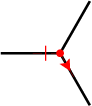}}-
{\centering\adjincludegraphics[reflect, angle=180,width = 1 cm, valign=c]{fig/A7-009-01.eps}}+
{\centering\adjincludegraphics[width = 1 cm, valign=c]{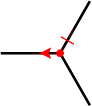}}-
{\centering\adjincludegraphics[reflect, angle=180,width = 1 cm, valign=c]{fig/A7-011-01.eps}}.
\end{equation}
Next, using integration by parts and transfer the background field from one line to another, the last vertices with a fixed position of the lines can be converted into the combination
\begin{equation}\label{33-w-4}
3\,{\centering\adjincludegraphics[reflect, angle=180, width = 1 cm, valign=c]{fig/A7-007-01.eps}}-3\,
	{\centering\adjincludegraphics[ width = 1 cm, valign=c]{fig/A7-007-01.eps}}+
	\Bigg(
	{\centering\adjincludegraphics[width = 1 cm, valign=c]{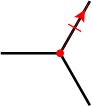}}-
	{\centering\adjincludegraphics[width = 1 cm, valign=c]{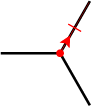}}\,
	\Bigg)-
\Bigg(
{\centering\adjincludegraphics[reflect, angle=180,width = 1 cm, valign=c]{fig/A7-014-01.eps}}+
{\centering\adjincludegraphics[reflect, angle=180,width = 1 cm, valign=c]{fig/A7-013-01.eps}}\,
\Bigg).
\end{equation}
The convenience of this representation is that the vertices in parentheses do not contain lines with a derivative, since the difference $\partial_{x_\mu}B_\mu(x)-B_\mu(x)\partial_{x_\mu}$ is a multiplication operator by a function. Therefore, it is necessary to consider only the first two terms. Thus, formula \eqref{33-w-2} is rewritten as
\begin{equation}\label{33-w-6}
6\mathbb{\hat{H}}_0^{\mathrm{sc}}\Bigg(\Gamma_{5,3}\,{\centering\adjincludegraphics[reflect, angle=180, width = 1 cm, valign=c]{fig/A7-007-01.eps}}\,-
\,\Gamma_{5,3}\,{\centering\adjincludegraphics[ width = 1 cm, valign=c]{fig/A7-007-01.eps}}\,\Bigg)\Bigg|_{G_\Lambda(\,\cdot\,,\,\cdot\,)\to
	2G^{\Lambda,\mathbf{f}}(\,\cdot\,-\,\cdot\,)},
\end{equation}
where the "hat" fixes the external lines. In other words, not all possible combinations of pairings should be considered, but only ones with fixed external lines.

Let us use the functionals from Section \ref{33:sec:dia}, and also define the auxiliary integral
\begin{align}\label{33-w-7}
	\mathrm{I}_1(\Lambda,\sigma)&\,=\,4\int_{\mathbf{B}_{1/\sigma}}\mathrm{d}^2x\,
	\Big(\partial_{x^\mu}G^{\Lambda,\mathbf{f}}(x)\Big)
	\Big(\partial_{x_\mu}G^{\Lambda,\mathbf{f}}(x)\Big)
	G^{\Lambda,\mathbf{f}}(x)\\\nonumber&\,=\,
	2\int_{\mathbf{B}_{1/\Lambda}}\mathrm{d}^2x\,
	\Big(A_0(x)G^{\Lambda,\mathbf{f}}(x)\Big)
	\Big(G^{\Lambda,\mathbf{f}}(x)\Big)^2\\\nonumber&\stackrel{\mathrm{s.p.}}{=}\,
	\frac{2LL_2-L^2}{2\pi^2}\stackrel{\mathrm{s.p.}}{=}\frac{L_1^2+2L\theta_1^{\phantom{1}}}{2\pi^2}.
\end{align}
Then, after a series of transformations, we get the answer in the form
\begin{equation}\label{33-w-8}
\mathcal{J}_{35}\stackrel{\mathrm{s.p.}}{=}\frac{15}{2}\mathrm{I}_1(\Lambda,\sigma)
\Big(2c_2^2J_1[B]-c_2J_2[B]+2J_3[B]\Big).
\end{equation}
For example, let us calculate one term explicitly:
\begin{equation}\label{33-w-9}
-6\mathbb{\hat{H}}_0^{\mathrm{sc}}\Bigg({\centering\adjincludegraphics[width = 2.2 cm, valign=c]{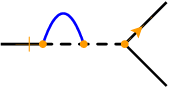}}
\,\,
{\centering\adjincludegraphics[reflect, angle=180, width = 1 cm, valign=c]{fig/A7-007-01.eps}}\,\Bigg)=-6\mathrm{I}_1(\Lambda,\sigma)\,
{\centering\adjincludegraphics[width = 0.6 cm, valign=c]{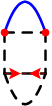}}=
6\mathrm{I}_1(\Lambda,\sigma)J_3[B].
\end{equation}
The final result of all calculations in this section can be presented in the form of a statement.
\begin{lemma}\label{33-lem-3}
Taking into account all of the above, the following relation holds
\begin{align}
\mathbb{H}_0^{\mathrm{sc}}(\Gamma_3\Gamma_5)\stackrel{\mathrm{s.p.}}{=}
&-\frac{25c_2L_1}{6\pi}\mathbb{H}_0^{\mathrm{sc}}(\Gamma_3\Gamma_3)+\frac{L_1}{3\pi}\mathbb{H}_0^{\mathrm{sc}}(\Gamma_3\hat{\Gamma}_3)\\\label{33-s-22}&+
\frac{30L}{\pi}\Big(\mathrm{A}_1+2\mathrm{A}_2-2\mathrm{A}_3-\mathrm{A}_4\Big)
-\frac{15c_2L}{\pi}
\Big(\mathrm{B}_1+2\mathrm{B}_2+2\mathrm{B}_3\Big)\\&+
\frac{15(L_1^2+2L\theta_1^{\phantom{1}})}{4\pi^2}
\Big(2c_2^2J_1[B]-c_2J_2[B]+2J_3[B]\Big).
\end{align}
\end{lemma}

\subsection{The case $\mathbb{H}_0^{\mathrm{sc}}(\Gamma_3^4)$}
\label{33:sec:cl:3}

Turning to the search for asymptotic expansions for the group of integrals from $\mathbb{H}_0(\Gamma_3^4)$, we note that, omitting operator and group notation, each of the available diagrams is equivalent to one of the following pictures
\begin{equation}\label{33-s-28}
{\centering\adjincludegraphics[width = 1.7 cm, valign=c]{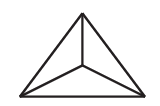}},\,\,\,
{\centering\adjincludegraphics[width = 1.1 cm, valign=c]{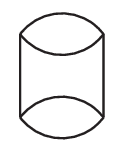}}.
\end{equation}
Thus, the calculation is divided into two parts. In addition, using the previously introduced definitions for the diagram technique, we can verify the validity of the following decomposition
\begin{equation}\label{33-1-3}
\mathbb{H}_2^{\mathrm{sc}}(\Gamma_3^2)=
2\mathrm{C}_1+
\mathrm{C}_2-
2\mathrm{C}_3-
\mathrm{C}_4-
\mathrm{C}_5+
\mathrm{C}_6+
\mathrm{C}_7+
2\mathrm{C}_8+
2\mathrm{C}_9-
\mathrm{C}_{10}-
2\mathrm{C}_{11}-
2\mathrm{C}_{12},
\end{equation}	
where
\begin{equation}\label{33-s-23}
\mathrm{C}_{1}=
{\centering\adjincludegraphics[width = 1.5 cm, valign=c]{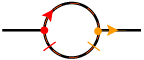}},\,\,\,
\mathrm{C}_{2}=
{\centering\adjincludegraphics[width = 1.5 cm, valign=c]{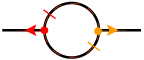}},\,\,\,
\mathrm{C}_{3}=
{\centering\adjincludegraphics[width = 1.5 cm, valign=c]{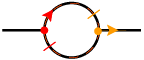}},\,\,\,
\mathrm{C}_{4}=
{\centering\adjincludegraphics[width = 1.5 cm, valign=c]{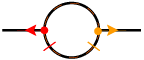}},
\end{equation}
\begin{equation}\label{33-s-24}
\mathrm{C}_{5}=
{\centering\adjincludegraphics[width = 1.5 cm, valign=c]{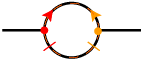}},\,\,\,
\mathrm{C}_{6}=
{\centering\adjincludegraphics[width = 1.5 cm, valign=c]{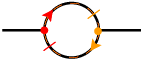}},\,\,\,
\mathrm{C}_{7}=
{\centering\adjincludegraphics[width = 1.5 cm, valign=c]{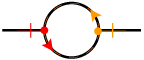}},\,\,\,
\mathrm{C}_{8}=
{\centering\adjincludegraphics[width = 1.5 cm, valign=c]{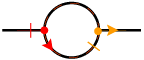}},
\end{equation}
\begin{equation}\label{33-s-25}
\mathrm{C}_{9}=
{\centering\adjincludegraphics[width = 1.5 cm, valign=c]{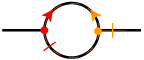}},\,
\mathrm{C}_{10}=
{\centering\adjincludegraphics[width = 1.5 cm, valign=c]{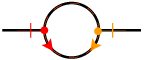}},\,
\mathrm{C}_{11}=
{\centering\adjincludegraphics[width = 1.5 cm, valign=c]{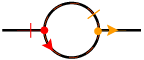}},\,
\mathrm{C}_{12}=
{\centering\adjincludegraphics[width = 1.5 cm, valign=c]{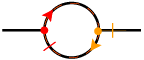}}.
\end{equation}
Next, we define two auxiliary vertices with two external lines
\begin{equation}\label{33-s-26}
\mathrm{X}_{1} =\frac{c_{2}}{4\pi}\int_{\mathbb{R}^{2}}\mathrm{d}^2x\,\phi^{a}(x)B^{ac}_{\mu}(x)B^{ce}_{\mu}(x)\phi^{e}(x)=
\,
{\centering\adjincludegraphics[width = 1 cm, valign=c]{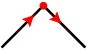}}\,,
\end{equation}
\begin{equation}\label{33-s-27}
\mathrm{X}_{2} = -\frac{1}{2\pi}\int_{\mathbb{R}^{2}}\mathrm{d}^2x\,f^{abc}f^{ced}\phi^{a}(x)\phi^{d}(x)B^{bg}_{\mu}(x)B^{ge}_{\mu}(x).
\end{equation}

\begin{lemma}\label{33-lem-1} Let $\sigma_1>0$. Then, taking into account all the above, the equality is true
\begin{equation}\label{33-s-30}
\mathbb{H}_0^{\mathrm{sc}}\big(\Gamma_3^4\big)\stackrel{\mathrm{s.p.}}{=}
-2^23^3\ln(\Lambda/\sigma_1)\mathbb{H}_0^{\mathrm{sc}}\big(\Gamma_3^2\mathrm{X}_1^{\phantom{1}}\big)-2^23^5\ln^2(\Lambda/\sigma_1)\mathbb{H}_0\big(\mathrm{X}_1^2\big).
\end{equation}	
\end{lemma}
\begin{proof} The contribution $\mathbb{H}_0\big(\Gamma_3^4\big)$ contains diagrams topologically equivalent to either triangular or cylindrical diagrams, see \eqref{33-s-28}. The first type of diagrams does not provide a singular contribution, which is a consequence of the low dimension, that is $d=2$, the number of integration operators, and the fact that only single lines appear in such diagrams. Auxiliary estimates for integrals appearing in triangular diagrams are given in Section \ref{33:sec:appl:3}.
	
Turning to the cylindrical part, we note that the contributions are not singular if two derivatives act on the side line (left or right), which is a consequence of the convergence of the integral
\begin{equation}\label{33-s-29}
\int_{\mathrm{B}_{1/\sigma}}\mathrm{d}^2x\,A_0^{\phantom{1}}(x)G^{1,\mathbf{f}}(x)=1,
\end{equation}
which in fact does not depend on the regularizing parameter $\Lambda$. This reasoning leads to the fact that it is convenient to cut a cylindrical diagram into two parts (upper and lower), which in diagrammatic language can be represented by the same vertex $\mathbb{H}_2^{\mathrm{c}}(\Gamma_3^2)$, and then investigate this part separately.

Thus, the analysis of the cylindrical diagram reduces to the asymptotic decomposition of the integral $\mathbb{H}_2^{\mathrm{c}}(\Gamma_3^2)$ with respect to the parameter $\Lambda$ as a functional considered on smooth fields. Note that it can be represented as a linear combination of auxiliary vertices, see \eqref{33-1-3}. Taking into account the result from Lemma \ref{33-lem-2} and the Jacobi identity \eqref{33-eq:Jacobi} for structure constants, repeating the calculations from the fifth stage of the proof of Lemma \ref{33-lem-3}, we can verify the validity of the following equalities
\begin{equation*}
\mathrm{C}_1\stackrel{\mathrm{s.p.}}{=}-2\ln(\Lambda/\sigma_1)\mathrm{X}_1,\,\,\,
\mathrm{C}_2\stackrel{\mathrm{s.p.}}{=}-4\ln(\Lambda/\sigma_1)\mathrm{X}_1,\,\,\,
\mathrm{C}_3\stackrel{\mathrm{s.p.}}{=}2\ln(\Lambda/\sigma_1)\mathrm{X}_1,
\end{equation*}
\begin{equation*}
\mathrm{C}_4\stackrel{\mathrm{s.p.}}{=}4\ln(\Lambda/\sigma_1)\mathrm{X}_1,\,\,\,
\mathrm{C}_5\stackrel{\mathrm{s.p.}}{=}2\ln(\Lambda/\sigma_1)\mathrm{X}_2,\,\,\,
\mathrm{C}_6\stackrel{\mathrm{s.p.}}{=}-2\ln(\Lambda/\sigma_1)\mathrm{X}_1+2\ln(\Lambda/\sigma_3)\mathrm{X}_2,
\end{equation*}
\begin{equation*}
\mathrm{C}_i\stackrel{\mathrm{s.p.}}{=}0,
\end{equation*}
where $i\in\{7,\ldots,12\}$. Next, for convenience, we define two functionals $\mathrm{C}=\mathbb{H}_2^{\mathrm{c}}(\Gamma_3^2)$ and $\mathrm{C}_0=-18\mathrm{X}_1$. Then, based on the last calculations, we can make sure that the equality $\mathrm{C}\stackrel{\mathrm{s.p.}}{=}\ln(\Lambda/\sigma_1)\mathrm{C}_0$ holds and the chain of relations is true
\begin{align*}
\mathbb{H}_0^{\mathrm{sc}}\big(\Gamma_3^4\big)&\stackrel{\mathrm{s.p.}}{=}
3\mathbb{H}_0^{\mathrm{sc}}\big(\mathbb{H}_2^{\mathrm{c}}(\Gamma_3^2)\mathbb{H}_2^{\mathrm{c}}(\Gamma_3^2)\big)=
3\mathbb{H}_0^{\mathrm{sc}}\big(\mathrm{C}^2\big)\\&\,=\,
3\ln^2(\Lambda/\sigma_1)\mathbb{H}_0^{\mathrm{sc}}\big(\mathrm{C}_0^2\big)+6\ln(\Lambda/\sigma_1) \mathbb{H}_0^{\mathrm{sc}}\big((\mathrm{C}-\ln(\Lambda/\sigma_1)\mathrm{C}_0^{\phantom{1}})\mathrm{C}_0^{\phantom{1}}\big)
+3\mathbb{H}_0^{\mathrm{sc}}\big((\mathrm{C}-\ln(\Lambda/\sigma_1)\mathrm{C}_0^{\phantom{1}})^2\big)\\&\stackrel{\mathrm{s.p.}}{=}
6\ln(\Lambda/\sigma_1)\mathbb{H}_0^{\mathrm{sc}}\big(\mathrm{C}\mathrm{C}_0^{\phantom{1}}\big)-3\ln^2(\Lambda/\sigma_1)\mathbb{H}_0^{\mathrm{c}}\big(\mathrm{C}_0^2\big).
\end{align*}
From the last equalities the main statement follows.
\end{proof}

\subsection{The case $\mathbb{H}_0^{\mathrm{sc}}(\Gamma_3^2\Gamma_4^{\phantom{1}})$}
\label{33:sec:cl:4}
Let us move on to a group of diagrams that can be constructed using two triple vertices and one quartic vertex. Such a group can be divided into two parts, the first of which includes diagrams containing the Green's function on the diagonal, and the second does not. Omitting operator and group elements, they can be schematically represented as follows 
\begin{equation}\label{33-d-1}
{\centering\adjincludegraphics[angle=-90, width = 1.9 cm, valign=c]{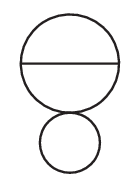}},\,\,\,
{\centering\adjincludegraphics[width = 1.6 cm, valign=c]{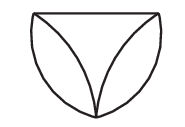}}.
\end{equation}
Thus, we obtain the decomposition
\begin{equation}\label{33-d-2}
\mathbb{H}_0^{\mathrm{sc}}(\Gamma_3^2\Gamma_4^{\phantom{1}})=\mathrm{H}_{334}^1+\mathrm{H}_{334}^2,
\end{equation}
where $\mathrm{H}_{334}^1$ contains diagrams corresponding to the first type, and $\mathrm{H}_{334}^2$ contains diagrams corresponding to the second category. We consider them separately.

\subsubsection{Part with one loop $\mathbb{H}_2^{\mathrm{c}}(\Gamma_4)$}
Let us move on to the first situation. In this case, it is convenient to use the expansion from \eqref{33-1-3}, taking into account the analysis of singularities from the proof of Lemma \ref{33-lem-1} and writing out the following chain of relations
\begin{align}\label{33-d-3}
\mathrm{H}_{334}^1&=\mathbb{H}_0^{\mathrm{sc}}\big(\Gamma_3^2\mathbb{H}_2^{\mathrm{c}}(\Gamma_4^{\phantom{1}})\big)=\mathbb{H}_0^{\mathrm{sc}}\big(\mathbb{H}_2^{\mathrm{c}}(\Gamma_3^2)\mathbb{H}_2^{\mathrm{c}}(\Gamma_4^{\phantom{1}})\big)\\\nonumber&=
-18\ln(\Lambda/\sigma_2)\mathbb{H}_0^{\mathrm{sc}}\big(\mathrm{X}_1^{\phantom{1}}\mathbb{H}_2^{\mathrm{c}}(\Gamma_4^{\phantom{1}})\big)+
\mathbb{H}_0^{\mathrm{sc}}\big(
\big[\mathbb{H}_2^{\mathrm{c}}(\Gamma_3^2)+18\ln(\Lambda/\sigma_2)\mathrm{X}_1^{\phantom{1}}
\big]
\mathbb{H}_2^{\mathrm{c}}(\Gamma_4^{\phantom{1}})\big),
\end{align}
where $\sigma_2>0$ is an auxiliary parameter. Next, we represent $\mathbb{H}_2^{\mathrm{c}}(\Gamma_4)$ as the sum of three terms
\begin{equation}\label{33-d-4}
\mathbb{H}_2^{\mathrm{c}}(\Gamma_4^{\phantom{1}})=\mathrm{D}_1^{\phantom{1}}+\mathrm{D}_2^{\phantom{1}}+\mathrm{D}_3^{\phantom{1}}.
\end{equation}
In the diagrammatic language, the last terms can be shown as follows
\begin{equation}\label{33-d-5}
\mathrm{D}_1^{\phantom{1}}=-\,{\centering\adjincludegraphics[width = 0.6 cm, valign=c]{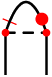}}
\,,\,\,\,
\mathrm{D}_3^{\phantom{1}}=-\,{\centering\adjincludegraphics[width = 0.65 cm, valign=c]{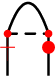}}\,,
\end{equation}
\begin{equation}\label{33-d-6}
\mathrm{D}_2^{\phantom{1}}=
\,{\centering\adjincludegraphics[width = 0.6 cm, valign=c]{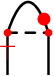}}+
\,{\centering\adjincludegraphics[width = 0.6 cm, valign=c]{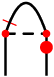}}+
\,{\centering\adjincludegraphics[width = 1.3 cm, valign=c]{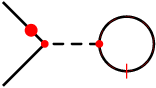}}+
\,{\centering\adjincludegraphics[width = 1.3 cm, valign=c]{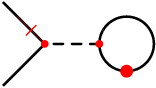}}\,.
\end{equation}
Further, for convenience of calculations, taking into account \eqref{33-d-4}, we divide the second term from \eqref{33-d-3} into the corresponding three parts
\begin{equation}\label{33-d-7}
\mathrm{H}_{334}^{1,i}=
\mathbb{H}_0^{\mathrm{sc}}\big(
\big[\mathbb{H}_2^{\mathrm{c}}(\Gamma_3^2)+18\ln(\Lambda/\sigma_2)\mathrm{X}_1^{\phantom{1}}
\big]\mathrm{D}_i^{\phantom{1}}\big),
\end{equation}
where $i=1,2,3$. Let us study them separately.\\

\noindent\underline{Part $\mathrm{H}_{334}^{1,1}$.} Consider an auxiliary vertex with two external lines and write out the kernel of the integral operator connecting the field functions
\begin{equation}\label{33-d-8}
\mathbb{H}_2^{\mathrm{c}}(\Gamma_3^2)+18\ln(\Lambda/\sigma_2)\mathrm{X}_1^{\phantom{1}}=
\int_{\mathbb{R}^2}\mathrm{d}^2x\int_{\mathbb{R}^2}\mathrm{d}^2y\,
K^{ab}(x,y)\phi^a(x)\phi^b(y).
\end{equation}
It is clear that this kernel is symmetric $K^{ab}(x,y)=K^{ba}(y,x)$ differential operator of the second (totally, with respect to $x$ and $y$) order. For each variable, it contains no more than one derivative. Note that symmetry follows from the presence of convolution with field functions that commute with each other. Moreover, the coefficients of such an operator are finite and, due to the calculations in Section \ref{33:sec:cl:3}, have a finite limit when removing the regularization of $\Lambda\to+\infty$. Taking into account the latest formula and the rules of diagrammatic technique from Section \ref{33:sec:int:1}, the value under study is written out explicitly
\begin{equation}\label{33-d-9}
\mathrm{H}_{334}^{1,1}=-
\int_{\mathbb{R}^2}\mathrm{d}^2z\,
\rho_2^{ed}(z)\Big(\partial_{z_\mu}D^{ef}_\mu(y)G_\Lambda^{df}(z,y)\Big)\Big|_{y=z},
\end{equation}
where
\begin{equation}\label{33-d-10}
\rho_2^{ed}(z)=2\int_{\mathbb{R}^2}\mathrm{d}^2x\int_{\mathbb{R}^2}\mathrm{d}^2y\,
K^{ab}(x,y)G_\Lambda^{ag}(x,z)G_\Lambda^{bh}(y,z)f^{gce}f^{chd}.
\end{equation}
It is clear that $\rho_2^{ed}(z)=\rho_2^{de}(z)$. Next, we integrate by parts\footnote{We do not take into account the boundary terms, assuming that at infinity the functions decrease quite quickly.} into \eqref{33-d-9} and then reconstruct the Laplace operator \eqref{33-o-3}, then we get
\begin{multline}\label{33-d-11}
\int_{\mathbb{R}^2}\mathrm{d}^2z\,
\Big(\partial_{z_\mu}\rho_2^{ed}(z)\Big)\Big(D^{ef}_\mu(y)G_\Lambda^{df}(z,y)\Big)\Big|_{y=z}-2
\int_{\mathbb{R}^2}\mathrm{d}^2z\,
\rho_2^{de}(z)\Big(A^{ef}(y)G_\Lambda^{fd}(y,z)\Big)\Big|_{y=z}-\\-\frac{1}{2}
\int_{\mathbb{R}^2}\mathrm{d}^2z\,
\rho_2^{de}(z)\Big(\partial_{z_\mu}B^{ef}_\mu(z)\Big)G_\Lambda^{fd}(z,z).
\end{multline}
It can be verified that due to the symmetry of $\rho_2^{ed}(z)$ and the presence of relations from Lemma \ref{33-lem-7}, the first and third terms in the last sum do not contain singular contributions. Therefore, it is necessary to consider only the second part, which, due to the possibility of transition
\begin{equation}\label{33-d-12}
A^{ef}(y)G_\Lambda^{fd}(y,z)\Big|_{y=z}\to-\partial_{z_\mu}\partial_{z^\mu}
G^{\Lambda,\mathbf{f}}(z)\Big|_{z=0}\delta^{ed}=\frac{\Lambda^2\alpha_1(\mathbf{f})}{4\pi}\delta^{ed}
\end{equation}
can be rewritten as
\begin{equation}\label{33-d-13}
\mathrm{H}_{334}^{1,1}\stackrel{\mathrm{s.p.}}{=}-
\frac{\Lambda^2\alpha_1(\mathbf{f})}{2\pi}
\int_{\mathbb{R}^2}\mathrm{d}^2z\,
\rho_2^{ee}(z).
\end{equation}
Thus, if we introduce an auxiliary vertex into consideration
\begin{equation}\label{33-d-14}
\widetilde{\Gamma}_2=\frac{c_2\alpha_1(\mathbf{f})}{2\pi}\int_{\mathbb{R}^2}\mathrm{d}^2x\,\phi^a(x)\phi^a(x),
\end{equation}
then the result can be rewritten in the form
\begin{equation}\label{33-d-15}
\mathrm{H}_{334}^{1,1}\stackrel{\mathrm{s.p.}}{=}\Lambda^2
\mathbb{H}_0^{\mathrm{sc}}\big(\Gamma_3^2\widetilde{\Gamma}_2^{\phantom{1}}\big)+
18\ln(\Lambda/\sigma_2)\Lambda^2\mathbb{H}_0^{\mathrm{sc}}\big(\mathrm{X}_1^{\phantom{1}}
\widetilde{\Gamma}_2^{\phantom{1}}\big).
\end{equation}
\noindent\underline{Part $\mathrm{H}_{334}^{1,2}$.} Turning to the second case from \eqref{33-d-7}, we transform individual parts of the vertex \eqref{33-d-6}. Note that the following chain of relations is true for the sum of the second and third terms
\begin{align}\label{33-d-16}
{\centering\adjincludegraphics[width = 0.6 cm, valign=c]{fig/A6-004.eps}}+
\,{\centering\adjincludegraphics[width = 1.5 cm, valign=c]{fig/A6-005.eps}}&\stackrel{\mathrm{s.p.}}{=}
\frac{\ln(\Lambda/\sigma_3)}{2\pi}\,
{\centering\adjincludegraphics[width = 0.6 cm, valign=c]{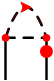}}+
\frac{\ln(\Lambda/\sigma_3)}{2\pi}
\,{\centering\adjincludegraphics[width = 1.5 cm, valign=c]{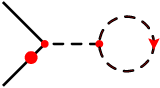}}
\\\nonumber&\,=
-\frac{3\ln(\Lambda/\sigma_3)c_2}{4\pi}
\int_{\mathbb{R}^2}\mathrm{d}^2x\,\phi^a(x)B_\mu^{ac}(x)D_\mu^{cb}(x)\phi^b(x),
\end{align}
where therelations from \eqref{33-eq:tr_S} and \eqref{33-r-28} were used in the transitions. Here $\sigma_3>0$ is an auxiliary parameter. Similarly, using \eqref{33-r-32} instead of \eqref{33-r-28}, we can obtain the following chain of equalities for the sum of the first and fourth terms
\begin{align}\label{33-d-17}
{\centering\adjincludegraphics[width = 0.6 cm, valign=c]{fig/A6-003.eps}}+
\,{\centering\adjincludegraphics[width = 1.5 cm, valign=c]{fig/A6-006.eps}}&\stackrel{\mathrm{s.p.}}{=}
-\frac{\ln(\Lambda/\sigma_3)}{2\pi}\,
{\centering\adjincludegraphics[width = 0.6 cm, valign=c]{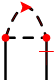}}
-\frac{\ln(\Lambda/\sigma_3)}{2\pi}
\,{\centering\adjincludegraphics[width = 1.5 cm, valign=c]{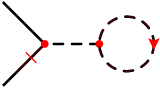}}
\\\nonumber&\,=
\frac{3\ln(\Lambda/\sigma_3)c_2}{4\pi}
\int_{\mathbb{R}^2}\mathrm{d}^2x\,\phi^a(x)B_\mu^{ab}(x)\partial_{x_\mu}^{\phantom{1}}\phi^b(x).
\end{align}
Then, in total, the vertex $\mathrm{D}_2$ can be transformed according to the following relation
\begin{equation}\label{33-d-18}
\mathrm{D}_2\stackrel{\mathrm{s.p.}}{=}
\frac{3\ln(\Lambda/\sigma_3)c_2}{4\pi}
\int_{\mathbb{R}^2}\mathrm{d}^2x\,\phi^a(x)B_\mu^{ac}(x)B_\mu^{cb}(x)\phi^b(x)=
3\ln(\Lambda/\sigma_3)\mathrm{X}_1=
3\ln(\Lambda/\sigma_3)\,
{\centering\adjincludegraphics[width = 1 cm, valign=c]{fig/A6-009.eps}}\,.
\end{equation}
The final answer for the second case is written out as follows
\begin{equation}\label{33-d-19}
\mathrm{H}_{334}^{1,2}\stackrel{\mathrm{s.p.}}{=}
3\ln(\Lambda/\sigma_3)
\mathbb{H}_0^{\mathrm{sc}}\big(\Gamma_3^2\mathrm{X}_1^{\phantom{1}}\big)
+54\ln(\Lambda/\sigma_2)\ln(\Lambda/\sigma_3)
\mathbb{H}_0^{\mathrm{sc}}\big(\mathrm{X}_1^{2}\big).
\end{equation}
\noindent\underline{Part $\mathrm{H}_{334}^{1,3}$.} Let us move on to the last part. In this case, the diagram is decomposed into the sum of two parts
\begin{equation}\label{33-d-20}
\mathrm{D}_3=-\,{\centering\adjincludegraphics[width = 0.7 cm, valign=c]{fig/A6-002.eps}}\stackrel{\mathrm{s.p.}}{=}
-\frac{L_1}{\pi}
\,{\centering\adjincludegraphics[width = 0.7 cm, valign=c]{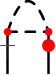}}-2
\,{\centering\adjincludegraphics[width = 0.7 cm, valign=c]{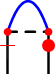}}\,.
\end{equation}
Further, by applying integration by parts and discarding the finite terms (before and after removing the regularization), we can proceed to the following auxiliary vertices with two external lines
\begin{equation}\label{33-d-21}
{\centering\adjincludegraphics[width = 0.7 cm, valign=c]{fig/A6-012.eps}}\to-2c_2
\int_{\mathbb{R}^2}\mathrm{d}^2x\,\phi^a(x)A^{ab}(x)\phi^b(x)=-2\pi\hat{\Gamma}_2,
\end{equation}
\begin{equation}\label{33-d-22}
{\centering\adjincludegraphics[width = 0.7 cm, valign=c]{fig/A6-013.eps}}\to
2\int_{\mathbb{R}^2}\mathrm{d}^2x\,\phi^a(x)
\Big(f^{abc}PS_\Lambda^{bd}(x,x)f^{cde}\Big)
A^{eh}(x)\phi^h(x)=2\overline{\Gamma}_2.
\end{equation}
Thus, the transition is acceptable
\begin{equation}\label{33-d-23}
\mathrm{H}_{334}^{1,3}=
\mathbb{H}_0^{\mathrm{sc}}\big(
\big[\mathbb{H}_2^{\mathrm{c}}(\Gamma_3^2)+18\ln(\Lambda/\sigma_2)\mathrm{X}_1^{\phantom{1}}
\big]\mathrm{D}_3^{\phantom{1}}\big)
\stackrel{\mathrm{s.p.}}{=}
\mathbb{H}_0^{\mathrm{sc}}\big(
\big[\mathbb{H}_2^{\mathrm{c}}(\Gamma_3^2)+18\ln(\Lambda/\sigma_2)\mathrm{X}_1^{\phantom{1}}
\big](2L_1^{\phantom{1}}\hat{\Gamma}_2^{\phantom{1}}-4\overline{\Gamma}_2^{\phantom{1}})\big).
\end{equation}
Let us consider all four terms separately. Performing calculations similar to those that were performed earlier, we obtain
\begin{equation}\label{33-d-26}
\mathbb{H}_0^{\mathrm{sc}}\big(
\mathrm{X}_1^{\phantom{1}}\hat{\Gamma}_2^{\phantom{1}}\big)\stackrel{\mathrm{s.p.}}{=}
\frac{c_2^2}{\pi^2}
J_1[B]
+
\kappa_1S[B]+\mathcal{O}\big(L^{-2}\big),
\end{equation}
\begin{equation}\label{33-d-24}
\mathbb{H}_0^{\mathrm{sc}}\big(\mathrm{X}_1^{\phantom{1}}\overline{\Gamma}_2^{\phantom{1}}\big)
\stackrel{\mathrm{s.p.}}{=}\frac{c_2}{\pi}
\rho_3[B]+
\frac{c_2L_2}{2\pi^2}J_2[B]
+
\kappa_2S[B]+\mathcal{O}\big(L^{-1}\big),
\end{equation}
\begin{equation}\label{33-d-30}
\mathbb{H}_0^{\mathrm{sc}}\big(\Gamma_3^2\hat{\Gamma}_2^{\phantom{1}}\big)
\stackrel{\mathrm{s.p.}}{=}
\frac{6c_2}{\pi}\mathbb{H}_0^{\mathrm{sc}}\big(\Gamma_3^2\big)
-\frac{36c_2^2\theta_2}{\pi^2}J_1[B]
+\kappa_3S[B]+\mathcal{O}\big(L^{-1}\big),
\end{equation}
\begin{align}\label{33-d-25}
\mathbb{H}_0^{\mathrm{sc}}\big(\Gamma_3^2\overline{\Gamma}_2^{\phantom{1}}\big)
\stackrel{\mathrm{s.p.}}{=}\,&\frac{36L}{\pi}
\Big(\mathrm{A}_1-2\mathrm{A}_2+\mathrm{A}_4\Big)
-\frac{18Lc_2}{\pi}\rho_3\\\nonumber+\,&\,
9\mathrm{I}_1(\Lambda,\sigma)\Big(-c_2J_2[B]+2J_3[B]\Big)+\frac{18L\theta_2}{\pi^2}J_3[B]
\\\nonumber+\,&\,\kappa_4S[B]+\mathcal{O}(1)
,
\end{align}
wher
\begin{equation}\label{33-d-27}
\rho_3[B]=\int_{\mathbb{R}^2}\mathrm{d}^2x\,
\Big(f^{abc}PS_\Lambda^{bd}(x,x)f^{cde}\Big)PS_\Lambda^{eg}(x,x)B_\mu^{gh}(x)B_\mu^{ha}(x)
=\mathrm{B}_1-2\mathrm{B}_2+\mathrm{B}_3,
\end{equation}
and the functionals $J_1$ and $J_2$ were defined by the equalities from \eqref{33-d-28} and \eqref{33-d-29}. The coefficients $\{\kappa_1,\ldots,\kappa_4\}$ are singular in the regularizing parameter. The auxiliary number $\theta_2$ from formula \eqref{33-p-38-1} was also used. Note that it is easier to perform calculations for \eqref{33-d-30} and \eqref{33-d-25} using representation \eqref{33-w-4} for a triple vertex.

Summing up the expressions obtained, we come to the equality
\begin{align}\label{33-d-25-1}
\mathrm{H}_{334}^{1,3}
=\,&\frac{12c_2L_1}{\pi}\mathbb{H}_0^{\mathrm{sc}}\big(\Gamma_3^2\big)
\\\nonumber-\,&\,\frac{144L}{\pi}
\Big(\mathrm{A}_1-2\mathrm{A}_2+\mathrm{A}_4\Big)+
\frac{36c_2^2L_1\big(\ln(\Lambda/\sigma_2)-2\theta_2^{\phantom{1}}\big)}{\pi^2}
J_1[B]\\\nonumber+\,&\,
c_2\frac{18\big(L_1^2+2L\theta_1^{\phantom{1}}-2\ln(\Lambda/\sigma_2)L_2\big)}{\pi^2}J_2[B]-
\frac{36\big(L_1^2+2L\theta_1^{\phantom{1}}+2L\theta_2^{\phantom{1}}\big)}{\pi^2}J_3[B]
\\\nonumber+\,&\,\hat{\kappa}_1S[B]+\mathcal{O}(1)
,
\end{align}
where is the coefficient $\hat{\kappa}_1$ is singular in the regularizing parameter $\Lambda$.

\begin{lemma}\label{33-lem-12}
Taking into account all the above, the following decomposition is true for the first part of \eqref{33-d-2}
\begin{align}\label{33-d-25-2}
	\mathrm{H}_{334}^{1}
	\stackrel{\mathrm{s.p.}}{=}\,&
	18\ln(\Lambda/\sigma_2)\Lambda^2\mathbb{H}_0^{\mathrm{sc}}\big(\mathrm{X}_1^{\phantom{1}}
	\widetilde{\Gamma}_2^{\phantom{1}}\big)+\Lambda^2
	\mathbb{H}_0^{\mathrm{sc}}\big(\Gamma_3^2\widetilde{\Gamma}_2^{\phantom{1}}\big)+\frac{12c_2L_1}{\pi}\mathbb{H}_0^{\mathrm{sc}}\big(\Gamma_3^2\big)\\\nonumber+\,&\,3\ln(\Lambda/\sigma_3)
	\mathbb{H}_0^{\mathrm{sc}}\big(\Gamma_3^2\mathrm{X}_1^{\phantom{1}}\big)-18\ln(\Lambda/\sigma_2)\mathbb{H}_0^{\mathrm{sc}}\big(\mathrm{X}_1^{\phantom{1}}\Gamma_4^{\phantom{1}}\big)
	+54\ln(\Lambda/\sigma_2)\ln(\Lambda/\sigma_3)
	\mathbb{H}_0^{\mathrm{sc}}\big(\mathrm{X}_1^{2}\big)
	\\\nonumber-\,&\,\frac{144L}{\pi}
	\Big(\mathrm{A}_1-2\mathrm{A}_2+\mathrm{A}_4\Big)+
	\frac{36c_2^2L_1\big(\ln(\Lambda/\sigma_2)-2\theta_2^{\phantom{1}}\big)}{\pi^2}
	J_1[B]\\\nonumber+\,&\,
	c_2\frac{18\big(L_1^2+2L\theta_1^{\phantom{1}}-2\ln(\Lambda/\sigma_2)L_2\big)}{\pi^2}J_2[B]-
	\frac{36\big(L_1^2+2L\theta_1^{\phantom{1}}+2L\theta_2^{\phantom{1}}\big)}{\pi^2}J_3[B]
	+\hat{\kappa}_1S[B].
\end{align}
\end{lemma}

\subsubsection{Part without loop $\mathbb{H}_2^{\mathrm{c}}(\Gamma_4)$}
This type of diagram contains a large number of terms, so it is convenient to first transform the vertices and thus remove the equivalent diagrams. It is clear from the view that the quartic vertex $\Gamma_4$ connects by two external lines to one triple vertex, and the remaining two to the other. In this regard, it is reasonable to write out combinations with fixed external lines, taking into account the division into pairs relative to the vertical axis. In this case, we assume that the usual derivative is always on the left, since another option can be obtained by rearranging the triple vertices, which actually double the result. There are three such elements. They have the form
\begin{equation}\label{33-d-33}
{\centering\adjincludegraphics[width = 1.4 cm, valign=c]{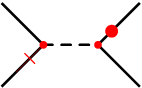}}+
{\centering\adjincludegraphics[width = 0.8 cm, valign=c]{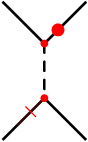}}-
{\centering\adjincludegraphics[width = 0.8 cm, valign=c]{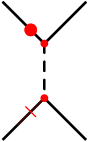}}.
\end{equation}
We transform the last vertex using the integration by parts
\begin{equation}\label{33-d-34}
-
{\centering\adjincludegraphics[width = 0.8 cm, valign=c]{fig/A7-003.eps}}=
{\centering\adjincludegraphics[width = 0.8 cm, valign=c]{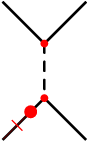}}+
{\centering\adjincludegraphics[width = 0.8 cm, valign=c]{fig/A7-002.eps}}+
{\centering\adjincludegraphics[width = 0.8 cm, valign=c]{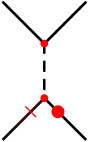}}.
\end{equation}
Next, in the last term, we use relation \eqref{33-eq:Jacobi} for structure constants and symmetry, which follows from the fact that the triple vertex is connected in all possible ways. We get equality
\begin{equation}\label{33-d-35}
{\centering\adjincludegraphics[width = 0.8 cm, valign=c]{fig/A7-005.eps}}=
{\centering\adjincludegraphics[angle=90, reflect,width = 1.4 cm, valign=c]{fig/A7-003.eps}}+
{\centering\adjincludegraphics[width = 1.2 cm, valign=c]{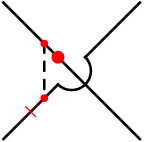}}
\longrightarrow-
{\centering\adjincludegraphics[width = 1.4 cm, valign=c]{fig/A7-001.eps}}+
{\centering\adjincludegraphics[width = 0.8 cm, valign=c]{fig/A7-002.eps}}.
\end{equation}
Finally, summing up all the calculations, we come to the next four-point vertex with fixed ends
\begin{equation}\label{33-d-36}
\Gamma_4\longrightarrow \Gamma_{4,1}+3\Gamma_{4,2}=
{\centering\adjincludegraphics[width = 0.8 cm, valign=c]{fig/A7-004.eps}}+
3{\centering\adjincludegraphics[width = 0.8 cm, valign=c]{fig/A7-002.eps}}.
\end{equation}
Further, returning to the discussion of the triple vertex in Section \ref{33:sec:cl:2}, we note that all possible permutations of the outer lines in $\Gamma_3$ can be represented in three ways.
\begin{enumerate}
	\item $\Gamma_3\to3\Gamma_{3,l1}-3\Gamma_{3,l2}+\Gamma_{3,l3}-\Gamma_{3,l4}$, where
	\begin{equation*}\label{33-d-37}
		\Gamma_{3,l1}=\,{\centering\adjincludegraphics[reflect, angle=180, width = 1 cm, valign=c]{fig/A7-007-01.eps}},\,\,\,
		\Gamma_{3,l2}=\,
		{\centering\adjincludegraphics[ width = 1 cm, valign=c]{fig/A7-007-01.eps}}
		,\,\,\,
		\Gamma_{3,l3}=
		\Bigg(
		{\centering\adjincludegraphics[width = 1 cm, valign=c]{fig/A7-014-01.eps}}-
		{\centering\adjincludegraphics[width = 1 cm, valign=c]{fig/A7-013-01.eps}}\,
		\Bigg),\,\,\,
		\Gamma_{3,l4}=
		\Bigg(
		{\centering\adjincludegraphics[reflect, angle=180,width = 1 cm, valign=c]{fig/A7-014-01.eps}}-
		{\centering\adjincludegraphics[reflect, angle=180,width = 1 cm, valign=c]{fig/A7-013-01.eps}}\,
		\Bigg).
	\end{equation*}
	\item $\Gamma_3\to3\Gamma_{3,r1}-3\Gamma_{3,r2}+\Gamma_{3,r3}-\Gamma_{3,r4}$, where
	\begin{equation*}\label{33-d-38}
		\Gamma_{3,r1}=\,{\centering\adjincludegraphics[reflect, width = 1 cm, valign=c]{fig/A7-007-01.eps}},\,\,\,
		\Gamma_{3,r2}=\,
		{\centering\adjincludegraphics[ angle=180, width = 1 cm, valign=c]{fig/A7-007-01.eps}}
		,\,\,\,
		\Gamma_{3,r3}=
		\Bigg(\,
		{\centering\adjincludegraphics[ angle=180,width = 1 cm, valign=c]{fig/A7-014-01.eps}}-
		{\centering\adjincludegraphics[ angle=180,width = 1 cm, valign=c]{fig/A7-013-01.eps}}
		\Bigg),\,\,\,
		\Gamma_{3,r4}=
		\Bigg(\,
		{\centering\adjincludegraphics[reflect,width = 1 cm, valign=c]{fig/A7-014-01.eps}}-
		{\centering\adjincludegraphics[reflect,width = 1 cm, valign=c]{fig/A7-013-01.eps}}
		\Bigg).
	\end{equation*}
	\item$\Gamma_3\to3\Gamma_{3,c1}-3\Gamma_{3,c2}+\Gamma_{3,c3}-\Gamma_{3,c4}$, where
	\begin{equation*}\label{33-d-39}
		\Gamma_{3,c1}=\,{\centering\adjincludegraphics[ width = 1 cm, valign=c]{fig/A7-011-01.eps}},\,\,\,
		\Gamma_{3,c2}=\,
		{\centering\adjincludegraphics[angle=180, reflect, width = 1 cm, valign=c]{fig/A7-009-01.eps}}
		,\,\,\,
		\Gamma_{3,c3}=
		\Bigg(
		{\centering\adjincludegraphics[width = 1 cm, valign=c]{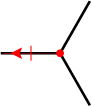}}-
		{\centering\adjincludegraphics[width = 1 cm, valign=c]{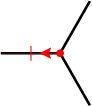}}\,
		\Bigg),\,\,\,
		\Gamma_{3,c4}=
		\Bigg(
		{\centering\adjincludegraphics[width = 1 cm, valign=c]{fig/A7-014-01.eps}}-
		{\centering\adjincludegraphics[width = 1 cm, valign=c]{fig/A7-013-01.eps}}\,
		\Bigg).
	\end{equation*}
\end{enumerate}
Thus, the diagrams under study are reduced to the study of the following two sets
\begin{equation}\label{33-d-40}
2\mathbb{\hat{H}}_0^{\mathrm{sc}}\Big((3\Gamma_{3,c1}-3\Gamma_{3,c2}+\Gamma_{3,c3}-\Gamma_{3,c4})\cdot\Gamma_{4,1}\cdot(3\Gamma_{3,r1}-3\Gamma_{3,r2}+\Gamma_{3,r3}-\Gamma_{3,r4})\Big),
\end{equation}
\begin{equation}\label{33-b-1}
6\mathbb{\hat{H}}_0^{\mathrm{sc}}\Big((3\Gamma_{3,l1}-3\Gamma_{3,l2}+\Gamma_{3,l3}-\Gamma_{3,l4})\cdot\Gamma_{4,2}\cdot(3\Gamma_{3,r1}-3\Gamma_{3,r2}+\Gamma_{3,r3}-\Gamma_{3,r4})\Big).
\end{equation}
Let us study each term separately.\\

\noindent\underline{Part with vertices $\Gamma_{3,c4}$ and $\Gamma_{4,1}$.} Let us connect these two vertices with fixed external lines into a new vertex, which is a component for \eqref{33-d-40}. It consists of two parts and has the form
\begin{equation}\label{33-d-41}
{\centering\adjincludegraphics[width = 2 cm, valign=c]{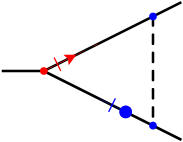}}-
{\centering\adjincludegraphics[width = 2 cm, valign=c]{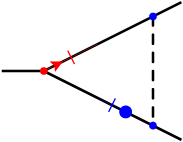}}.
\end{equation}
We note an important feature of this component. It contains three differential operators in the inner loop. This means that the rest of the diagram contains only one derivative. This fact makes it possible to isolate a local singularity by shifting a variable. Thus, from formula \eqref{33-d-41}, we can proceed to the sum of three terms of the form
\begin{equation}\label{33-d-42}
\frac{2\ln(\Lambda/\sigma_4)}{\pi}\Bigg(
{\centering\adjincludegraphics[width = 2 cm, valign=c]{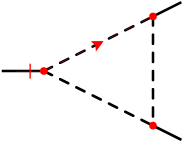}}+
{\centering\adjincludegraphics[width = 2 cm, valign=c]{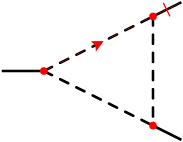}}+
{\centering\adjincludegraphics[width = 2 cm, valign=c]{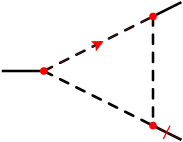}}\,\,\,\Bigg)=
-\frac{2\ln(\Lambda/\sigma_4)}{\pi}\Big(-c_2\Gamma_3-3\overline{\Gamma}_3\Big),
\end{equation}
where $\sigma_4>0$ is an auxiliary parameter. Note that the last vertices are connected in an arbitrary way to $\Gamma_3$. Therefore, we arrive at an equality of the form
\begin{multline}\label{33-d-43}
-2\mathbb{\hat{H}}_0^{\mathrm{sc}}\Big(\Gamma_{3,c4}\Gamma_{4,1}\big(3\Gamma_{3,r1}-3\Gamma_{3,r2}+\Gamma_{3,r3}-\Gamma_{3,r4}\big)\Big)\stackrel{\mathrm{s.p.}}{=}\\\stackrel{\mathrm{s.p.}}{=}
-\frac{4\ln(\Lambda/\sigma_4)}{\pi}\mathbb{\hat{H}}_0^{\mathrm{sc}}\Big(\big(c_2\Gamma_3+3\overline{\Gamma}_3\big)\big(3\Gamma_{3,r1}-3\Gamma_{3,r2}+\Gamma_{3,r3}-\Gamma_{3,r4}\big)\Big)
=\\=
-\frac{4\ln(\Lambda/\sigma_4)}{\pi}\mathbb{H}_0^{\mathrm{sc}}\big((c_2\Gamma_3+3\overline{\Gamma}_3)\Gamma_3\big).
\end{multline}
\noindent\underline{Part with vertices $\Gamma_{3,c3}$ and $\Gamma_{4,1}$.} This component does not contain singular contributions, that is
\begin{equation}\label{33-d-44}
2\mathbb{\hat{H}}_2^{\mathrm{sc}}\Big(\Gamma_{3,c3}\Gamma_{4,1}(3\Gamma_{3,r1}-3\Gamma_{3,r2}+\Gamma_{3,r3}-\Gamma_{3,r4})\Big)\stackrel{\mathrm{s.p.}}{=}0.
\end{equation}
We can verify this by performing similar calculations, as was done for the previous part, taking into account the antisymmetry of the structure constants. \\

\noindent\underline{Part with vertices $\Gamma_{3,c2}$ and $\Gamma_{4,1}$.} In this case, the calculation is more cumbersome, so it is convenient to divide the contribution into four parts
\begin{equation}\label{33-d-45}
\mathrm{H}_{334}^1=-18\mathbb{\hat{H}}_2^{\mathrm{sc}}\big(\Gamma_{3,c2}\Gamma_{4,1}\Gamma_{3,r1}\big),\,\,\,
\mathrm{H}_{334}^2=18\mathbb{\hat{H}}_2^{\mathrm{sc}}\big(\Gamma_{3,c2}\Gamma_{4,1}\Gamma_{3,r2}\big),
\end{equation}
\begin{equation}\label{33-d-46}
\mathrm{H}_{334}^3=-6\mathbb{\hat{H}}_2^{\mathrm{sc}}\big(\Gamma_{3,c2}\Gamma_{4,1}\Gamma_{3,r3}\big),\,\,\,
\mathrm{H}_{334}^4=6\mathbb{\hat{H}}_2^{\mathrm{sc}}\big(\Gamma_{3,c2}\Gamma_{4,1}\Gamma_{3,r4}\big).
\end{equation}
Consider the diagram $\mathrm{H}_{334}^1$. Taking into account the definitions from Section \ref{33:sec:int:1}, it is explicitly written as follows
\begin{equation}\label{33-d-47}
\mathrm{H}_{334}^1=
-18\,\,{\centering\adjincludegraphics[width = 3 cm, valign=c]{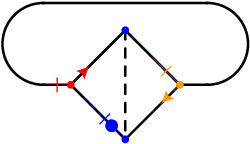}}.
\end{equation}
Note that the figure contains dots of three colors. This fact reflects the presence of three integration operators. It is convenient to analyze such a diagram using the method of addition and subtraction. We do all the calculations in detail, and for the following diagrams we will present only the answers. First, we write out the analytical form
\begin{equation}\label{33-d-48}
-18\int_{\mathbb{R}^{2\times3}}\mathrm{d}^2x\mathrm{d}^2y\mathrm{d}^2z
\begin{tabular}{ccc}
	$f^{a_1a_2a_3}$ & $\partial_{x_{\mu}}G_{\Lambda}^{a_1b_1}(x,y)$ &$f^{b_1b_4b_3}$ \\
	$B_\mu^{a_2c_2}(x)G_{\Lambda}^{c_2c_4}(x,z)$& $f^{c_1c_4c_3}$ & $\partial_{y_{\nu}}G_{\Lambda}^{c_3b_3}(x,y)$ \\
	$D_\rho^{d_4d_3}(z)\partial_{z_{\rho}}G_{\Lambda}^{a_3d_3}(x,z)$& $f^{d_4c_1d_2}$ & $G_{\Lambda}^{d_2b_2}(z,y)B_\nu^{b_4b_2}(y)$
\end{tabular}.
\end{equation}
In the last expression, a special arrangement of functions is used for clarity. In fact, it is simple multiplication. Next, we subtract and add an integral of the form 
\begin{equation*}\label{33-d-49}
\mathrm{\widetilde{H}}_{334}^{1}=
-144\int_{\mathbb{R}^{2\times3}}\mathrm{d}^2x\mathrm{d}^2y\mathrm{d}^2z
\begin{tabular}{ccc}
	$f^{a_1a_2a_3}$ & $\partial_{x_{\mu}}G^{\Lambda,\mathbf{f}}(x-y)$ &$f^{a_1b_4b_3}$ \\
	$B_\mu^{a_2c_2}(x)G_{\Lambda}^{c_2c_4}(x,z)$& $f^{c_1c_4b_3}$ & $\partial_{y_{\nu}}G^{\Lambda,\mathbf{f}}(x-y)$ \\
	$D_\rho^{d_4d_3}(z)\partial_{z_{\rho}}G_{\Lambda}^{a_3d_3}(x,z)$& $f^{d_4c_1d_2}$ & $\Big[G^{\Lambda,\mathbf{f}}(z-y)B_\nu^{b_4d_2}(z)$\\
	&&~~$+PS_{\Lambda}^{d_2b_2}(z,z)B_\nu^{b_4b_2}(z)\Big]$
\end{tabular}.
\end{equation*}
Then the part with the difference $\mathrm{H}_{334}^1-\mathrm{\widetilde{H}}_{334}^{1}$ and the contribution $\mathrm{\widetilde{H}}_{334}^{1}$ should be considered separately. Note that the first part can be represented as
\begin{equation}\label{33-d-52}
\mathrm{H}_{334}^1-\mathrm{\widetilde{H}}_{334}^{1}=
\int_{\mathbb{R}^{2\times2}}\mathrm{d}^2x\mathrm{d}^2z\,
G_{\Lambda}^{c_2c_4}(x,z)D_\rho^{d_4d_3}(z)\partial_{z_{\rho}}G_{\Lambda}^{a_3d_3}(x,z)
\rho_{4}^{c_2c_4,a_3d_4}(x,z).
\end{equation}
At the same time, the singularity is concentrated only in the multiplier
\begin{equation}\label{33-d-50}
G_{\Lambda}^{c_2c_4}(x,z)D_\rho^{d_4d_3}(z)\partial_{z_{\rho}}G_{\Lambda}^{a_3d_3}(x,z),
\end{equation}
because the rest of the internal singularities have been neutralized by subtraction. In other words, there is a finite limit to the function $\rho_{4}^{c_2c_4,a_3d_4}(x,z)$ for $z\to x$. Thus, using the result from Lemma \ref{33-lem-4}, the following substitution can be performed
\begin{align*}
G_{\Lambda}^{c_2c_4}(x,z)D_\rho^{d_4d_3}(z)\partial_{z_{\rho}}G_{\Lambda}^{a_3d_3}(x,z)
&\longrightarrow-2
G_{\Lambda}^{c_2c_4}(x,z)A^{d_4d_3}(x)G_{\Lambda}^{a_3d_3}(x,z)\\
&\longrightarrow-\frac{2\ln(\Lambda/\sigma_4)}{\pi}\delta_{c_2c_4}\delta_{d_4a_3}\delta(x-z),
\end{align*}
which leads to a relation of the form
\begin{equation}\label{33-d-53}
\mathrm{H}_{334}^1-\mathrm{\widetilde{H}}_{334}^{1}\stackrel{\mathrm{s.p.}}{=}
-\frac{2\ln(\Lambda/\sigma_4)}{\pi}
\int_{\mathbb{R}^{2}}\mathrm{d}^2x\,
\rho_{4}^{a_2a_2,a_3a_3}(x,x).
\end{equation}
Describing the latter relation in explicit terms, we obtain definitively
\begin{equation}\label{33-d-54}
\mathrm{H}_{334}^1-\mathrm{\widetilde{H}}_{334}^{1}\stackrel{\mathrm{s.p.}}{=}
\frac{36\ln(\Lambda/\sigma_4)}{\pi}\,\,{\centering\adjincludegraphics[width = 3 cm, valign=c]{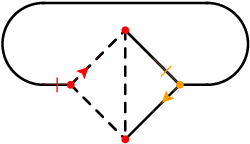}}+\kappa_5S[B],
\end{equation}
where $\kappa_5$ is a singular coefficient that does not depend on the background field. In the calculation, an auxiliary relation was also used
\begin{equation}\label{33-d-55}
{\centering\adjincludegraphics[width = 3 cm, valign=c]{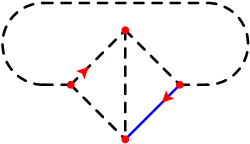}}=0.
\end{equation}
The rest part (the difference) can be reduced to the form\footnote{Here, the first term is zero because of \eqref{33-d-55}. However, the general view is useful for verifying calculations.} using the methods mentioned above
\begin{equation}\label{33-d-56}
\frac{\mathrm{\widetilde{H}}_{334}^{1}}{3^22^5}\stackrel{\mathrm{s.p.}}{=}
\big(\mathrm{I}_3(\Lambda,\sigma)+\mathrm{I}_4(\Lambda,\sigma)\big)
{\centering\adjincludegraphics[width = 3 cm, valign=c]{fig/A7-024.eps}}
-\frac{L}{2\pi}\,
{\centering\adjincludegraphics[width = 3 cm, valign=c]{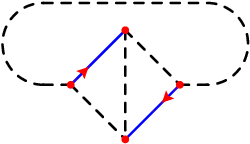}}+
\kappa_6S[B],
\end{equation}
where $\kappa_6$ is a coefficient singular in the parameter $\Lambda$, which also does not depend on the background field. The auxiliary integrals used have the form
\begin{align}\label{33-d-57}
\mathrm{I}_3(\Lambda,\sigma)&\,=\,
\int_{\mathbf{B}_{1/\Lambda}}\int_{\mathbf{B}_{1/\sigma}}\mathrm{d}^2x\mathrm{d}^2y\,
\Big(A_0(x)G^{\Lambda,\mathbf{f}}(x)\Big)G^{\Lambda,\mathbf{f}}(x)
\Big(\partial_{x^\mu}G^{\Lambda,\mathbf{f}}(x-y)\Big)
\partial_{y_\mu}G^{\Lambda,\mathbf{f}}(y)\\\noindent&\,=\,-
\int_{\mathbf{B}_{1/\Lambda}}\int_{\mathbf{B}_{1/\Lambda}}\mathrm{d}^2x\mathrm{d}^2y\,
\Big(A_0(x)G^{\Lambda,\mathbf{f}}(x)\Big)G^{\Lambda,\mathbf{f}}(x)
G^{\Lambda,\mathbf{f}}(x-y)
A_0(y)G^{\Lambda,\mathbf{f}}(y)\\\noindent&\stackrel{\mathrm{s.p.}}{=}\,-
\frac{4LL_2-2L^2+2L\theta_2}{8\pi^2}\stackrel{\mathrm{s.p.}}{=}-
\frac{L_1^2+2L\theta_1^{\phantom{1}}+L\theta_2^{\phantom{1}}}{4\pi^2}
,
\end{align}
\begin{align}\label{33-d-58}
\mathrm{I}_4(\Lambda,\sigma)&\,=\
\int_{\mathbf{B}_{1/\Lambda}}\int_{\mathbf{B}_{1/\sigma}}\mathrm{d}^2x\mathrm{d}^2y\,
\Big(A_0(x)G^{\Lambda,\mathbf{f}}(x)\Big)
\Big(\partial_{x^\mu}G^{\Lambda,\mathbf{f}}(x-y)\Big)G^{\Lambda,\mathbf{f}}(y)
\partial_{y_\mu}G^{\Lambda,\mathbf{f}}(y)\\\noindent&\,=\,-\frac{1}{2}
\int_{\mathbf{B}_{1/\Lambda}}\int_{\mathbf{B}_{1/\Lambda}}\mathrm{d}^2x\mathrm{d}^2y\,
\Big(A_0(x)G^{\Lambda,\mathbf{f}}(x)\Big)
\Big(G^{\Lambda,\mathbf{f}}(x-y)\Big)^2
A_0(x)G^{\Lambda,\mathbf{f}}(y)\\\noindent&\stackrel{\mathrm{s.p.}}{=}\,
-\frac{2LL_2-L^2+2L\theta_2}{8\pi^2}\stackrel{\mathrm{s.p.}}{=}
-\frac{L_1^2+2L\theta_1^{\phantom{1}}+2L\theta_2^{\phantom{1}}}{8\pi^2}
.
\end{align}
Finally, using relation \eqref{33-d-55}, we obtain the equality
\begin{equation}\label{33-d-59}
{\centering\adjincludegraphics[width = 3 cm, valign=c]{fig/A7-025.eps}}=
\frac{c_2}{2}\mathrm{B}_2-2\mathrm{A}_2+\mathrm{A}_3+\mathrm{A}_4,
\end{equation}
and also summing up the results of \eqref{33-d-54} and \eqref{33-d-56}, we get the answer in the form
\begin{equation}\label{33-d-60}
\mathrm{H}_{334}^1\stackrel{\mathrm{s.p.}}{=}
\frac{36\ln(\Lambda/\sigma_4)}{\pi}\,\,{\centering\adjincludegraphics[width = 3 cm, valign=c]{fig/A7-023.eps}}-\frac{72L}{\pi}\Big(c_2\mathrm{B}_2-4\mathrm{A}_2+2\mathrm{A}_3+2\mathrm{A}_4\Big)+(\kappa_5+\kappa_6)S[B].
\end{equation}
Consider the diagram $\mathrm{H}_{334}^2$. The analysis largely repeats previous calculations, so we only write out the final answer
\begin{align}\label{33-d-70}
\mathrm{H}_{334}^2\stackrel{\mathrm{s.p.}}{=}&-
\frac{36\ln(\Lambda/\sigma_4)}{\pi}\,\,{\centering\adjincludegraphics[width = 3 cm, valign=c]{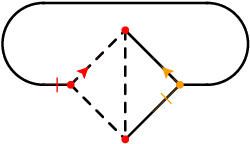}}+\frac{72c_2L}{\pi}\Big(\mathrm{B}_3-\mathrm{B}_2\Big)\\
\nonumber&+72
c_2^2J_1[B]\bigg(\frac{\ln(\Lambda/\sigma_4)L_2}{4\pi^2}+\mathrm{I}_3(\Lambda,\sigma)+\mathrm{I}_4(\Lambda,\sigma)\bigg)
+\kappa_7S[B],
\end{align}
where $\kappa_7$ is a coefficient depending on the parameter $\Lambda$.
Next, consider the remaining two contributions $\mathrm{H}_{334}^3$ and $\mathrm{H}_{334}^4$. They contain one less derivative, so it is enough to consider only the part of the diagram corresponding to the two connected vertices $\Gamma_{3,c2}\Gamma_{4,1}$. Calculations are performed using a shift of variables, and the formulas for the sum of $\mathrm{H}_{334}^3+\mathrm{H}_{334}^4$ take the form
\begin{equation*}
\mathrm{H}_{334}^3+\mathrm{H}_{334}^4=
-6\mathbb{\hat{H}}_0^{\mathrm{sc}}\Bigg(
{\centering\adjincludegraphics[width = 2 cm, valign=c]{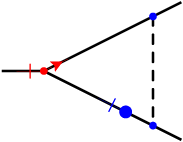}}\,
\Big(\Gamma_{3,r3}-\Gamma_{3,r4}\Big)\Bigg)
\stackrel{\mathrm{s.p.}}{=}
\frac{12\ln(\Lambda/\sigma_4)}{\pi}\mathbb{\hat{H}}_0^{\mathrm{sc}}\Bigg(
{\centering\adjincludegraphics[width = 2 cm, valign=c]{fig/A7-019.eps}}\,
\Big(\Gamma_{3,r3}-\Gamma_{3,r4}\Big)\Bigg).
\end{equation*}
Finally, by collecting all the parts of \eqref{33-d-45} and \eqref{33-d-46} and using the relation in the form
\begin{equation}\label{33-d-61}
\mathbb{\hat{H}}_0^{\mathrm{sc}}\Bigg(
{\centering\adjincludegraphics[width = 2 cm, valign=c]{fig/A7-019.eps}}\,
\Big(3\Gamma_{3,r1}-3\Gamma_{3,r2}+\Gamma_{3,r3}-\Gamma_{3,r4}\Big)\Bigg)=
\mathbb{H}_0^{\mathrm{sc}}\Bigg(
\Big(\overline{\Gamma}_3+\frac{c_2}{2}\Gamma_3\Big)\Gamma_3\Bigg)
,
\end{equation}
we get the following answer
\begin{align}\label{33-d-62}
-6\mathbb{\hat{H}}_2^{\mathrm{sc}}\Big(\Gamma_{3,c2}\Gamma_{4,1}(3\Gamma_{3,r1}-3\Gamma_{3,r2}&+\Gamma_{3,r3}-\Gamma_{3,r4})\Big)\stackrel{\mathrm{s.p.}}{=}
\frac{12\ln(\Lambda/\sigma_4)}{\pi}
\mathbb{H}_0^{\mathrm{sc}}\Bigg(
\Big(\overline{\Gamma}_3+\frac{c_2}{2}\Gamma_3\Big)\Gamma_3\Bigg)
\\\nonumber&
+\frac{144L}{\pi}\Big(2\mathrm{A}_2-\mathrm{A}_3-\mathrm{A}_4\Big)
-\frac{72c_2L}{\pi}\Big(2\mathrm{B}_2-\mathrm{B}_3\Big)
\\\nonumber&
-\frac{9
c_2^2}{\pi^2}J_1[B]
\Big(3L_1^2+6L\theta_1^{\phantom{1}}+4L\theta_2^{\phantom{1}}-2\ln(\Lambda/\sigma_4)L_2\Big)
+\hat{\kappa}_2S[B],
\end{align}
where $\hat{\kappa}_2$ is a singular coefficient that does not depend on the background field.\\

\noindent\underline{Part with vertices $\Gamma_{3,c1}$ and $\Gamma_{4,1}$.} This contribution is analyzed in a similar way. At the same time, its singular component is proportional to the classical action, that is
\begin{equation}\label{33-d-63}
6\mathbb{\hat{H}}_2^{\mathrm{sc}}\Big(\Gamma_{3,c2}\Gamma_{4,1}(3\Gamma_{3,r1}-3\Gamma_{3,r2}+\Gamma_{3,r3}-\Gamma_{3,r4})\Big)\stackrel{\mathrm{s.p.}}{=}\kappa_8S[B],
\end{equation}
where $\kappa_8$ is a singular coefficient.\\

\noindent\underline{Part with vertices \eqref{33-b-1}.} Further, using the linearity of the operator $\mathbb{\hat{H}}$, the set of diagrams can be divided into three parts
\begin{equation}\label{33-d-64}
\mathrm{H}_{334}^{5,1}=
6\mathbb{\hat{H}}_0^{\mathrm{sc}}\Big((\Gamma_{3,l3}-\Gamma_{3,l4})\Gamma_{4,2}(\Gamma_{3,r3}-\Gamma_{3,r4})\Big),
\end{equation}
\begin{equation}\label{33-d-65}
\mathrm{H}_{334}^{5,2}=
36\mathbb{\hat{H}}_0^{\mathrm{sc}}\Big((\Gamma_{3,l1}-\Gamma_{3,l2})\Gamma_{4,2}(\Gamma_{3,r3}-\Gamma_{3,r4})\Big),
\end{equation}
\begin{equation}\label{33-d-66}
\mathrm{H}_{334}^{5,3}=
54\mathbb{\hat{H}}_0^{\mathrm{sc}}\Big((\Gamma_{3,l1}-\Gamma_{3,l2})\Gamma_{4,2}(\Gamma_{3,r1}-\Gamma_{3,r2})\Big).
\end{equation}
The calculation of these three contributions largely repeats the calculations already performed. The main method of analysis is addition and subtraction, so here are just the answers
\begin{equation}\label{33-d-71}
\mathrm{H}_{334}^{5,1}\stackrel{\mathrm{s.p.}}{=}0,\,\,\,
\mathrm{H}_{334}^{5,2}\stackrel{\mathrm{s.p.}}{=}
\frac{18c_2\ln(\Lambda/\sigma_{5})}{\pi}\mathbb{\hat{H}}_0^{\mathrm{sc}}\Bigg({\centering\adjincludegraphics[ width = 1 cm, valign=c]{fig/A7-011-01.eps}}\,\,\big(\Gamma_{3,r3}-\Gamma_{3,r4}\big)\Bigg),
\end{equation}
\begin{align}\label{33-d-72}
\mathrm{H}_{334}^{5,3}\stackrel{\mathrm{s.p.}}{=}&\,
\frac{18c_2\ln(\Lambda/\sigma_{5})}{\pi}\mathbb{\hat{H}}_0^{\mathrm{sc}}\Bigg({\centering\adjincludegraphics[ width = 1 cm, valign=c]{fig/A7-011-01.eps}}\,\,\big(3\Gamma_{3,r1}-3\Gamma_{3,r2}\big)\Bigg)\\\nonumber
+&
\frac{2^33^3L}{\pi}\Big(c_2\mathrm{B}_2-4\mathrm{A}_2+2\mathrm{A}_3+2\mathrm{A}_4\Big)
\\\nonumber+&
54c_2^2J_1[B]\bigg(\frac{\ln(\Lambda/\sigma_{5})}{\pi^2}\Big(\ln(\Lambda/\sigma_{5})+L_2\Big)+8\mathrm{I}_2(\Lambda,\sigma)\bigg)+\kappa_9S[B],
\end{align}
where $\sigma_{5}$ is an auxiliary parameter, and $\kappa_9$ is a singular coefficient. The following notation for the integral was also used
\begin{align}\label{33-d-73}
\mathrm{I}_2(\Lambda,\sigma)&\,=\,\int_{\mathrm{B}_{1/\sigma}\times\mathrm{B}_{1/\sigma}}
\mathrm{d}^2x\mathrm{d}^2y\,\Big(\partial_{x_\mu}\partial_{x_\nu}G^{\Lambda,\mathbf{f}}(x)\Big)
G^{\Lambda,\mathbf{f}}(x-y)
\Big(\partial_{y^\mu}G^{\Lambda,\mathbf{f}}(y)\Big)
\Big(\partial_{y^\nu}G^{\Lambda,\mathbf{f}}(y)\Big)\\\nonumber&\stackrel{\mathrm{s.p.}}{=}\,
\int_{\mathrm{B}_{1/\sigma}\times\mathrm{B}_{1/\Lambda}}
\mathrm{d}^2x\mathrm{d}^2y\,\Big(\partial_{x_\nu}G^{\Lambda,\mathbf{f}}(x)\Big)
G^{\Lambda,\mathbf{f}}(x-y)
\Big(A_0(y)G^{\Lambda,\mathbf{f}}(y)\Big)
\Big(\partial_{y^\nu}G^{\Lambda,\mathbf{f}}(y)\Big)\\\nonumber&\,-\,\frac{1}{2}
\int_{\mathrm{B}_{1/\Lambda}\times\mathrm{B}_{1/\sigma}}
\mathrm{d}^2x\mathrm{d}^2y\,\Big(A_0(x)G^{\Lambda,\mathbf{f}}(x)\Big)
G^{\Lambda,\mathbf{f}}(x-y)
\Big(\partial_{y^\mu}G^{\Lambda,\mathbf{f}}(y)\Big)
\Big(\partial_{y_\mu}G^{\Lambda,\mathbf{f}}(y)\Big)\\\nonumber&\stackrel{\mathrm{s.p.}}{=}\,
-\frac{1}{2}
\int_{\mathrm{B}_{1/\Lambda}\times\mathrm{B}_{1/\sigma}}
\mathrm{d}^2x\mathrm{d}^2y\,\Big(A_0(x)G^{\Lambda,\mathbf{f}}(x)\Big)
G(x-y)
\Big(\partial_{y^\mu}G^{\Lambda,\mathbf{f}}(y)\Big)
\Big(\partial_{y_\mu}G^{\Lambda,\mathbf{f}}(y)\Big)
\\\nonumber&\,=\,
-\frac{1}{2}
\int_{\mathrm{B}_{1/\Lambda}\times\mathrm{B}_{1/\sigma}}
\mathrm{d}^2x\mathrm{d}^2y\,\Big(A_0(x)G^{\Lambda,\mathbf{f}}(x)\Big)
G^{1/|x|,\mathbf{0}}(y)
\Big(\partial_{y^\mu}G^{\Lambda,\mathbf{f}}(y)\Big)
\Big(\partial_{y_\mu}G^{\Lambda,\mathbf{f}}(y)\Big)
\\\nonumber&\stackrel{\mathrm{s.p.}}{=}\,
-\frac{1}{2}
\int_{\mathrm{B}_{1/\sigma}}\mathrm{d}^2y\,
G^{\Lambda,\mathbf{f}}(y)
\Big(\partial_{y^\mu}G^{\Lambda,\mathbf{f}}(y)\Big)
\Big(\partial_{y_\mu}G^{\Lambda,\mathbf{f}}(y)\Big)
\stackrel{\mathrm{s.p.}}{=}\,
-\frac{2LL_2-L^2}{16\pi^2}\stackrel{\mathrm{s.p.}}{=}-\frac{L_1^2+2L\theta_1^{\phantom{1}}}{16\pi^2}
,
\end{align}
where the calculation \eqref{33-w-7} was used in the last transition. Thus, summing up the found terms, we get
\begin{align}\label{33-d-72-1}
	\eqref{33-b-1}\stackrel{\mathrm{s.p.}}{=}&\,
	\frac{18c_2\ln(\Lambda/\sigma_{5})}{\pi}\mathbb{H}_0^{\mathrm{sc}}\big(\Gamma_3^2\big)
	+
	\frac{2^33^3L}{\pi}\Big(c_2\mathrm{B}_2-4\mathrm{A}_2+2\mathrm{A}_3+2\mathrm{A}_4\Big)
	\\\nonumber+&
	\frac{27c_2^2J_1[B]}{\pi^2}\Big(2\ln(\Lambda/\sigma_{5})(\ln(\Lambda/\sigma_{5})+L_2)-L_1^2-2L\theta_1^{\phantom{1}}\Big)+\kappa_9S[B].
\end{align}
\begin{lemma}\label{33-lem-6}Taking into account all the above, the following decomposition is true for the second part of \eqref{33-d-2}
\begin{align}\label{33-d-62-1}
\mathrm{H}_{334}^{1}\stackrel{\mathrm{s.p.}}{=}&
\frac{2c_2}{\pi}\big(\ln(\Lambda/\sigma_4)+9\ln(\Lambda/\sigma_{5})\big)
\mathbb{H}_0^{\mathrm{sc}}\big(\Gamma_3^2\big)
	\\\nonumber
	-\,&\frac{288L}{\pi}\Big(2\mathrm{A}_2-\mathrm{A}_3-\mathrm{A}_4\Big)
	+\frac{72c_2L}{\pi}\Big(\mathrm{B}_2+\mathrm{B}_3\Big)
	\\\nonumber
	+\,&\frac{9
	c_2^2J_1[B]}{\pi^2}
	\Big(-6L_1^2-12L\theta_1^{\phantom{1}}-4L\theta_2^{\phantom{1}}+2\big(\ln(\Lambda/\sigma_4)+3\ln(\Lambda/\sigma_5)\big)L_2+6\ln^2(\Lambda/\sigma_5)\Big)\\\nonumber
	+\,&\hat{\kappa}_3S[B],
\end{align}
where $\hat{\kappa}_3$ is a singular coefficient that does not depend on the background field.
\end{lemma}

\subsubsection{Sum of parts for $\mathbb{H}_0^{\mathrm{sc}}(\Gamma_3^2\Gamma_4^{\phantom{1}})$}
\label{33:sec:sum1}

\begin{lemma}\label{33-lem-13} Let $\{\sigma_1,\sigma_2,\sigma_3,\sigma_4\}$ be a set of positive parameters. The numbers $L_1$, $L_2$, and $\theta_2$ are defined according to the formulations from Section \ref{33:sec:cl:1}. Then, for the sum of the diagrams $\mathbb{H}_0^{\mathrm{sc}}(\Gamma_3^2\Gamma_4^{\phantom{1}})$ the following decomposition into singular components is valid
\begin{align}\nonumber
	\mathbb{H}_0^{\mathrm{sc}}&(\Gamma_3^2\Gamma_4^{\phantom{1}})\stackrel{\mathrm{s.p.}}{=}	18\ln(\Lambda/\sigma_2)\Lambda^2\mathbb{H}_0^{\mathrm{sc}}\big(\mathrm{X}_1^{\phantom{1}}
	\widetilde{\Gamma}_2^{\phantom{1}}\big)+
	\Lambda^2
	\mathbb{H}_0^{\mathrm{sc}}\big(\Gamma_3^2\widetilde{\Gamma}_2^{\phantom{1}}\big)+
	\frac{2c_2(\ln(\Lambda/\sigma_4)+9\ln(\Lambda/\sigma_{5})+6L_1)}{\pi}
	\mathbb{H}_0^{\mathrm{sc}}\big(\Gamma_3^2\big)
	\\\label{33-d-62-2}
	-&18\ln(\Lambda/\sigma_2)\mathbb{H}_0^{\mathrm{sc}}\big(\mathrm{X}_1^{\phantom{1}}\Gamma_4^{\phantom{1}}\big)+
	3\ln(\Lambda/\sigma_3)
	\mathbb{H}_0^{\mathrm{sc}}\big(\Gamma_3^2\mathrm{X}_1^{\phantom{1}}\big)
	+54\ln(\Lambda/\sigma_2)\ln(\Lambda/\sigma_3)
	\mathbb{H}_0^{\mathrm{sc}}\big(\mathrm{X}_1^{2}\big)
	\\\nonumber
	-&
	\frac{144L}{\pi}\Big(\mathrm{A}_1+2\mathrm{A}_2-2\mathrm{A}_3-\mathrm{A}_4\Big)+
	\frac{72c_2L}{\pi}\Big(\mathrm{B}_2+\mathrm{B}_3\Big)
	\\\nonumber
	+&\frac{18
		c_2^2J_1[B]}{\pi^2}
	\Big(-3L_1^2-6L\theta_1^{\phantom{1}}-6L\theta_2^{\phantom{1}}+2L_1\ln(\Lambda/\sigma_2)+\big(\ln(\Lambda/\sigma_4)+3\ln(\Lambda/\sigma_5)\big)L_2+3\ln^2(\Lambda/\sigma_5)\Big)\\\nonumber+&
	c_2\frac{18(L_1^2+2L\theta_1^{\phantom{1}}-2\ln(\Lambda/\sigma_2)L_2)}{\pi^2}J_2[B]-
	\frac{36(L_1^2+2L\theta_1^{\phantom{1}}+2L\theta_2^{\phantom{1}})}{\pi^2}J_3[B]
	+(\hat{\kappa}_1+\hat{\kappa}_3)S[B],
\end{align}
where $\hat{\kappa}_1$ and $\hat{\kappa}_3$ are singular coefficients by the parameter $\Lambda$ that are independent of the background field.
\end{lemma}

\subsection{The case $\mathbb{H}_0^{\mathrm{sc}}(\Gamma_4^2)$}
\label{33:sec:cl:5}

Let us move on to consider the following set of diagrams, which is constructed from two quartic vertices. Omitting all operator notation, such diagrams can be divided into two types
\begin{equation}\label{33-d-74}
{\centering\adjincludegraphics[ width = 2 cm, valign=c]{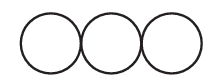}},\,\,\,
{\centering\adjincludegraphics[width = 1.4 cm, valign=c]{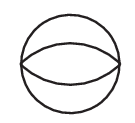}}.
\end{equation}
Thus, the required contribution is represented as the sum of two parts
\begin{equation}\label{33-d-75}
\mathbb{H}_0^{\mathrm{sc}}(\Gamma_4\Gamma_4)=
\mathrm{H}_{44}^1+\mathrm{H}_{44}^2.
\end{equation}
where the contribution $\mathrm{H}_{44}^1$ includes all diagrams of the type "caterpillar", and the contribution of $\mathrm{H}_{44}^2$ includes diagrams of the type "watermelon".
We study them separately.
\subsubsection{Part with loops $\mathbb{H}_2^{\mathrm{c}}(\Gamma_4)$}
The part with loops $\mathbb{H}_2^{\mathrm{c}}(\Gamma_4)$ according to the decomposition of \eqref{33-d-75} is denoted by $\mathrm{H}_{44}^1$. Note that the vertex $\mathbb{H}_2^{\mathrm{c}}(\Gamma_4) A $ with two outer lines can be represented as the sum of three parts, see \eqref{33-d-4}, so we can use the following decomposition
\begin{equation*}\label{33-d-75-1}
\mathbb{H}_0^{\mathrm{sc}}(\Gamma_4\Gamma_4)=
\mathbb{H}_0^{\mathrm{sc}}\big((\mathbb{H}_2^{\mathrm{c}}(\Gamma_4))^2\big)=
\mathbb{H}_0^{\mathrm{sc}}(\mathrm{D}_1^2)+
\mathbb{H}_0^{\mathrm{sc}}(\mathrm{D}_2^2)+
\mathbb{H}_0^{\mathrm{sc}}(\mathrm{D}_3^2)+
2\mathbb{H}_0^{\mathrm{sc}}(\mathrm{D}_1\mathrm{D}_2)+
2\mathbb{H}_0^{\mathrm{sc}}(\mathrm{D}_2\mathrm{D}_3)+
2\mathbb{H}_0^{\mathrm{sc}}(\mathrm{D}_3\mathrm{D}_1).
\end{equation*}
Let us consider them separately as well.\\

\noindent\underline{Part $\mathbb{H}_0^{\mathrm{sc}}(\mathrm{D}_2^2)$.} First, we transform the vertex $\mathrm{D}_2$. It is clear that under the sign of the operator $\mathbb{H}_0^{\mathrm{sc}}$, the vertices acquire a symmetrical appearance relative to the permutation of the outer lines, so we write a combination in which the derivative acts on the lower outer line
\begin{equation}\label{33-d-91}
\mathrm{D}_2=
{\centering\adjincludegraphics[ width = 0.8 cm, valign=c]{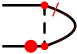}}+
{\centering\adjincludegraphics[reflect, angle=180, width = 0.8 cm, valign=c]{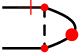}}+
{\centering\adjincludegraphics[reflect, angle=180,width = 1.3 cm, valign=c]{fig/A6-005.eps}}+
{\centering\adjincludegraphics[reflect, angle=180,width = 1.3 cm, valign=c]{fig/A6-006.eps}}.
\end{equation}
Next, we perform calculations similar to those in \eqref{33-d-16} and \eqref{33-d-17}, taking into account the result of Lemma \ref{33-lem-7}, then we obtain a decomposition of the form
\begin{equation}\label{33-d-92}
\mathrm{D}_2=3L_1\mathrm{X}_1+2\mathrm{\hat{D}}_2+\mathcal{O}\big(L/\Lambda^2\big),
\end{equation}
where the correction term is the finite sum of vertices with two external lines containing at most one derivative in total, and
\begin{equation}\label{33-d-90}
\mathrm{\hat{D}}_2=
{\centering\adjincludegraphics[ width = 0.8 cm, valign=c]{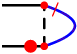}}+
{\centering\adjincludegraphics[ reflect, angle=180, width = 0.8 cm, valign=c]{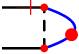}}+
{\centering\adjincludegraphics[ width = 1.3 cm, valign=c]{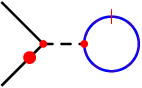}}+
{\centering\adjincludegraphics[ width = 1.3 cm, valign=c]{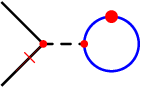}}.
\end{equation}
Therefore, by performing addition and subtraction, for the contribution of interest, we obtain
\begin{equation}\label{33-d-93}
\mathbb{H}_0^{\mathrm{sc}}(\mathrm{D}_2^2)\stackrel{\mathrm{s.p.}}{=}
-9L_1^2\mathbb{H}_0^{\mathrm{sc}}(\mathrm{X}_1^2)+
6L_1\mathbb{H}_0^{\mathrm{sc}}(\mathrm{D}_2\mathrm{X}_1)+4
\mathbb{H}_0^{\mathrm{sc}}(\mathrm{\hat{D}}_2^2).
\end{equation}
Let us consider the last term in more detail. Note that only logarithmic singularities of the first degree appear in it, so the vertices from \eqref{33-d-90} can be changed in the sense of equality of the singular parts in the final diagram as follows:
\begin{enumerate}
	\item replace the point $\bullet$ on the outer line with the dash $\mid\,$;
	\item rearrange the derivative from one external line to another with a minus sign.
\end{enumerate}
Therefore, if we consider the symmetric combination of the vertex \eqref{33-d-90}
\begin{equation}\label{33-d-94}
{\centering\adjincludegraphics[ width = 0.8 cm, valign=c]{fig/A8-019.eps}}+
{\centering\adjincludegraphics[ reflect, angle=180, width = 0.8 cm, valign=c]{fig/A8-002.eps}}+
{\centering\adjincludegraphics[reflect, angle=180, width = 0.8 cm, valign=c]{fig/A8-019.eps}}+
{\centering\adjincludegraphics[ width = 0.8 cm, valign=c]{fig/A8-002.eps}}+
{\centering\adjincludegraphics[ width = 1.3 cm, valign=c]{fig/A8-026.eps}}+
{\centering\adjincludegraphics[ width = 1.3 cm, valign=c]{fig/A8-004.eps}}+
{\centering\adjincludegraphics[ reflect, angle=180,width = 1.3 cm, valign=c]{fig/A8-026.eps}}+
{\centering\adjincludegraphics[ reflect, angle=180,width = 1.3 cm, valign=c]{fig/A8-004.eps}},
\end{equation}
then, for example, the second and fourth terms, taking into account the use of identity \eqref{33-eq:Jacobi} for the structure constants, can be transformed as follows
\begin{equation}\label{33-d-95}
{\centering\adjincludegraphics[ reflect, angle=180, width = 0.8 cm, valign=c]{fig/A8-002.eps}}+
{\centering\adjincludegraphics[ width = 0.8 cm, valign=c]{fig/A8-002.eps}}\to
{\centering\adjincludegraphics[ reflect, angle=180, width = 0.8 cm, valign=c]{fig/A8-002.eps}}-
{\centering\adjincludegraphics[  reflect, angle=180,width = 0.8 cm, valign=c]{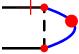}}
=
{\centering\adjincludegraphics[width = 1.3 cm, valign=c]{fig/A8-004.eps}}.
\end{equation}
Applying similar transformations to the rest of the diagrams, we obtain that the combination from \eqref{33-d-94} can be rewritten as
\begin{equation}\label{33-d-96}
3{\centering\adjincludegraphics[reflect, angle=180,width = 1.3 cm, valign=c]{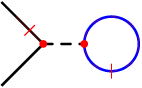}}+3
{\centering\adjincludegraphics[width = 1.3 cm, valign=c]{fig/A8-004.eps}}=
6{\centering\adjincludegraphics[reflect, angle=180,width = 1.3 cm, valign=c]{fig/A8-003.eps}}+
3{\centering\adjincludegraphics[reflect, angle=180,width = 1.3 cm, valign=c]{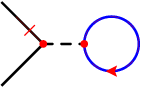}}
=\mathrm{\widetilde{D}}_2.
\end{equation}
Thus, after a series of calculations taking into account the result of Lemma \ref{33-lem-2}, we have
\begin{equation}\label{33-d-97}
4\mathbb{H}_0^{\mathrm{sc}}(\mathrm{\hat{D}}_2^2)\stackrel{\mathrm{s.p.}}{=}
2\mathbb{\hat{H}}_0^{\mathrm{sc}}(\mathrm{\widetilde{D}}_2^2)\stackrel{\mathrm{s.p.}}{=}
-\frac{72c_2L}{\pi}\mathrm{B}_5-
\frac{72c_2L}{\pi}\mathrm{B}_4+
\frac{18c_2L}{\pi}\mathrm{B}_3,
\end{equation}
where the notation from \eqref{33-s-9} was used, as well as two functions from Section \ref{33:sec:dia}.
Finally, the answer for the contribution under study takes the form
\begin{equation}\label{33-d-93-1}
	\mathbb{H}_0^{\mathrm{sc}}(\mathrm{D}_2^2)\stackrel{\mathrm{s.p.}}{=}
	-9L_1^2\mathbb{H}_0^{\mathrm{sc}}(\mathrm{X}_1^2)+
	6L_1\mathbb{H}_0^{\mathrm{sc}}(\mathrm{D}_2\mathrm{X}_1)+\frac{18c_2L}{\pi}\Big(-4\mathrm{B}_5-4
	\mathrm{B}_4+\mathrm{B}_3\Big).
\end{equation}

\noindent\underline{Part $2\mathbb{H}_0^{\mathrm{sc}}(\mathrm{D}_1\mathrm{D}_2)$.} To begin with, consider the vertex $\mathrm{D}_1$, taking into account the symmetrization\footnote{This is exactly the combination that appears when using the operator $\mathbb{H}_0^{\mathrm{sc}}$.} of external lines. Using the integration by parts, the following decompositions can be obtained
\begin{align}\label{33-d-76}
-{\centering\adjincludegraphics[ width = 0.8 cm, valign=c]{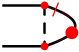}}-
{\centering\adjincludegraphics[ reflect, angle=180, width = 0.8 cm, valign=c]{fig/A8-009.eps}}&\,=\,
{\centering\adjincludegraphics[ reflect, angle=180, width = 0.85 cm, valign=c]{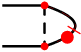}}+
{\centering\adjincludegraphics[ reflect, angle=180, width = 0.8 cm, valign=c]{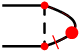}}+
{\centering\adjincludegraphics[ width = 0.8 cm, valign=c]{fig/A8-017.eps}}+
{\centering\adjincludegraphics[ reflect, angle=180, width = 0.8 cm, valign=c]{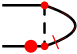}}+
{\centering\adjincludegraphics[ reflect, angle=180, width = 0.8 cm, valign=c]{fig/A8-011.eps}}+
{\centering\adjincludegraphics[ reflect, angle=180, width = 0.8 cm, valign=c]{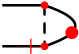}}
\\\nonumber
&\,=\,
{\centering\adjincludegraphics[ width = 0.85 cm, valign=c]{fig/A8-022.eps}}+
{\centering\adjincludegraphics[ width = 0.8 cm, valign=c]{fig/A8-010.eps}}+
{\centering\adjincludegraphics[reflect, angle=180, width = 0.8 cm, valign=c]{fig/A8-017.eps}}+
{\centering\adjincludegraphics[ width = 0.8 cm, valign=c]{fig/A8-018.eps}}+
{\centering\adjincludegraphics[ width = 0.8 cm, valign=c]{fig/A8-011.eps}}+
{\centering\adjincludegraphics[ width = 0.8 cm, valign=c]{fig/A8-012.eps}}.
\end{align}
Further, using the result of Lemma \ref{33-lem-7}, we can verify the validity of the following set of equalities
\begin{align}\label{33-d-78-1}
{\centering\adjincludegraphics[ width = 0.8 cm, valign=c]{fig/A8-017.eps}}+
{\centering\adjincludegraphics[ width = 0.8 cm, valign=c]{fig/A8-018.eps}}&\,=\,
2\,{\centering\adjincludegraphics[ width = 0.8 cm, valign=c]{fig/A8-019.eps}}+
2\,{\centering\adjincludegraphics[ width = 0.8 cm, valign=c]{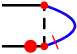}}+\mathcal{O}\big(L/\Lambda^2\big),
\\\nonumber
{\centering\adjincludegraphics[ reflect, angle=180, width = 0.8 cm, valign=c]{fig/A8-018.eps}}+
{\centering\adjincludegraphics[reflect, angle=180, width = 0.8 cm, valign=c]{fig/A8-017.eps}}
&\,=\,
2\,{\centering\adjincludegraphics[reflect, angle=180, width = 0.8 cm, valign=c]{fig/A8-019.eps}}+
2\,{\centering\adjincludegraphics[reflect, angle=180, width = 0.8 cm, valign=c]{fig/A8-020.eps}}+\mathcal{O}\big(L/\Lambda^2\big),
\\\nonumber
{\centering\adjincludegraphics[ reflect, angle=180, width = 0.8 cm, valign=c]{fig/A8-011.eps}}+
{\centering\adjincludegraphics[ width = 0.8 cm, valign=c]{fig/A8-012.eps}}
&\,=\,
2\,{\centering\adjincludegraphics[reflect, angle=180, width = 0.8 cm, valign=c]{fig/A8-005.eps}}+
2\,{\centering\adjincludegraphics[reflect, angle=180, width = 0.8 cm, valign=c]{fig/A8-002.eps}}+\mathcal{O}\big(L/\Lambda^2\big),
\\\nonumber
{\centering\adjincludegraphics[ reflect, angle=180, width = 0.8 cm, valign=c]{fig/A8-012.eps}}+
{\centering\adjincludegraphics[ width = 0.8 cm, valign=c]{fig/A8-011.eps}}
&\,=\,
2\,{\centering\adjincludegraphics[ width = 0.8 cm, valign=c]{fig/A8-005.eps}}+
2\,{\centering\adjincludegraphics[ width = 0.8 cm, valign=c]{fig/A8-002.eps}}+\mathcal{O}\big(L/\Lambda^2\big),
\end{align}
where the correction terms are vertices with two external lines. In this case, the outer lines can contain no more than one derivative in total. Next, paying attention to the relation
\begin{equation}\label{33-d-78}
\partial_{x_\nu}D_\nu^{ad}(x)+D_\nu^{ad}(x)\partial_{x_\nu}=\,\,
\mid\bullet+\bullet\mid\,\,=-4A^{ad}(x),
\end{equation}
the definition for the vertex $\widetilde{\Gamma}_2$ from \eqref{33-d-14}, as well as the decomposition
\begin{equation}\label{33-d-79}
\frac{1}{2}\Big(A^{ad}(x)G_\Lambda^{db}(x,y)\big|_{y=x}+A^{bd}(x)G_\Lambda^{da}(x,y)\big|_{y=x}\Big)=
\frac{\Lambda^2\alpha_1(\mathbf{f})}{4\pi}\delta^{ed}-\frac{\theta_2}{4\pi}B_\mu^{ad}(x)B_\mu^{db}(x)
+o(1),
\end{equation}
derived in \cite{Iv-244}, we obtain equality with fixed external lines of the form\footnote{The vertices $\widetilde{\Gamma}_2$ and $\mathrm{X}_2$ are symmetric with respect to the permutation of the outer lines, so the fixation option is not important.}
\begin{equation}\label{33-d-81}
-{\centering\adjincludegraphics[ width = 0.8 cm, valign=c]{fig/A8-009.eps}}-
{\centering\adjincludegraphics[ reflect, angle=180, width = 0.8 cm, valign=c]{fig/A8-009.eps}}=
2\Lambda^2\widetilde{\Gamma}_2-2\theta_2\mathrm{X}_2
+2\mathrm{\widetilde{D}}_1
+ L\Lambda^{-2}\Theta_1+\Lambda^{-1}\Theta_2,
\end{equation}
where
\begin{equation}\label{33-d-99}
\mathrm{\widetilde{D}}_1=
-{\centering\adjincludegraphics[ width = 0.8 cm, valign=c]{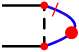}}
-{\centering\adjincludegraphics[reflect, angle=180, width = 0.8 cm, valign=c]{fig/A8-025.eps}}
-\frac{1}{2}\bigg({\centering\adjincludegraphics[ width = 0.8 cm, valign=c]{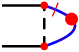}}
+{\centering\adjincludegraphics[ width = 0.8 cm, valign=c]{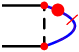}}\bigg)
-\frac{1}{2}\bigg({\centering\adjincludegraphics[ reflect, angle=180,width = 0.8 cm, valign=c]{fig/A8-023.eps}}
+{\centering\adjincludegraphics[reflect, angle=180, width = 0.8 cm, valign=c]{fig/A8-024.eps}}\bigg).
\end{equation}
Here $\Theta_1$ is the sum of the correction terms from \eqref{33-d-78-1}, and the term $\Theta_2$ is a vertex with two outer lines without derivatives and a finite density.

Then, using the linearity of the operator $\mathbb{H}_0^{\mathrm{sc}}$ and the previously obtained representation \eqref{33-d-92} for the vertex $\mathrm{D}_2$, we rewrite the contribution as follows
\begin{equation}\label{33-d-82}
2\mathbb{H}_0^{\mathrm{sc}}(\mathrm{D}_1\mathrm{D}_2)=
6L_1\mathbb{H}_0^{\mathrm{sc}}(\mathrm{D}_1\mathrm{X}_1)+
2\mathbb{H}_0^{\mathrm{sc}}\big(\mathrm{D}_1(\mathrm{D}_2-3L_1\mathrm{X}_1)\big).
\end{equation}
Further, substituting the expansion from \eqref{33-d-81}, we note that the equality is valid
\begin{equation}\label{33-d-83}
\mathbb{H}_0^{\mathrm{sc}}\big(\mathrm{D}_1(\mathrm{D}_2-3L_1\mathrm{X}_1)\big)\stackrel{\mathrm{s.p.}}{=}\Lambda^2\mathbb{H}_0^{\mathrm{sc}}\big(\widetilde{\Gamma}_2(\mathrm{D}_2-3L_1\mathrm{X}_1)\big).
\end{equation}
Thus, we get
\begin{equation}\label{33-d-84}
2\mathbb{H}_0^{\mathrm{sc}}(\mathrm{D}_1\mathrm{D}_2)\stackrel{\mathrm{s.p.}}{=}
6L_1\mathbb{H}_0^{\mathrm{sc}}(\mathrm{D}_1\mathrm{X}_1)-
6L_1\Lambda^2\mathbb{H}_0^{\mathrm{sc}}\big(\widetilde{\Gamma}_2\mathrm{X}_1\big)+
2\Lambda^2\mathbb{H}_0^{\mathrm{sc}}\big(\widetilde{\Gamma}_2\mathrm{D}_2\big).
\end{equation}
\noindent\underline{Part $2\mathbb{H}_0^{\mathrm{sc}}(\mathrm{D}_3\mathrm{D}_2)$.} Let us start considering the contribution again with auxiliary transformations. We choose a symmetric combination of the vertex $\mathrm{D}_3$, then after integration by parts we get
\begin{equation}\label{33-d-85}
-{\centering\adjincludegraphics[ width = 0.8 cm, valign=c]{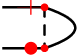}}-
{\centering\adjincludegraphics[ reflect, angle=180, width = 0.8 cm, valign=c]{fig/A8-015.eps}}=
{\centering\adjincludegraphics[ width = 0.85 cm, valign=c]{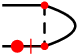}}+
{\centering\adjincludegraphics[ width = 0.8 cm, valign=c]{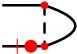}}+
{\centering\adjincludegraphics[ width = 0.8 cm, valign=c]{fig/A8-017.eps}}+
{\centering\adjincludegraphics[ width = 0.8 cm, valign=c]{fig/A8-018.eps}}+
{\centering\adjincludegraphics[ width = 0.8 cm, valign=c]{fig/A8-012.eps}}+
{\centering\adjincludegraphics[ reflect, angle=180, width = 0.8 cm, valign=c]{fig/A8-011.eps}}.
\end{equation}
Note that in the same way, on the right side, we can get a combination that is symmetrical with respect to the horizontal axis. We continue to use this fact by constructing the Laplace operator from the convenient side. Next, we note that the relations are correct
\begin{equation}\label{33-d-86}
{\centering\adjincludegraphics[ width = 0.8 cm, valign=c]{fig/A8-017.eps}}+
{\centering\adjincludegraphics[ width = 0.8 cm, valign=c]{fig/A8-018.eps}}=
2\,{\centering\adjincludegraphics[ width = 0.8 cm, valign=c]{fig/A8-019.eps}}+
2\,{\centering\adjincludegraphics[ width = 0.8 cm, valign=c]{fig/A8-020.eps}}
+\mathcal{O}\big(L/\Lambda^2\big),
\end{equation}
\begin{equation}\label{33-d-87}
{\centering\adjincludegraphics[ width = 0.8 cm, valign=c]{fig/A8-012.eps}}+
{\centering\adjincludegraphics[ reflect, angle=180, width = 0.8 cm, valign=c]{fig/A8-011.eps}}=
2\,{\centering\adjincludegraphics[ reflect, angle=180, width = 0.8 cm, valign=c]{fig/A8-005.eps}}+
2\,{\centering\adjincludegraphics[ reflect, angle=180, width = 0.8 cm, valign=c]{fig/A8-002.eps}}
+\mathcal{O}\big(L/\Lambda^2\big),
\end{equation}
where the correction terms are vertices with two external lines containing at most one derivative in total. Thus, we get
\begin{equation}\label{33-d-103}
-{\centering\adjincludegraphics[ width = 0.8 cm, valign=c]{fig/A8-015.eps}}-
{\centering\adjincludegraphics[ reflect, angle=180, width = 0.8 cm, valign=c]{fig/A8-015.eps}}=
{\centering\adjincludegraphics[ width = 0.85 cm, valign=c]{fig/A8-016.eps}}+
{\centering\adjincludegraphics[ width = 0.8 cm, valign=c]{fig/A8-021.eps}}+
2\mathrm{\widetilde{D}}_3+\mathcal{O}\big(L/\Lambda^2\big),
\end{equation}
where
\begin{equation}\label{33-d-104}
\mathrm{\widetilde{D}}_3=
{\centering\adjincludegraphics[ width = 0.8 cm, valign=c]{fig/A8-019.eps}}+
{\centering\adjincludegraphics[ width = 0.8 cm, valign=c]{fig/A8-020.eps}}+
{\centering\adjincludegraphics[ reflect, angle=180, width = 0.8 cm, valign=c]{fig/A8-005.eps}}+
{\centering\adjincludegraphics[ reflect, angle=180, width = 0.8 cm, valign=c]{fig/A8-002.eps}}.
\end{equation}
Now let us go back to the original combination and rewrite it as follows
\begin{equation}\label{33-d-88}
2\mathbb{H}_0^{\mathrm{sc}}(\mathrm{D}_3\mathrm{D}_2)=
6L_1\mathbb{H}_0^{\mathrm{sc}}(\mathrm{D}_3\mathrm{X}_1)+
2\mathbb{H}_0^{\mathrm{sc}}\big(\mathrm{D}_3(\mathrm{D}_2-3L_1\mathrm{X}_1)\big).
\end{equation}
Further, substituting expansion \eqref{33-d-92} for $\mathrm{D}_2-3L_1\mathrm{X}_1$ and using formula \eqref{33-d-85}, the second term in \eqref{33-d-88} can be rewritten as the sum of two parts
\begin{equation}\label{33-d-89}
2\mathbb{H}_0^{\mathrm{sc}}\big(\mathrm{D}_3(\mathrm{D}_2-3L_1\mathrm{X}_1)\big)=
\mathrm{H}_{44}^{1,1}+\mathrm{H}_{44}^{1,2},
\end{equation}
where
\begin{equation}\label{33-d-100}
\mathrm{H}_{44}^{1,1}=8\mathbb{\hat{H}}_0^{\mathrm{sc}}\Bigg(\bigg(
{\centering\adjincludegraphics[angle=180, width = 0.8 cm, valign=c]{fig/A8-019.eps}}+
{\centering\adjincludegraphics[ reflect, width = 0.8 cm, valign=c]{fig/A8-002.eps}}+
{\centering\adjincludegraphics[angle=180, width = 1.3 cm, valign=c]{fig/A8-026.eps}}+
{\centering\adjincludegraphics[angle=180, width = 1.3 cm, valign=c]{fig/A8-004.eps}}\,\bigg)
\bigg({\centering\adjincludegraphics[ width = 0.8 cm, valign=c]{fig/A8-019.eps}}+
{\centering\adjincludegraphics[ width = 0.8 cm, valign=c]{fig/A8-020.eps}}+
{\centering\adjincludegraphics[ reflect, angle=180, width = 0.8 cm, valign=c]{fig/A8-005.eps}}+
{\centering\adjincludegraphics[ reflect, angle=180, width = 0.8 cm, valign=c]{fig/A8-002.eps}}\bigg)
\Bigg),
\end{equation}
\begin{equation}\label{33-d-101}
\mathrm{H}_{44}^{1,2}=4\mathbb{\hat{H}}_0^{\mathrm{sc}}\Bigg(\bigg(
{\centering\adjincludegraphics[angle=180, width = 0.8 cm, valign=c]{fig/A8-019.eps}}+
{\centering\adjincludegraphics[ reflect, width = 0.8 cm, valign=c]{fig/A8-002.eps}}+
{\centering\adjincludegraphics[angle=180, width = 1.3 cm, valign=c]{fig/A8-026.eps}}+
{\centering\adjincludegraphics[angle=180, width = 1.3 cm, valign=c]{fig/A8-004.eps}}\,\bigg)
\bigg({\centering\adjincludegraphics[ width = 0.85 cm, valign=c]{fig/A8-016.eps}}+
{\centering\adjincludegraphics[ width = 0.8 cm, valign=c]{fig/A8-021.eps}}\bigg)
\Bigg).
\end{equation}
By performing calculations similar to those above, we obtain the following equalities
\begin{equation}\label{33-d-102}
\mathrm{H}_{44}^{1,1}\stackrel{\mathrm{s.p.}}{=}
-\frac{4L}{\pi}\Big(c_2\mathrm{B}_3-8\mathrm{A}_2+4\mathrm{A}_3+4\mathrm{A}_4+
8\mathrm{A}_5+
8\mathrm{A}_6
-8\mathrm{A}_7
-8\mathrm{A}_8
\Big)+\kappa_{10}S[B],
\end{equation}
\begin{align}\label{33-d-105}
\mathrm{H}_{44}^{1,2}\stackrel{\mathrm{s.p.}}{=}
&-\frac{8L}{\pi}\Big(c_2\mathrm{B}_3+8c_2\mathrm{B}_4+8c_2\mathrm{B}_5-8c_2\mathrm{B}_6+8c_2\mathrm{B}_7-2\mathrm{A}_2+4\mathrm{A}_3-2\mathrm{A}_4\Big)\\\nonumber&-\frac{24c_2^2L_1}{\pi}
\Bigg(\frac{4\theta_2-2\theta_1}{2\pi}J_4[B]-
\frac{4\theta_2+2L_1}{4\pi}J_1[B]\Bigg)+\kappa_{11}S[B]
,
\end{align}
where the definitions from \eqref{33-d-28} and \eqref{33-d-28-1} were used in the answer, as well as diagrams from Section \ref{33:sec:dia}. Also $\kappa_{10}$ and $\kappa_{11}$ are singular coefficients independent of the background field. Thus, we get the final answer in the form
\begin{align}\label{33-d-105-1}
2\mathbb{H}_0^{\mathrm{sc}}(\mathrm{D}_3\mathrm{D}_2)\stackrel{\mathrm{s.p.}}{=}&6L_1\mathbb{H}_0^{\mathrm{sc}}(\mathrm{D}_3\mathrm{X}_1)+(\kappa_{10}+\kappa_{11})S[B]
	\\\nonumber
	&-\frac{8L}{\pi}\Big(-6\mathrm{A}_2+6\mathrm{A}_3+4\mathrm{A}_5+
	4\mathrm{A}_6
	-4\mathrm{A}_7
	-4\mathrm{A}_8\Big)
	\\\nonumber
	&-\frac{8Lc_2}{\pi}\Bigg(\frac{3}{2}\mathrm{B}_3+8\mathrm{B}_4+8\mathrm{B}_5-8\mathrm{B}_6+8\mathrm{B}_7\Bigg)\\\nonumber&-\frac{6c_2^2L_1}{\pi^2}
	\Bigg(\Big(8\theta_2-4\theta_1\Big)J_4[B]-
	\Big(4\theta_2+2L_1\Big)J_1[B]\Bigg).
\end{align}

\noindent\underline{Part $\mathbb{H}_0^{\mathrm{sc}}(\mathrm{D}_1^2)$.} In this case, it is convenient to use the decompositions from \eqref{33-d-79} and \eqref{33-d-81}. Then, taking into account the relation
\begin{equation}\label{33-d-108}
\mathbb{H}_0^{\mathrm{sc}}\big((\mathrm{D}_1-\Lambda^2\widetilde{\Gamma}_2)^2\big)\stackrel{\mathrm{s.p.}}{=}0,
\end{equation}
after the appropriate addition and subtraction, we get
\begin{equation}\label{33-d-109}
\mathbb{H}_0^{\mathrm{sc}}(\mathrm{D}_1^2)\stackrel{\mathrm{s.p.}}{=}
-\Lambda^4\mathbb{H}_0^{\mathrm{sc}}(\widetilde{\Gamma}_2\widetilde{\Gamma}_2)+
2\Lambda^2\mathbb{H}_0^{\mathrm{sc}}(\mathrm{D}_1\widetilde{\Gamma}_2).
\end{equation}

\noindent\underline{Part $2\mathbb{H}_0^{\mathrm{sc}}(\mathrm{D}_1\mathrm{D}_3)$.} Let us apply the expansions from \eqref{33-d-79} and \eqref{33-d-103}, then, using addition and subtraction again, we can write out the relation
\begin{equation}\label{33-d-110}
2\mathbb{H}_0^{\mathrm{sc}}(\mathrm{D}_1\mathrm{D}_3)=
2\Lambda^2\mathbb{H}_0^{\mathrm{sc}}(\widetilde{\Gamma}_2\mathrm{D}_3)+
2\mathbb{H}_0^{\mathrm{sc}}\big((\mathrm{D}_1-\Lambda^2\widetilde{\Gamma}_2)\mathrm{D}_3\big).
\end{equation}
In this case, the last term, in turn, can be decomposed as
\begin{equation}\label{33-d-111}
2\mathbb{H}_0^{\mathrm{sc}}\big((\mathrm{D}_1-\Lambda^2\widetilde{\Gamma}_2)\mathrm{D}_3\big)=
\mathrm{H}_{44}^{1,3}+
\mathrm{H}_{44}^{1,4}+\mathcal{O}(1),
\end{equation}
where
\begin{equation}\label{33-d-112}
\mathrm{H}_{44}^{1,3}=2
\mathbb{\hat{H}}_0^{\mathrm{sc}}
\Bigg(
\bigg(
{\centering\adjincludegraphics[angle=180, width = 0.85 cm, valign=c]{fig/A8-016.eps}}+
{\centering\adjincludegraphics[angle=180, width = 0.8 cm, valign=c]{fig/A8-021.eps}}
\bigg)\bigg(
-{\centering\adjincludegraphics[ width = 0.8 cm, valign=c]{fig/A8-025.eps}}
-{\centering\adjincludegraphics[reflect, angle=180, width = 0.8 cm, valign=c]{fig/A8-025.eps}}
-\frac{1}{2}\bigg[{\centering\adjincludegraphics[ width = 0.8 cm, valign=c]{fig/A8-023.eps}}
+{\centering\adjincludegraphics[ width = 0.8 cm, valign=c]{fig/A8-024.eps}}\bigg]
-\frac{1}{2}\bigg[{\centering\adjincludegraphics[ reflect, angle=180,width = 0.8 cm, valign=c]{fig/A8-023.eps}}
+{\centering\adjincludegraphics[reflect, angle=180, width = 0.8 cm, valign=c]{fig/A8-024.eps}}\bigg]\bigg)\Bigg),
\end{equation}
\begin{equation}\label{33-d-113}
\mathrm{H}_{44}^{1,4}=-2\theta_2
\mathbb{\hat{H}}_0^{\mathrm{sc}}
\Bigg(
\bigg(
{\centering\adjincludegraphics[angle=180, width = 0.85 cm, valign=c]{fig/A8-016.eps}}+
{\centering\adjincludegraphics[angle=180, width = 0.8 cm, valign=c]{fig/A8-021.eps}}
\bigg)\mathrm{X}_2\Bigg).
\end{equation}
Note that transition \eqref{33-d-111} used the fact that combinations of the vertex connection $\mathrm{\widetilde{D}}_3$ with the vertices $\mathrm{\widetilde{D}}_1$ and $\mathrm{X}_2$ give diagrams without singular components. Further, the contributions from \eqref{33-d-112} and \eqref{33-d-113} can be investigated using the methods already used above, so we only write out the answers. We have
\begin{equation}\label{33-d-114}
\mathrm{H}_{44}^{1,3}\stackrel{\mathrm{s.p.}}{=}
\frac{32Lc_2}{\pi}\Big(\mathrm{B}_4+\mathrm{B}_5+2\mathrm{B}_6-2\mathrm{B}_7\Big)
-\frac{32L}{\pi}\Big(\mathrm{A}_5+\mathrm{A}_6-\mathrm{A}_7-\mathrm{A}_8\Big),
\end{equation}
\begin{equation}\label{33-d-115}
\mathrm{H}_{44}^{1,4}\stackrel{\mathrm{s.p.}}{=}
\frac{8L\theta_2}{\pi^2}\Big(-J_3[B]+c_2J_2[B]\Big)+\kappa_{12}S[B],
\end{equation}
where $\kappa_{12}$ is a singular coefficient independent of the background field. Summing up the answers, we get
\begin{align}\label{33-d-105-2}
	2\mathbb{H}_0^{\mathrm{sc}}(\mathrm{D}_1\mathrm{D}_3)\stackrel{\mathrm{s.p.}}{=}&\,2\Lambda^2\mathbb{H}_0^{\mathrm{sc}}(\widetilde{\Gamma}_2\mathrm{D}_3)+\kappa_{12}S[B]
	\\\nonumber
	&-\frac{32L}{\pi}\Big(\mathrm{A}_5+\mathrm{A}_6-\mathrm{A}_7-\mathrm{A}_8\Big)
	\\\nonumber
	&+\frac{32Lc_2}{\pi}\Big(\mathrm{B}_4+\mathrm{B}_5+2\mathrm{B}_6-2\mathrm{B}_7\Big)\\\nonumber&+\frac{8L\theta_2}{\pi^2}\Big(-J_3[B]+c_2J_2[B]\Big).
\end{align}

\noindent\underline{Part $\mathbb{H}_0^{\mathrm{sc}}(\mathrm{D}_3^2)$.} Let us use the transformation from \eqref{33-d-103} for a symmetric vertex, then the contribution can be rewritten as three terms
\begin{equation}\label{33-d-116}
\mathbb{H}_0^{\mathrm{sc}}(\mathrm{D}_3^2)=
\mathrm{H}_{44}^{1,4}+
\mathrm{H}_{44}^{1,5}+
\mathrm{H}_{44}^{1,6}
+o(1)
\end{equation}
where
\begin{equation}\label{33-d-117}
\mathrm{H}_{44}^{1,4}=
2\mathbb{\hat{H}}_0^{\mathrm{sc}}\Bigg(
\bigg(
{\centering\adjincludegraphics[ angle=180,width = 0.8 cm, valign=c]{fig/A8-019.eps}}+
{\centering\adjincludegraphics[ angle=180,width = 0.8 cm, valign=c]{fig/A8-020.eps}}+
{\centering\adjincludegraphics[ angle=180,reflect, angle=180, width = 0.8 cm, valign=c]{fig/A8-005.eps}}+
{\centering\adjincludegraphics[ angle=180,reflect, angle=180, width = 0.8 cm, valign=c]{fig/A8-002.eps}}
\bigg)
\bigg(
{\centering\adjincludegraphics[ width = 0.8 cm, valign=c]{fig/A8-019.eps}}+
{\centering\adjincludegraphics[ width = 0.8 cm, valign=c]{fig/A8-020.eps}}+
{\centering\adjincludegraphics[ reflect, angle=180, width = 0.8 cm, valign=c]{fig/A8-005.eps}}+
{\centering\adjincludegraphics[ reflect, angle=180, width = 0.8 cm, valign=c]{fig/A8-002.eps}}
\bigg)
\Bigg),
\end{equation}
\begin{equation}\label{33-d-118}
\mathrm{H}_{44}^{1,5}=2
\mathbb{\hat{H}}_0^{\mathrm{sc}}\Bigg(
\bigg({\centering\adjincludegraphics[angle=180, width = 0.85 cm, valign=c]{fig/A8-016.eps}}+
{\centering\adjincludegraphics[ angle=180,width = 0.8 cm, valign=c]{fig/A8-021.eps}}\bigg)
\bigg(
{\centering\adjincludegraphics[ width = 0.8 cm, valign=c]{fig/A8-019.eps}}+
{\centering\adjincludegraphics[ width = 0.8 cm, valign=c]{fig/A8-020.eps}}+
{\centering\adjincludegraphics[ reflect, angle=180, width = 0.8 cm, valign=c]{fig/A8-005.eps}}+
{\centering\adjincludegraphics[ reflect, angle=180, width = 0.8 cm, valign=c]{fig/A8-002.eps}}
\bigg)
\Bigg),
\end{equation}
\begin{equation}\label{33-d-119}
\mathrm{H}_{44}^{1,6}=\frac{1}{2}
\mathbb{\hat{H}}_0^{\mathrm{sc}}\Bigg(
\bigg({\centering\adjincludegraphics[angle=180, width = 0.85 cm, valign=c]{fig/A8-016.eps}}+
{\centering\adjincludegraphics[ angle=180,width = 0.8 cm, valign=c]{fig/A8-021.eps}}\bigg)
\bigg({\centering\adjincludegraphics[ width = 0.85 cm, valign=c]{fig/A8-016.eps}}+
{\centering\adjincludegraphics[width = 0.8 cm, valign=c]{fig/A8-021.eps}}\bigg)
\Bigg).
\end{equation}
All the components can be examined using the methods already used, so we only write down the answers
\begin{equation}\label{33-d-120}
\mathrm{H}_{44}^{1,4}\stackrel{\mathrm{s.p.}}{=}
-\frac{2L}{\pi}\Big(c_2\mathrm{B}_3-8\mathrm{A}_2+4\mathrm{A}_3+4\mathrm{A}_4+
8\mathrm{A}_5+
8\mathrm{A}_6
-8\mathrm{A}_7
-8\mathrm{A}_8\Big),
\end{equation}
\begin{equation}\label{33-d-121}
\mathrm{H}_{44}^{1,5}\stackrel{\mathrm{s.p.}}{=}
\frac{4c_2L}{\pi}\Big(\mathrm{B}_3+8\mathrm{B}_4+8\mathrm{B}_5+16\mathrm{B}_6-16\mathrm{B}_7\Big)
+
\frac{16L}{\pi}\Big(-2\mathrm{A}_2+\mathrm{A}_3+\mathrm{A}_4\Big),
\end{equation}
\begin{align}\label{33-d-122}
\mathrm{H}_{44}^{1,6}\stackrel{\mathrm{s.p.}}{=}&2^5\Lambda^2\alpha_6(\mathbf{f})
\bigg(J_6[B]+J_5[B]\frac{c_2^2L_1}{\pi}\bigg)\\\nonumber&+\frac{8L_1c_2\big(\theta_2+\theta_3\big)}{\pi^2}J_2[B]+2^6\alpha_5(\mathbf{f})\alpha_6(\mathbf{f})LJ_3[B]
+\kappa_{13}S[B],
\end{align}
where the functionals from Section \ref{33:sec:dia} were used, and the coefficient $\kappa_{13}$ is singular and does not depend on the background field. Summing up the answers, we get
\begin{align}\label{33-d-124}
\mathbb{H}_0^{\mathrm{sc}}(\mathrm{D}_3^2)\stackrel{\mathrm{s.p.}}{=}&2^5\Lambda^2\alpha_6(\mathbf{f})
\bigg(J_6[B]+J_5[B]\frac{c_2^2L_1}{\pi}\bigg)
\\\nonumber&
-\frac{2L}{\pi}\Big(8\mathrm{A}_2-4\mathrm{A}_3-4\mathrm{A}_4+
8\mathrm{A}_5+
8\mathrm{A}_6
-8\mathrm{A}_7
-8\mathrm{A}_8\Big)
\\\nonumber&
+\frac{4c_2L}{\pi}\Big(\mathrm{B}_3/2+8\mathrm{B}_4+8\mathrm{B}_5+16\mathrm{B}_6-16\mathrm{B}_7\Big)
\\\nonumber&+\frac{8L_1c_2\big(\theta_2+\theta_3\big)}{\pi^2}J_2[B]+2^6\alpha_5(\mathbf{f})\alpha_6(\mathbf{f})LJ_3[B]
+\kappa_{13}S[B].
\end{align}
Let us finally calculate the sum of all the contributions studied in this subsection.
\begin{lemma}\label{33-lem-9}
Taking into account all the above, the following decomposition is true for the first part of \eqref{33-d-75}
\begin{align}\label{33-d-124-1}
\mathrm{H}_{44}^{1}\stackrel{\mathrm{s.p.}}{=}&
-\Lambda^4\mathbb{H}_0^{\mathrm{sc}}(\widetilde{\Gamma}_2^2)+
2\Lambda^2\mathbb{H}_0^{\mathrm{sc}}(\Gamma_4\widetilde{\Gamma}_2)
-6L_1\Lambda^2\mathbb{H}_0^{\mathrm{sc}}\big(\widetilde{\Gamma}_2\mathrm{X}_1\big)
-9L_1^2\mathbb{H}_0^{\mathrm{sc}}(\mathrm{X}_1^2)+
6L_1\mathbb{H}_0^{\mathrm{sc}}(\Gamma_4\mathrm{X}_1)
\\\nonumber&+2^5\Lambda^2\alpha_6(\mathbf{f})
	\bigg(J_6[B]+J_5[B]\frac{c_2^2L_1}{\pi}\bigg)
	\\\nonumber&
	+\frac{L}{\pi}\Big(32\mathrm{A}_2-40\mathrm{A}_3+8\mathrm{A}_4
	-80\mathrm{A}_5
	-80\mathrm{A}_6
	+80\mathrm{A}_7
	+80\mathrm{A}_8\Big)
	\\\nonumber&
	+\frac{c_2L}{\pi}\Big(8\mathrm{B}_3-72\mathrm{B}_4-72\mathrm{B}_5+192\mathrm{B}_6-192\mathrm{B}_7\Big)
	\\\nonumber&
	+\frac{6c_2^2}{\pi^2}J_1[B]
	L_1\big(4\theta_2+2L_1\big)
	-\frac{12c_2^2}{\pi^2}J_4[B]L_1\big(4\theta_2-2\theta_1\big)
	\\\nonumber&+\frac{8Lc_2\big(2\theta_2+\theta_3\big)}{\pi^2}J_2[B]+\Big(2^6\alpha_5(\mathbf{f})\alpha_6(\mathbf{f})L-8L\theta_2/\pi^2\Big)J_3[B]
	+\hat{\kappa}_4S[B],
\end{align}
where $\hat{\kappa}_4$ is a singular coefficient that does not depend on the background field.
\end{lemma}

\subsubsection{Part without loops $\mathbb{H}_2^{\mathrm{c}}(\Gamma_4)$}
Let us move on to the second part of the decomposition from \eqref{33-d-75}. In this case, it is more convenient to transform the quartic vertices first, and only then proceed to the connection. To do this, let us pay attention to the possibility of rewriting\footnote{During the derivation, it was possible to rearrange the external lines belonging to one side, either left or right. This property is also present in the contribution under consideration.} the sum of three quartic vertices \eqref{33-d-33} with fixed outer lines as the sum of two vertices \eqref{33-d-36}, in which the outer lines are also fixed. This means that a set of diagrams $\mathrm{H}_{44}^2$ can be obtained by connecting the outer lines of two vertices
\begin{equation}\label{33-k-4}
{\centering\adjincludegraphics[width = 0.8 cm, valign=c]{fig/A7-004.eps}}+
3{\centering\adjincludegraphics[width = 0.8 cm, valign=c]{fig/A7-002.eps}}
\end{equation}
with the external lines of two vertices
\begin{equation}\label{33-k-5}
{\centering\adjincludegraphics[width = 1.4 cm, valign=c]{fig/A7-001.eps}}+
{\centering\adjincludegraphics[reflect, width = 1.4 cm, valign=c]{fig/A7-001.eps}}
\end{equation}
with the condition that the left ends are connected to the left, and the right ends to the right, in all possible ways. By performing several auxiliary integrations in parts, we obtain that the desired part is expressed as the following sum
\begin{equation}\label{33-k-3}
\mathrm{H}_{44}^2=\sum_{i=1}^{11}t_i,
\end{equation}
where
\begin{equation*}\label{33-k-1}
t_1=2{\centering\adjincludegraphics[ width = 1.3 cm, valign=c]{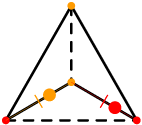}},\,\,\,
t_2=4{\centering\adjincludegraphics[ width = 1.3 cm, valign=c]{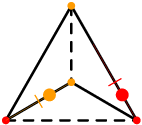}},\,\,\,
t_3=2{\centering\adjincludegraphics[ width = 1.3 cm, valign=c]{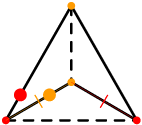}},\,\,\,
t_4=8{\centering\adjincludegraphics[ width = 1.3 cm, valign=c]{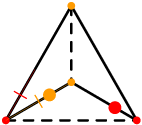}},\,\,\,
t_5=5{\centering\adjincludegraphics[ width = 1.3 cm, valign=c]{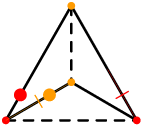}},\,\,\,
t_6=-2{\centering\adjincludegraphics[ width = 1.3 cm, valign=c]{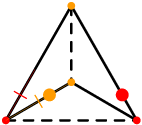}},\,\,\,
\end{equation*}
\begin{equation*}\label{33-k-2}
t_7=6{\centering\adjincludegraphics[ width = 1.3 cm, valign=c]{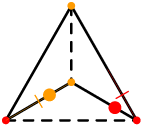}},\,\,\,
t_8=6{\centering\adjincludegraphics[ width = 1.3 cm, valign=c]{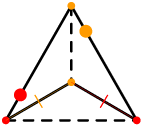}},\,\,\,
t_9=6{\centering\adjincludegraphics[ width = 1.3 cm, valign=c]{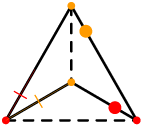}},\,\,\,
t_{10}=-6{\centering\adjincludegraphics[ width = 1.3 cm, valign=c]{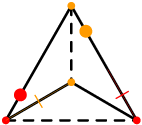}},\,\,\,
t_{11}=-6{\centering\adjincludegraphics[ width = 1.3 cm, valign=c]{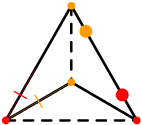}}.
\end{equation*}
The calculation of singular components for diagrams containing two integration operators largely repeats the well-known methods used in the study of two-loop diagrams, and boils down to shifting the variable and additional summation. Omitting the calculations that were performed using the properties from Section \ref{33:sec:sp}, we present the final decompositions for each individual diagram:
\begin{fleqn}
\begin{align}\label{33-k-6}
t_{1}\stackrel{\mathrm{s.p.}}{=}&-16\Lambda^2c_2\alpha_6J_7
+16\frac{\Lambda^2c_2^2}{\pi}J_5(L\alpha_6+\alpha_7)+
\frac{c_2^2L}{\pi^2}J_4(2\alpha_8-4\alpha_9)+
\frac{L}{\pi^2}J_3(4\theta_2 - 12\theta_3)\\\nonumber+&
\frac{c_2L}{\pi^2}J_2(4\alpha_9 +2 \theta_2 - 2\theta_3 -32 \pi^2 \alpha_5 \alpha_6 - 2\alpha_8)+
2\alpha_9\frac{c_2^2L}{\pi^2}J_1+
\kappa_{1,1}S[B]+\kappa_{2,1},
\end{align}
\end{fleqn}
\begin{fleqn}
\begin{align}\label{33-k-7}
t_{2}\stackrel{\mathrm{s.p.}}{=}&\,32\Lambda^2\alpha_6\Big(c_2J_7+2J_6\Big)+
32\frac{\Lambda^2c_2^2}{\pi}J_5(L\alpha_6+\alpha_7)+
\frac{c_2^2L}{\pi^2}J_4(8\alpha_8-16\alpha_9)\\\nonumber+&
\frac{L}{\pi^2}J_3(-16\alpha_9 -8 \theta_2 +8\theta_3 +128 \pi^2 \alpha_5 \alpha_6 +8\alpha_8)+
\frac{c_2L}{\pi^2}J_2(-8\alpha_9 -12 \theta_2 +28\theta_3 +64 \pi^2 \alpha_5 \alpha_6 +4\alpha_8)
\\\nonumber+&
\frac{c_2^2L}{\pi^2}J_1(-4\theta_2+4\theta_3)+
\kappa_{1,2}S[B]+\kappa_{2,2},
\end{align}
\end{fleqn}
\begin{fleqn}
\begin{align}\label{33-k-8}
t_{3}\stackrel{\mathrm{s.p.}}{=}&\,
\frac{L}{\pi}\Big(
4\mathrm{A}_1
-16\mathrm{A}_2
+8\mathrm{A}_3
+4\mathrm{A}_4
-16\mathrm{A}_6
+4\mathrm{A}_7
+4\mathrm{A}_8
\Big)+\frac{Lc_2}{\pi}\Big(
2\mathrm{B}_1
-2\mathrm{B}_4
-8\mathrm{B}_6
+4\mathrm{B}_7
\Big)\\\nonumber
-&16\alpha_{10}\Lambda^2c_2^2J_5+
\frac{c_2^2L}{\pi^2}J_4(-2\theta_2-2\alpha_{11})+
\frac{1}{\pi^2}J_3(L_1^2+2L\theta_1)+
\frac{c_2}{\pi^2}J_2\Big(\frac{3L_1^2}{2}+L(3 \theta_1 + 2\alpha_{11})\Big)
\\\nonumber+&
\frac{c_2^2}{\pi^2}J_1\Big(\frac{L_1^2}{2}+L(\theta_1 +  \theta_2 -\alpha_{11})\Big)+
\kappa_{1,3}S[B]+\kappa_{2,3},
\end{align}
\end{fleqn}
\begin{fleqn}
\begin{align}\label{33-k-9}
t_{4}\stackrel{\mathrm{s.p.}}{=}&\,
\frac{L}{\pi}\Big(
16\mathrm{A}_1
-16\mathrm{A}_4
-64\mathrm{A}_6
+48\mathrm{A}_7
+48\mathrm{A}_8
\Big)+\frac{Lc_2}{\pi}\Big(
8\mathrm{B}_1
-16\mathrm{B}_2
-24\mathrm{B}_4
-32\mathrm{B}_6
+48\mathrm{B}_7
\Big)\\\nonumber-&
64\alpha_{10}\Lambda^2c_2^2J_5+
\frac{c_2^2L}{\pi^2}J_4(-8\theta_2-8\alpha_{11})+
\frac{1}{\pi^2}J_3(4L_1^2+8L\theta_1)+
\frac{c_2}{\pi^2}J_2\Big(6L_1^2+L(12 \theta_1 + 8\alpha_{11})\Big)
\\\nonumber+&
\frac{c_2^2}{\pi^2}J_1\Big(2L_1^2+L(4\theta_1 +  4\theta_2 -4\alpha_{11})\Big)+
\kappa_{1,4}S[B]+\kappa_{2,4},
\end{align}
\end{fleqn}
\begin{fleqn}
\begin{align}\label{33-k-10}
t_{5}\stackrel{\mathrm{s.p.}}{=}&\,
\frac{L}{\pi}\Big(
-8\mathrm{A}_1
+16\mathrm{A}_2
-8\mathrm{A}_4
+16\mathrm{A}_7
\Big)+\frac{Lc_2}{\pi}\Big(
-4\mathrm{B}_1
+16\mathrm{B}_4
+16\mathrm{B}_5
+16\mathrm{B}_6
-16\mathrm{B}_7
-8\mathrm{B}_8
\Big)\\\nonumber-&
32\alpha_{10}\Lambda^2c_2^2J_5+
\frac{c_2^2L}{\pi^2}J_4(-4\theta_1+2\theta_2-4\alpha_{11})+
\frac{1}{\pi^2}J_3(-2L_1^2+L(-4\theta_1+4\theta_2))\\\nonumber+&
\frac{c_2}{\pi^2}J_2\Big(-L_1^2+L(-2\theta_1+4\theta_2+4\alpha_{11})\Big)+
\frac{c_2^2}{\pi^2}J_1\Big(-2L_1^2+L(-2\theta_1+\theta_2)\Big)+
\kappa_{1,5}S[B]+\kappa_{2,5},
\end{align}
\end{fleqn}
\begin{fleqn}
\begin{align}\label{33-k-11}
t_{6}\stackrel{\mathrm{s.p.}}{=}&\,
\frac{L}{\pi}\Big(
4\mathrm{A}_1
-4\mathrm{A}_4
+8\mathrm{A}_7
\Big)+\frac{Lc_2}{\pi}\Big(
-2\mathrm{B}_1
-8\mathrm{B}_4
-8\mathrm{B}_5
-8\mathrm{B}_6
+8\mathrm{B}_7
-4\mathrm{B}_8
\Big)\\\nonumber+&16\alpha_{10}\Lambda^2c_2^2J_5+
\frac{c_2^2L}{\pi^2}J_4(2\theta_1-\theta_2+4\alpha_{11})+
\frac{1}{\pi^2}J_3(L_1^2+L(2\theta_1+2\theta_2+6\alpha_{11}))\\\nonumber+&
\frac{c_2}{\pi^2}J_2\Big(\frac{L_1^2}{2}+L(\theta_1+2\theta_2)\Big)+
\frac{c_2^2}{\pi^2}J_1\Big(L_1^2+L(\theta_1+3\theta_2/2+\alpha_{11})\Big)+
\kappa_{1,6}S[B]+\kappa_{2,6},
\end{align}
\end{fleqn}
\begin{fleqn}
\begin{align}\label{33-k-12}
t_{7}\stackrel{\mathrm{s.p.}}{=}&\,
\frac{L}{\pi}\Big(
-12\mathrm{A}_1
-24\mathrm{A}_2
+24\mathrm{A}_3
+12\mathrm{A}_4
-24\mathrm{A}_8
\Big)\\\nonumber+&\frac{Lc_2}{\pi}\Big(
-6\mathrm{B}_1
+12\mathrm{B}_2
+36\mathrm{B}_4
+24\mathrm{B}_5
+24\mathrm{B}_6
-48\mathrm{B}_7
-12\mathrm{B}_8
\Big)\\\nonumber-&48\alpha_{10}\Lambda^2c_2^2J_5+
\frac{c_2^2L}{\pi^2}J_4(-6\theta_1+3\theta_2-6\alpha_{11})+
\frac{1}{\pi^2}J_3(-3L_1^2+L(-6\theta_1+6\theta_2))\\\nonumber+&
\frac{c_2}{\pi^2}J_2\Big(-\frac{3L_1^2}{2}+L(-3\theta_1+6\theta_2+6\alpha_{11})\Big)+
\frac{c_2^2}{\pi^2}J_1\Big(-3L_1^2+L(-3\theta_1+3\theta_2/2)\Big)+
\kappa_{1,7}S[B]+\kappa_{2,7},
\end{align}
\end{fleqn}
\begin{fleqn}
\begin{align}\label{33-k-13}
t_{8}\stackrel{\mathrm{s.p.}}{=}&\,
\frac{L}{\pi}\Big(
48\mathrm{A}_5
+48\mathrm{A}_6
-48\mathrm{A}_7
-48\mathrm{A}_8
\Big)+\frac{Lc_2}{\pi}\Big(
12\mathrm{B}_2
+24\mathrm{B}_4
+24\mathrm{B}_5
+96\mathrm{B}_6
-96\mathrm{B}_7
\Big)\\\nonumber+&
\frac{c_2^2L}{\pi^2}J_4(12\theta_1+24\alpha_{11})+
\frac{c_2^2L}{\pi^2}J_1(-6\theta_1+12\alpha_{11})+
\kappa_{1,8}S[B]+\kappa_{2,8},
\end{align}
\end{fleqn}
\begin{fleqn}
\begin{align}\label{33-k-14}
t_{9}\stackrel{\mathrm{s.p.}}{=}&\,
\frac{L}{\pi}\Big(
48\mathrm{A}_5
+48\mathrm{A}_6
-48\mathrm{A}_7
-48\mathrm{A}_8
\Big)+\frac{Lc_2}{\pi}\Big(
12\mathrm{B}_2
+24\mathrm{B}_4
+24\mathrm{B}_5
+96\mathrm{B}_6
-96\mathrm{B}_7
\Big)
\\\nonumber+&
\frac{c_2^2L}{\pi^2}J_4(12\theta_1+24\alpha_{11})+
\frac{c_2^2L}{\pi^2}J_1(-6\theta_1+12\alpha_{11})+
\kappa_{1,9}S[B]+\kappa_{2,9},
\end{align}
\end{fleqn}
\begin{fleqn}
\begin{align}\label{33-k-15}
t_{10}\stackrel{\mathrm{s.p.}}{=}&\,
\frac{L}{\pi}\Big(
12\mathrm{A}_1
+24\mathrm{A}_2
-24\mathrm{A}_3
-12\mathrm{A}_4
+12\mathrm{A}_5
+12\mathrm{A}_6
-24\mathrm{A}_7
\Big)\\\nonumber+&\frac{Lc_2}{\pi}\Big(
6\mathrm{B}_1
-12\mathrm{B}_2
+6\mathrm{B}_5
+36\mathrm{B}_6
-24\mathrm{B}_7
+12\mathrm{B}_8
\Big)\\\nonumber
-&24\alpha_{10}\Lambda^2c_2^2J_5+
\frac{c_2^2L}{\pi^2}J_4(6\theta_1-3\theta_2+6\alpha_{11})+
\frac{1}{\pi^2}J_3(3L_1^2+L(6\theta_1-6\theta_2-6\alpha_{11}))\\\nonumber+&
\frac{c_2}{\pi^2}J_2\Big(\frac{3L_1^2}{2}+L(3\theta_1-6\theta_2-3\alpha_{11})\Big)+
\frac{c_2^2}{\pi^2}J_1\Big(3L_1^2+L(3\theta_1-3\theta_2/2+3\alpha_{11})\Big)+
\kappa_{1,10}S[B]+\kappa_{2,10},
\end{align}
\end{fleqn}
\begin{fleqn}
\begin{align}\label{33-k-16}
t_{11}\stackrel{\mathrm{s.p.}}{=}&\,
\frac{L}{\pi}\Big(
12\mathrm{A}_5
+12\mathrm{A}_6
-36\mathrm{A}_7
-12\mathrm{A}_8
\Big)\\\nonumber+&\frac{Lc_2}{\pi}\Big(
12\mathrm{B}_1
+6\mathrm{B}_4
+6\mathrm{B}_5
+36\mathrm{B}_6
-36\mathrm{B}_7
+12\mathrm{B}_8
\Big)\\\nonumber-&24\alpha_{10}\Lambda^2c_2^2J_5+
\frac{c_2^2L}{\pi^2}J_4(6\theta_1-3\theta_2+6\alpha_{11})+
\frac{1}{\pi^2}J_3(L(-6\theta_2-6\alpha_{11}))+\kappa_{2,11}\\\nonumber+&
\frac{c_2}{\pi^2}J_2\Big(3L_1^2+L(6\theta_1-6\theta_2-3\alpha_{11})\Big)+
\frac{c_2^2}{\pi^2}J_1\Big(-\frac{3L_1^2}{2}+L(-6\theta_1-3\theta_2/2+3\alpha_{11})\Big)+
\kappa_{1,11}S[B].
\end{align}
\end{fleqn}
\begin{lemma}\label{33-lem-11}
Taking into account all of the above, the following decomposition is true
\begin{align}\label{33-y-20-1}
\mathrm{H}_{44}^2\stackrel{\mathrm{s.p.}}{=}&\,\frac{L}{\pi}\Big(
16\mathrm{A}_1
+8\mathrm{A}_3
-24\mathrm{A}_4
+120\mathrm{A}_5
+40\mathrm{A}_6
-80\mathrm{A}_7
-80\mathrm{A}_8
\Big)
\\\nonumber+&\frac{Lc_2}{\pi}\Big(
16\mathrm{B}_1
+8\mathrm{B}_2
+72\mathrm{B}_4
+92\mathrm{B}_5
+256\mathrm{B}_6
-256\mathrm{B}_7
\Big)
\\\nonumber+&
16\Lambda^2\alpha_6\Big(c_2J_7+4J_6+\frac{3L_1c_2^2}{\pi}J_5\Big)+\frac{48c_2^2}{\pi}\Lambda^2J_5(-\mathbf{f}(0)/2+\alpha_7-4\alpha_{10})
\\\nonumber+&
\frac{Lc_2^2J_1}{\pi^2}\Big(2\alpha_9+4\theta_3+2\theta_2-14\theta_1+26\alpha_{11}\Big)
\\\nonumber+&
\frac{c_2J_2}{\pi^2}\Big(10L_1^2+L\big[32\pi^2\alpha_5\alpha_6+2\alpha_8-4\alpha_9+26\theta_3-10\theta_2+20\theta_1+14\alpha_{11}\big]\Big)
\\\nonumber+&
\frac{J_3}{\pi^2}\Big(4L_1^2+L\big[128\pi^2\alpha_5\alpha_6+8\alpha_8-16\alpha_9-4\theta_3-4\theta_2+8\theta_1-6\alpha_{11}\big]\Big)
\\\nonumber+&
\frac{Lc_2^2J_4}{\pi^2}\Big(10\alpha_8-20\alpha_9-12\theta_2+28\theta_1+44\alpha_{11}\Big)
+\hat{\kappa}_{1,1}S[B]+\hat{\kappa}_{2,1},
\end{align}
where $\hat{\kappa}_{1,1}$ and $\hat{\kappa}_{2,1}$ are singular coefficients that do not depend on the background field.
\end{lemma}

\subsubsection{Sum of parts for $\mathbb{H}_0^{\mathrm{sc}}(\Gamma_4^2)$}
\label{33:sec:sum2}
Consider the sum of the contributions for all diagrams constructed using two quartic vertices. When summing up, we additionally take into account two relations. First, note that some of the diagrams in Section \ref{33:sec:dia} are linearly dependent. In particular, equality is required
\begin{equation*}
\mathrm{A}_5=\mathrm{A}_6-\frac{c_2}{2}\mathrm{B}_5.
\end{equation*}
Secondly, we define an auxiliary vertex with four external lines
\begin{equation*}
\widetilde{\Gamma}_4[\phi]=\int_{\mathbb{R}^2}\mathrm{d}^2x\,
f^{a_1b_1a_2}\phi^{b_1}(x)f^{a_2b_2a_3}\phi^{b_2}(x)
f^{a_3b_3a_4}\phi^{b_3}(x)f^{a_4b_4a_1}\phi^{b_4}(x).
\end{equation*}
For such a vertex, using standard methods and properties from Section \ref{33:sec:cl:1}, we can prove the relation
\begin{align}\label{33-e-1}
\mathbb{H}_0^{\mathrm{sc}}(\widetilde{\Gamma}_4)+\kappa_{44,1}=\,\,&
\bigg(2\,
{\centering\adjincludegraphics[width = 0.5 cm, valign=c]{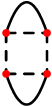}}+
{\centering\adjincludegraphics[width = 0.7 cm, valign=c]{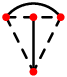}}
\bigg)+\kappa_{44,1}
=\,\bigg(3\,
{\centering\adjincludegraphics[width = 0.5 cm, valign=c]{fig/A11-001.eps}}
+\frac{c_2}{2}
{\centering\adjincludegraphics[width = 0.7 cm, valign=c]{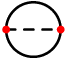}}
\bigg)+\kappa_{44,1}
\\\nonumber\stackrel{\mathrm{s.p.}}{=}\,&12J_6[B]+2c_2J_7[B]+\frac{10c_2^2L_1}{\pi}J_5[B]
\\\nonumber
+\,\,&
\frac{4\alpha_5(\mathbf{f})}{\Lambda^2}\big(6J_3[B]+c_2J_2[B]\big)+
\frac{L\kappa_{44,2}}{\Lambda^2}S[B]
+\mathcal{O}\big(1/\Lambda^2\big),
\end{align}
where $\kappa_{44,1}$ is a singular coefficient, while $\kappa_{44,2}$ is just a finite constant.

\begin{lemma}\label{33-lem-14}
Taking into account all the above, for the sum of the diagrams $\mathbb{H}_0^{\mathrm{sc}}(\Gamma_4\Gamma_4)$ the following asymptotic expansion is valid
\begin{align}\label{33-d-124-1-1}
		\mathbb{H}_0^{\mathrm{sc}}(\Gamma_4^2)\stackrel{\mathrm{s.p.}}{=}&
		-\Lambda^4\mathbb{H}_0^{\mathrm{sc}}(\widetilde{\Gamma}_2^2)+
		2\Lambda^2\mathbb{H}_0^{\mathrm{sc}}(\Gamma_4\widetilde{\Gamma}_2)
		-6L_1\Lambda^2\mathbb{H}_0^{\mathrm{sc}}\big(\widetilde{\Gamma}_2\mathrm{X}_1\big)
		-9L_1^2\mathbb{H}_0^{\mathrm{sc}}(\mathrm{X}_1^2)+
		6L_1\mathbb{H}_0^{\mathrm{sc}}(\Gamma_4\mathrm{X}_1)
		\\\nonumber&+8\alpha_6\Lambda^2\mathbb{H}_0^{\mathrm{sc}}(\widetilde{\Gamma}_4)+\frac{48c_2^2}{\pi}\Lambda^2J_5(-\mathbf{f}(0)/2+\alpha_7-4\alpha_{10})
		\\\nonumber&
		+\frac{16L}{\pi}\Big(\mathrm{A}_1+2\mathrm{A}_2-2\mathrm{A}_3-\mathrm{A}_4\Big)
		+\frac{c_2L}{\pi}\Big(16\mathrm{B}_1
		+8\mathrm{B}_2+8\mathrm{B}_3+448\mathrm{B}_6-448\mathrm{B}_7\Big)
	\\\nonumber&+
	\frac{c_2^2J_1}{\pi^2}\Big(12L_1^2+L\big[2\alpha_9+4\theta_3+26\theta_2-14\theta_1+26\alpha_{11}\big]\Big)
	\\\nonumber&+
	\frac{c_2J_2}{\pi^2}\Big(10L_1^2+L\big[2\alpha_8-4\alpha_9+34\theta_3+6\theta_2+20\theta_1+14\alpha_{11}\big]\Big)
	\\\nonumber&+
	\frac{J_3}{\pi^2}\Big(4L_1^2+L\big[8\alpha_8-16\alpha_9-4\theta_3-12\theta_2+8\theta_1-6\alpha_{11}\big]\Big)
	\\\nonumber&+
	\frac{Lc_2^2J_4}{\pi^2}\Big(10\alpha_8-20\alpha_9-60\theta_2+52\theta_1+44\alpha_{11}\Big)
		+(\hat{\kappa}_4+\hat{\kappa}_5)S[B]+\hat{\kappa}_6,
\end{align}
	where $\{\hat{\kappa}_4,\hat{\kappa}_5,\hat{\kappa}_6\}$ are singular coefficients and do not depend on the background field.
\end{lemma}
\begin{lemma}\label{33-lem-15}
Taking into account the representation from \eqref{33-k-4} and the decomposition from \eqref{33-d-81}, the formula is valid
\begin{multline}\label{33-d-124-5}
		\mathbb{H}_0^{\mathrm{sc}}(\Gamma_4)=\Lambda^2\mathbb{H}_0^{\mathrm{sc}}(\widetilde{\Gamma}_2)
		+12\Big(\mathrm{B}_6-\mathrm{B}_7\Big)+\frac{3L_1c_2}{2\pi}J_1[B]+\frac{\theta_2}{\pi}J_2[B]+\\+S[B]\bigg(-\frac{3L_1^2c_2^2}{8\pi^2}+\frac{c_2^2L_1\theta_2}{2\pi^2}\bigg)
		+\hat{\kappa}_7+o\big(1/\Lambda\big),
	\end{multline}
	where the coefficient $\hat{\kappa}_7$ is singular and independent of the background field.
\end{lemma}

\subsection{The case $\mathbb{H}_0^{\mathrm{sc}}(\Gamma_6)$}
\label{33:sec:cl:6}
Let us turn to the consideration of the sextic vertex $\Gamma_6$. It is clear that after applying the operator $\mathbb{H}_0^{\mathrm{sc}}$, we obtain a set of fifteen diagrams, each of which contains three Green's functions on the diagonal, that is, with matching arguments. However, some Green's functions may contain derivatives (no more than two). Thus, by performing all possible connections of the external lines, we obtain the following decomposition
\begin{equation}\label{33-y-6}
\mathbb{H}_0^{\mathrm{sc}}(\Gamma_6)=\sum_{i=1}^{15}e_i.
\end{equation}
Omitting the operator elements, the last diagrams can be represented as follows
\begin{equation}\label{33-s-31-t}
	{\centering\adjincludegraphics[width = 1.6 cm, valign=c]{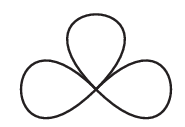}}.
\end{equation}
Taking into account the explicit form of the sextic vertex \eqref{33-g-6}, the contributions have the form
\begin{equation}\label{33-y-1}
e_1=-{\centering\adjincludegraphics[ width = 2 cm, valign=c]{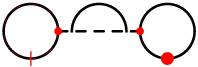}},\,\,\,
e_2={\centering\adjincludegraphics[ width = 1.5 cm, valign=c]{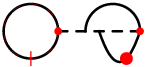}},\,\,\,
e_3=-{\centering\adjincludegraphics[ width = 1.5 cm, valign=c]{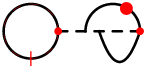}},\,\,\,
e_4={\centering\adjincludegraphics[ width = 1.5 cm, valign=c]{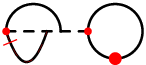}},\,\,\,
e_5=-{\centering\adjincludegraphics[ width = 1.2 cm, valign=c]{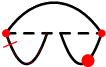}},
\end{equation}
\begin{equation}\label{33-y-3}
e_6={\centering\adjincludegraphics[ width = 1.2 cm, valign=c]{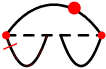}},\,\,\,
e_7={\centering\adjincludegraphics[ width = 1.5 cm, valign=c]{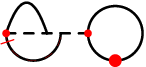}},\,\,\,
e_8=-{\centering\adjincludegraphics[ width = 1.2 cm, valign=c]{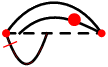}},\,\,\,
e_9={\centering\adjincludegraphics[ width = 1.2 cm, valign=c]{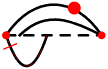}},\,\,\,
e_{10}={\centering\adjincludegraphics[ width = 1.2 cm, valign=c]{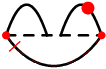}},
\end{equation}
\begin{equation}\label{33-y-5}
e_{11}={\centering\adjincludegraphics[ width = 1.2 cm, valign=c]{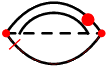}},\,\,\,
e_{12}={\centering\adjincludegraphics[ width = 1.2 cm, valign=c]{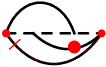}},\,\,\,
e_{13}=-{\centering\adjincludegraphics[ width = 1.2 cm, valign=c]{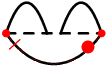}},\,\,\,
e_{14}=-{\centering\adjincludegraphics[ width = 1.2 cm, valign=c]{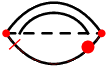}},\,\,\,
e_{15}=-{\centering\adjincludegraphics[ width = 1.2 cm, valign=c]{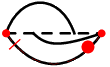}}.
\end{equation}
The study of recent contributions actually boils down to substituting the Green's function and its first derivatives on the diagonal. In most cases, it is enough to simply use formulas \eqref{33-r-28} and \eqref{33-r-32}. Omitting the step of some small summations, we present the answers for the diagrams.
\begin{fleqn}
\begin{equation}\label{33-y-7}
e_1\stackrel{\mathrm{s.p.}}{=}\frac{Lc_2}{\pi}\Big(2\mathrm{B}_3+4\mathrm{B}_4+4\mathrm{B}_5\Big)
+\frac{L_1^2}{\pi^2}\Big(-c_2^2\frac{3J_1}{2}\Big)
+\kappa_{4,1}S[B],
\end{equation}
\end{fleqn}
\begin{fleqn}
\begin{equation}\label{33-y-8}
e_2\stackrel{\mathrm{s.p.}}{=}\frac{Lc_2}{\pi}\Big(2\mathrm{B}_2+3\mathrm{B}_4+2\mathrm{B}_5-2\mathrm{B}_7\Big)+\frac{L_1^2}{\pi^2}\Big(-c_2^2\frac{3J_1}{4}\Big)
+\kappa_{4,2}S[B],
\end{equation}
\end{fleqn}
\begin{fleqn}
\begin{equation}\label{33-y-9}
e_3\stackrel{\mathrm{s.p.}}{=}\frac{L}{\pi}\Big(2\mathrm{A}_8-2\mathrm{A}_7\Big)+\frac{Lc_2}{\pi}\Big(2\mathrm{B}_1+4\mathrm{B}_4+4\mathrm{B}_5+2\mathrm{B}_8\Big)
+\frac{L_1^2}{\pi^2}\Big(-c_2^2J_1+c_2\frac{J_2}{2}\Big)+\kappa_{4,3}S[B],
\end{equation}
\end{fleqn}
\begin{fleqn}
\begin{equation}\label{33-y-10}
e_4\stackrel{\mathrm{s.p.}}{=}\frac{Lc_2}{\pi}\Big(\mathrm{B}_3+\mathrm{B}_4+2\mathrm{B}_5+2\mathrm{B}_7\Big)+\frac{L_1^2}{\pi^2}\Big(-c_2^2\frac{3J_1}{4}\Big)+\kappa_{4,4}S[B],
\end{equation}
\end{fleqn}
\begin{fleqn}
\begin{equation}\label{33-y-11}
e_5\stackrel{\mathrm{s.p.}}{=}\frac{L}{\pi}\Big(-4\mathrm{A}_5
+2\mathrm{A}_7
+2\mathrm{A}_8
\Big)+\frac{Lc_2}{\pi}\Big(\mathrm{B}_2+\mathrm{B}_4\Big)
+\frac{L_1^2}{\pi^2}\Big(-c_2^2\frac{3J_1}{8}\Big)+\kappa_{4,5}S[B],
\end{equation}
\end{fleqn}
\begin{fleqn}
\begin{equation}\label{33-y-12}
e_6\stackrel{\mathrm{s.p.}}{=}\frac{L}{\pi}2\mathrm{A}_8+\frac{Lc_2}{\pi}\Big(\mathrm{B}_1-4\mathrm{B}_6+4\mathrm{B}_7+\mathrm{B}_8\Big)
+\frac{L_1^2}{\pi^2}\Big(-c_2^2\frac{J_1}{2}+c_2\frac{J_2}{4}\Big)+\kappa_{4,6}S[B],
\end{equation}
\end{fleqn}
\begin{fleqn}
\begin{align}\label{33-y-13}
e_7\stackrel{\mathrm{s.p.}}{=}&\,\frac{L}{\pi}\Big(
-4\mathrm{A}_2+4\mathrm{A}_3+2\mathrm{A}_7-2\mathrm{A}_8
\Big)+\frac{Lc_2}{\pi}\Big(\mathrm{B}_3+4\mathrm{B}_4+4\mathrm{B}_5-2\mathrm{B}_8\Big)
\\\nonumber+&\frac{L_1^2}{\pi^2}\Big(-c_2^2J_1+c_2\frac{J_2}{2}\Big)+\kappa_{4,7}S[B],
\end{align}
\end{fleqn}
\begin{fleqn}
\begin{align}\label{33-y-14}
e_8\stackrel{\mathrm{s.p.}}{=}&\,\frac{L}{\pi}\Big(
-\mathrm{A}_1
+\mathrm{A}_4-4\mathrm{A}_6+2\mathrm{A}_7+2\mathrm{A}_8
\Big)+\frac{Lc_2}{\pi}\Big(-\frac{1}{2}\mathrm{B}_1+\mathrm{B}_2-\mathrm{B}_4-2\mathrm{B}_6+2\mathrm{B}_7\Big)
\\\nonumber+&\frac{L_1^2}{\pi^2}\Big(-c_2^2\frac{J_1}{8}-c_2\frac{3J_2}{8}-\frac{J_3}{4}\Big)+\kappa_{4,8}S[B],
\end{align}
\end{fleqn}
\begin{fleqn}
\begin{align}\label{33-y-15}
e_9\stackrel{\mathrm{s.p.}}{=}&\,\frac{L}{\pi}\Big(
-\mathrm{A}_1
+\mathrm{A}_4-\mathrm{A}_7+\mathrm{A}_8
\Big)+\frac{Lc_2}{\pi}\Big(\frac{1}{2}\mathrm{B}_1-\frac{1}{2}\mathrm{B}_4-2\mathrm{B}_6+3\mathrm{B}_7+\mathrm{B}_8\Big)\\\nonumber
+&\frac{L_1^2}{\pi^2}\Big(-c_2^2\frac{J_1}{4}-c_2\frac{J_2}{8}-\frac{J_3}{4}\Big)+\kappa_{4,9}S[B],
\end{align}
\end{fleqn}
\begin{fleqn}
\begin{align}\label{33-y-16}
e_{10}\stackrel{\mathrm{s.p.}}{=}&\,\frac{L}{\pi}\Big(
-6\mathrm{A}_2
+4\mathrm{A}_3+2\mathrm{A}_4-2\mathrm{A}_8
\Big)+\frac{Lc_2}{\pi}\Big(\mathrm{B}_3-4\mathrm{B}_6+4\mathrm{B}_7-\mathrm{B}_8\Big)
\\\nonumber+&\frac{L_1^2}{\pi^2}\Big(-c_2^2\frac{J_1}{2}+c_2\frac{J_2}{4}\Big)+\kappa_{4,10}S[B],
\end{align}
\end{fleqn}
\begin{fleqn}
\begin{equation}\label{33-y-17}
e_{11}\stackrel{\mathrm{s.p.}}{=}\frac{L}{\pi}\Big(
-2\mathrm{A}_1
+2\mathrm{A}_2
\Big)+\frac{Lc_2}{\pi}\Big(-4\mathrm{B}_6+4\mathrm{B}_7\Big)
+\frac{L_1^2}{\pi^2}\Big(-c_2^2\frac{J_1}{2}-c_2\frac{J_2}{4}-\frac{J_3}{2}\Big)+\kappa_{4,11}S[B],
\end{equation}
\end{fleqn}
\begin{fleqn}
\begin{align}\label{33-y-18}
e_{12}\stackrel{\mathrm{s.p.}}{=}&\,\frac{L}{\pi}\Big(
-\mathrm{A}_1
-2\mathrm{A}_2+2\mathrm{A}_3+\mathrm{A}_4+\mathrm{A}_7-\mathrm{A}_8
\Big)\\\nonumber+&\frac{Lc_2}{\pi}\Big(
-\frac{1}{2}\mathrm{B}_1
+\mathrm{B}_2
+\frac{1}{2}\mathrm{B}_4
-2\mathrm{B}_6
+\mathrm{B}_7
-\mathrm{B}_8
\Big)+\frac{L_1^2}{\pi^2}\Big(-c_2^2\frac{J_1}{4}-c_2\frac{J_2}{8}-\frac{J_3}{4}\Big)+\kappa_{4,12}S[B],
\end{align}
\end{fleqn}
\begin{fleqn}
\begin{align}\label{33-y-19}
e_{13}+e_{14}+e_{15}\stackrel{\mathrm{s.p.}}{=}&\,\frac{L}{\pi}\Big(
4\mathrm{A}_5
+4\mathrm{A}_6-4\mathrm{A}_7-4\mathrm{A}_8
\Big)+\frac{16Lc_2}{\pi}\Big(
-\mathrm{B}_4
-\mathrm{B}_5
-2\mathrm{B}_6
+2\mathrm{B}_7
\Big)\\\nonumber-&\frac{\Lambda^2\alpha_1(\mathbf{f})}{\pi}\Big(\frac{5L_1c_2^2}{\pi}J_5[B]+6J_6[B]+c_2J_7[B]\Big)
-\frac{2L\alpha_1(\mathbf{f})\alpha_5(\mathbf{f})}{\pi}\Big(c_2J_2+6J_3\Big)
\\\nonumber+&\frac{L\theta_2}{\pi^2}\Big(-4c_2J_2+J_3\Big)+\kappa_{4,13}S[B]+\kappa_{4,14}.
\end{align}
\end{fleqn}
Here, the coefficients $\kappa_{4,1}$--$\kappa_{4,14}$ depend on the regularizing parameter and do not contain the background field. Summing up all the contributions, we can obtain a general expression for the case of the sextic vertex.
\begin{lemma}\label{33-lem-10}
Taking into account all of the above, the following decomposition is true
\begin{align}\label{33-y-20}
\mathbb{H}_0^{\mathrm{sc}}(\Gamma_6)\stackrel{\mathrm{s.p.}}{=}
-&\frac{\Lambda^2\alpha_1(\mathbf{f})}{2\pi}\mathbb{H}_0^{\mathrm{sc}}(\widetilde{\Gamma}_4)
+\frac{L}{\pi}\Big(
-5\mathrm{A}_1
-10\mathrm{A}_2
+10\mathrm{A}_3
+5\mathrm{A}_4
\Big)
\\\nonumber+&\frac{Lc_2}{\pi}\Big(
\frac{5}{2}\mathrm{B}_1
+5\mathrm{B}_2
+5\mathrm{B}_3
-50\mathrm{B}_6
+50\mathrm{B}_7
\Big)
\\\nonumber+&
\frac{L_1^2}{\pi^2}\Big(-c_2^2\frac{15J_1}{2}+c_2\frac{5J_2}{8}-\frac{5J_3}{4}\Big)
\\\nonumber+&
\frac{L\theta_2}{\pi^2}\Big(-4c_2J_2+J_3\Big)+\hat{\kappa}_8S[B]+\hat{\kappa}_9,
\end{align}
where $\hat{\kappa}_8$ and $\hat{\kappa}_9$ are singular coefficients and do not depend on the background field.
\end{lemma}

\section{Additional verification}
\label{33:sec:dop}

An important property of the diagrams from Section \ref{33:sec:cl} is their independence from the auxiliary dimensional parameter $\sigma$, which appears when dividing the integration regions. This leads to the fact that the parameter $\sigma$ is contained not only in the logarithms $\ln(\Lambda/\sigma)$, but also in the function $PS_\Lambda$, and, moreover, in the correction terms. Given the fact that the structure of the singularities should not depend on the choice of $\sigma$, we can conclude that when shifting $\sigma\to\sigma_1$, the singularities should be redistributed in a consistent manner. We use this observation to verify the calculations.

So, let us define the operation $\Upsilon$, which acts on the functions $f(\sigma)$, depending on the parameter $\sigma$, by the equality
\begin{equation}\label{33-e-3}
\Upsilon(f)=\lim_{\sigma_1\to\sigma}\frac{f(\sigma_1)-f(\sigma)}{\sigma_1-\sigma}.
\end{equation}
Then, using the properties of the functions from Section \ref{33:sec:cl:1}, two key relations can be obtained
\begin{equation}\label{33-e-4}
\Upsilon\big(PS_\Lambda^{ab}(x,x)\big)=-\frac{\delta^{ab}}{2\pi}+\mathcal{O}\big(1/\Lambda^2\big),
\end{equation}
\begin{equation}\label{33-e-5}
\Upsilon\Big(\partial_{y^\mu}PS_\Lambda^{ab}(y,x)\big|_{y=x}\Big)=-\frac{B^{ab}_\mu(x)}{2\pi}+\mathcal{O}\big(1/\Lambda^2\big),
\end{equation}
from which a series of equations is instantly obtained
\begin{equation}\label{33-e-6}
\Upsilon\big(\mathrm{A}_{1}\big)=\frac{1}{2\pi}(c_2J_2-J_3)+\mathcal{O}\big(1/\Lambda^2\big)
,\,\,\,
\Upsilon\big(\mathrm{A}_{2}\big)=\frac{1}{4\pi}(3c_2J_2+c_2^2J_1)+\mathcal{O}\big(1/\Lambda^2\big),
\end{equation}
\begin{equation}\label{33-e-7}
\Upsilon\big(\mathrm{A}_{3}\big)=\frac{1}{4\pi}(2c_2J_2+c_2^2J_1)+\mathcal{O}\big(1/\Lambda^2\big)
,\,\,\,
\Upsilon\big(\mathrm{A}_{4}\big)=\frac{1}{\pi}c_2J_2+\mathcal{O}\big(1/\Lambda^2\big),
\end{equation}
\begin{equation}\label{33-e-8}
\Upsilon\big(\mathrm{A}_{5}\big)=\frac{1}{8\pi}c_2^2J_4+\mathcal{O}\big(1/\Lambda^2\big)
,\,\,\,
\Upsilon\big(\mathrm{A}_{6}\big)=-\frac{1}{8\pi}c_2^2J_4+\mathcal{O}\big(1/\Lambda^2\big),
\end{equation}
\begin{equation}\label{33-e-9}
\Upsilon\big(\mathrm{A}_{7}\big)=\frac{1}{8\pi}(c_2J_2-2c_2^2J_4)+\mathcal{O}\big(1/\Lambda^2\big)
,\,\,\,
\Upsilon\big(\mathrm{A}_{8}\big)=\frac{1}{8\pi}(-c_2J_2+2c_2^2J_4)+\mathcal{O}\big(1/\Lambda^2\big),
\end{equation}
\begin{equation}\label{33-e-10}
\Upsilon\big(\mathrm{B}_{1}\big)=\frac{1}{2\pi}(c_2J_1-J_2)+\mathcal{O}\big(1/\Lambda^2\big)
,\,\,\,
\Upsilon\big(\mathrm{B}_{2}\big)=\frac{1}{2\pi}c_2J_1+\mathcal{O}\big(1/\Lambda^2\big),
\end{equation}
\begin{equation}\label{33-e-11}
\Upsilon\big(\mathrm{B}_{3}\big)=\frac{1}{\pi}c_2J_1+\mathcal{O}\big(1/\Lambda^2\big)
,\,\,\,
\Upsilon\big(\mathrm{B}_{4}\big)=\frac{1}{4\pi}(c_2J_1+2c_2J_4)+\mathcal{O}\big(1/\Lambda^2\big),
\end{equation}
\begin{equation}\label{33-e-12}
\Upsilon\big(\mathrm{B}_{5}\big)=-\frac{1}{2\pi}c_2J_4+\mathcal{O}\big(1/\Lambda^2\big)
,\,\,\,
\Upsilon\big(\mathrm{B}_{6}\big)=\frac{1}{4\pi}c_2J_4+\mathcal{O}\big(1/\Lambda^2\big),
\end{equation}
\begin{equation}\label{33-e-13}
\Upsilon\big(\mathrm{B}_{7}\big)=\frac{1}{8\pi}(c_2J_1+2c_2J_4)+\mathcal{O}\big(1/\Lambda^2\big),
\end{equation}
\begin{equation}\label{33-e-15}
\Upsilon\big(J_i\big)=\varepsilon_iS[B]+\mathcal{O}\big(1/\Lambda^2\big)
,\,\,\,\mbox{where}\,\,\,i\in\{1,2,3,4\},
\end{equation}
\begin{equation}\label{33-e-16}
\Upsilon\big(J_6\big)=-\frac{1}{\pi}c_2^2J_5+\mathcal{O}\big(1/\Lambda^2\big),\,\,\,
\Upsilon\big(J_7\big)=\frac{1}{\pi}c_2J_5+\mathcal{O}\big(1/\Lambda^2\big),
\end{equation}
\begin{equation}\label{33-e-14}
\Upsilon\big(\mathbb{H}_0^{\mathrm{sc}}(\Gamma)\big)=0,
\end{equation}
where $\Gamma$ is a set of vertices independent of the auxiliary parameter $\sigma$, and $\{\varepsilon_i\}$ is a set of constants. By applying the operator $\Upsilon$ to the relations from the lemmas of Section \ref{33:sec:cl} and applying the last equalities, we make sure that the structure of the divergences does not change, and all functions containing singular coefficients acquire the substitution of the argument $\sigma\to\sigma_1$.

\section{Quasi-local vertices}
\label{33:sec:kva}
In Section \ref{33:sec:osn}, it was shown that when local auxiliary vertices are selected, we obtain additional contributions to the classical action containing logarithmic singularities of the form $L_1^2$. It is important to note that the introduced regularization in Section \ref{33:sec:ps} is based on the averaging of fluctuations. This approach violates locality, but preserves quasi-locality, since averaging occurs over a "small" neighborhood of the selected point. 

Consider a series of quasi-local vertices. In this case, the quasi-locality lies in the fact that in addition to averaging fluctuations, an integral operator also appears, and the kernel of which has a rather small support. So let $\mathrm{K}(\cdot)$ be a continuous function such that $\mathrm{supp}(\mathrm{K})\subset\mathrm{B}_1$, then we define four functionals
\begin{equation}\label{33-j-1}
	\mathrm{V}_1[\mathrm{K},\phi]=\int_{\mathbb{R}^2}\mathrm{d}^2x\int_{\mathbb{R}^2}\mathrm{d}^2y\,
	B^{ba}(x)\phi^a(x)\Big(f^{bec}\Lambda^2\mathrm{K}\big((x-y)\Lambda\big)G_\Lambda^{cd}(x,y)B^{gd}_\mu(y)f^{gef}\Big)\phi^f(y),
\end{equation}
\begin{equation}\label{33-j-2}
	\mathrm{V}_2[\mathrm{K},\phi]=\int_{\mathbb{R}^2}\mathrm{d}^2x\int_{\mathbb{R}^2}\mathrm{d}^2y\,
	B^{ba}(x)\phi^a(x)\Big(f^{bec}\Lambda^2\mathrm{K}\big((x-y)\Lambda\big)G_\Lambda^{cd}(x,y)f^{def}\Big)
	B^{fg}_\mu(y)\phi^g(y),
\end{equation}
\begin{equation}\label{33-j-3}
	\mathrm{V}_3[\mathrm{K},\phi]=\int_{\mathbb{R}^2}\mathrm{d}^2x\int_{\mathbb{R}^2}\mathrm{d}^2y\,
	B^{ba}(x)\phi^a(x)f^{ibk}\Big(f^{kec}\Lambda^2\mathrm{K}\big((x-y)\Lambda\big)G_\Lambda^{cd}(x,y)f^{def}\Big)f^{fli}
	B^{lg}_\mu(y)\phi^g(y),
\end{equation}
\begin{equation}\label{33-j-4}
	\mathrm{V}_4[\mathrm{K},\phi]=\int_{\mathbb{R}^2}\mathrm{d}^2x\int_{\mathbb{R}^2}\mathrm{d}^2y\,
	\Big(\partial_{x_\mu}\phi^b(x)\Big)\Big(f^{bec}\Lambda^2\mathrm{K}\big((x-y)\Lambda\big)G_\Lambda^{cd}(x,y)f^{def}\Big)
	B^{fg}_\mu(y)\phi^g(y).
\end{equation}
Using standard methods for calculating diagrams with two integration operators, we obtain
\begin{equation}\label{33-j-5}
	\mathbb{H}_0^{\mathrm{sc}}(\mathrm{V}_1[\mathrm{K},\phi])\stackrel{\mathrm{s.p.}}{=}
	\bigg[-\frac{2L_1c_2}{\pi}J_1[B]+\frac{L_1^2c_2^2}{2\pi^2}S[B]\bigg]
	\int_{\mathrm{B}_1}\mathrm{d}^2x\,\mathrm{K}(x)+\varrho_1S[B],
\end{equation}
\begin{equation}\label{33-j-6}
	\mathbb{H}_0^{\mathrm{sc}}(\mathrm{V}_2[\mathrm{K},\phi])\stackrel{\mathrm{s.p.}}{=}
	\bigg[\frac{2L_1c_2}{\pi}(c_2J_1[B]-J_2[B])-\frac{L_1^2c_2^2}{\pi^2}S[B]\bigg]
	\int_{\mathrm{B}_1}\mathrm{d}^2x\,\mathrm{K}(x)+\varrho_2S[B],
\end{equation}
\begin{equation}\label{33-j-7}
	\mathbb{H}_0^{\mathrm{sc}}(\mathrm{V}_3[\mathrm{K},\phi])\stackrel{\mathrm{s.p.}}{=}
	\bigg[\frac{2L_1}{\pi}(-c_2^2J_1[B]-J_3[B])+\frac{L_1^2c_2^3}{\pi^2}S[B]\bigg]
	\int_{\mathrm{B}_1}\mathrm{d}^2x\,\mathrm{K}(x)+\varrho_3S[B],
\end{equation}
\begin{equation}\label{33-j-8}
	\mathbb{H}_0^{\mathrm{sc}}(\mathrm{V}_4[\mathrm{K},\phi])\stackrel{\mathrm{s.p.}}{=}
	\bigg[\frac{L_1}{\pi}(2c_2J_4[B]-J_2[B])-\frac{L_1^2c_2^2}{2\pi^2}S[B]\bigg]
	\int_{\mathrm{B}_1}\mathrm{d}^2x\,\mathrm{K}(x)+\varrho_4S[B],
\end{equation}
where the singular functions $\varrho_i$ are proportional to $L_1$ and independent of the background field. Combinations of such vertices can remove unnecessary nonlocal and non-linear terms, which is formulated and proved in Theorem \ref{33-tt}.

\section{Conclusion}
\label{33:sec:zak}

\subsection{Comments}

In this paper, three-loop diagrams for the two-dimensional non-linear sigma model were studied using the background field method and the cutoff regularization in the coordinate representation. All the necessary renormalization vertices were found, and the consistency of the results with the previously obtained ones was shown. The general formulation is presented in Section \ref{33:sec:r}.
Separately, it is worth paying attention to the fact that such studies are directly related to the generalization of the $\mathcal{R}$-operation for the case of the regularization under consideration and, using the example of a number of diagrams, completely coincide with it. As a comparison of the structure of the singularities, the standard cutoff in the momentum representation was considered. A special feature of this approach is the presence of a projector, due to the properties of which most of terms with differences are reduced.

\vspace{2mm}
\noindent{\textbf{Classical fields.}} As noted in Section \ref{33:sec:ps}, the renormalization procedure was considered in the context of the background field $B^a_\mu(x)$ satisfying the quantum equation of motion. A feature of this approach is not only the absence of the vertex $\Gamma_1$ in formal decompositions, and with it the negative powers of the coupling constant $\gamma$, but also the restriction of series \eqref{33-p-27-1} into a strongly connected part. However, it is worth noting that, according to the general idea, the renormalization procedure itself should be invariant with respect to the choice of an extreme trajectory for decomposition. In particular, the field $B^a_{\mathrm{cl},\mu}(x)$ can be selected as the background field, which satisfies only the classical equation of motion $\partial_{x_\mu}B^a_{\mathrm{cl},\mu}(x)=0$. In this case, the vertex with one external line will disappear again, and decomposition \eqref{33-p-27-1} will contain the operator $\mathbb{H}_0^{\mathrm{c}}$ instead of $\mathbb{H}_0^{\mathrm{sc}}$. Thus, additional diagrams will appear in the relations from \eqref{33-p-33}, which are connected, but not strongly connected. By direct calculation, we can make sure that they do not introduce other singularities, and the whole procedure is agreed upon.\footnote{In the general case, this is a consequence of the fact that the regularization preserves the connection between the quantum equation of motion and the quantum action, see Section 3.3 in \cite{sksk}.}.

Using the example of two-loop calculation, the remark looks the simplest. Indeed, in this case, an additional diagram appears in formula \eqref{33-p-30}, proportional to
\begin{equation*}
\mathbb{H}_0^{\mathrm{c}}\big(\mathbb{H}_1^{\mathrm{c}}(\Gamma_3)\mathbb{H}_1^{\mathrm{c}}(\Gamma_3)\big).
\end{equation*}
Using the representation from \eqref{33-w-4}, we note that the third and fourth terms are equal to zero due to the classical equation, and the first two lead to the relation
\begin{equation*}
\mathbb{H}_1^{\mathrm{c}}(\Gamma_3)
=
3\,{\centering\adjincludegraphics[width = 1.4 cm, valign=c]{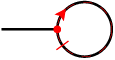}}
=
\int_{\mathbb{R}^2}\mathrm{d}^2x\,\phi^a(x)q^a(x)
,
\end{equation*}
in which the density $q^a(x)$ is finite and contains no singularities.

\vspace{2mm}
\noindent{\textbf{A shift of the nonlocal part.}} In work \cite{Kharuk-2021}, the shift of the Green's function for the four-dimensional Yang--Mills quantum theory to local\footnote{See the definition in \cite{30-1-1}.} zero modes was studied. It has been shown that the singular component of the two-loop part of the quantum action is invariant with respect to such shifts. In this paper, it was shown that a similar statement is true for the three-loop part of the two-dimensional non-linear sigma model. Indeed, shifting the "nonlocal" part by a sufficiently smooth function, which vanishes after applying thethe Laplace operator, will not introduce additional singularities. Note that auxiliary (counter) vertices that subtract singular "nonlocal" parts play an important role in this statement.

\vspace{2mm}
\noindent{\textbf{A logarithmic contribution.}} One of the directions of further work may be the calculation of an explicit formula for the logarithmic singularity in the third renormalization coefficient $a_2$ for the coupling constant, see Theorem \ref{33-t}. Given the results obtained in this work, we can use a special background field selection for the search. In particular, we can consider a constant field. Moreover, as a Green's function, we can use some other function that is similar only in the first orders when decomposed near the diagonal. This fact can greatly simplify the calculation, as we can add additional properties for convenience.

\vspace{2mm}
\noindent{\textbf{The quantum equation.}} It follows from Theorem \ref{33-t} that at the third step of the renormalization process, a triple vertex $\hat{\Gamma}_3$ appears, proportional to $\partial_{x_\mu}B^a_{\mu}(x)$. It is clear that such a vertex is zero on the set of solutions of the classical equation. As mentioned above, it is assumed that in higher orders of perturbative decomposition, parts will arise that will add up to a functional equal to zero on the background field. The study of the validity of this hypothesis can be noted as open tasks.

\vspace{2mm}
\noindent{\textbf{On the structure of singularities.}} In Theorem \ref{33-t}, it was shown that singularities of the form $\Lambda^2L$ do not appear in the renormalization coefficient and in the vertices. This fact is not trivial, since often in models, with an increase in the number of loops, the degree of logarithmic singularity also increases. It is clear that $\Lambda^2L$ was accumulated by the quartic vertex $\widetilde{\Gamma}_4$. Nevertheless, the question remains: at all steps of the renormalization process, will auxiliary vertices remove (by including inside) singularities of the form $\Lambda^2L^k$, where $k>0$?

\vspace{2mm}
\noindent{\textbf{About other regularizations.}} A separate interesting task is to study the relationship of the proposed cutoff in the coordinate representation with some other regularizations. It has been shown that the standerd cutoff in the momentum representation has important differences, primarily related to the absence of quasi-locality. In particular, some of the divergences have significantly different behavior in terms of the regularization parameter. Thus, a formal transition to the limit, at which the area of deformation of the free Green's function becomes infinitely large, is not correct.

\vspace{2mm}
\noindent{\textbf{About quasi-local vertices.}} The use of quasi-local vertices was discussed in Theorem \ref{33-t}. It was also noted that this method contains arbitrariness. Let us give another example. To do this, we define the quartic functional as follows
\begin{multline*}
	\mathrm{V}_5[\phi]=
	\int_{\mathbb{R}^2}\mathrm{d}^2x\int_{\mathbb{R}^2}\mathrm{d}^2y\,
	\phi^a(x)\phi^b(x)
	\Big(
	f^{ac_1g}A^{c_1c_2}(x)G_\Lambda^{c_2c_3}(x,y)f^{fc_3c}\\\times
	f^{bge_1}A^{e_3e_2}(y)G_\Lambda^{e_1e_2}(x,y)f^{fde_3}
	\Big)
	\phi^c(y)\phi^d(y).
\end{multline*}
Using the above methods, it is not difficult to prove a relation of the form
\begin{equation*}
	\mathrm{H}_{44}^{1,6}+t_1+t_2\stackrel{\mathrm{s.p.}}{=}8\mathbb{H}_{0}^{\mathrm{sc}}(\mathrm{V}_5),
\end{equation*}
where the notation from \eqref{33-d-119} and \eqref{33-k-3} were used. It is clear that such a nonlocal vertex leads to an alternative way to eliminate the power singularity $\Lambda^2$, however, at the same time, unlike the vertex $\widetilde{\Gamma}_4$ from \eqref{33-d-28-1-2}, also contains logarithmic terms with $\alpha_8$ and $\theta_3$, which exactly reduce the similar parts in the coefficients of $\tau_i$. The remaining parts of the numbers $\tau_i$ can be eliminated by vertices from Section \ref{33:sec:kva}.

\vspace{2mm}
\noindent{\textbf{About the non-linear sigma model.}} The main object of the study was a special case of two-dimensional non-linear sigma model, or the model of the principal chiral field. Thus, one of the possible generalization options is to move on to a more general non-linear sigma model.

\subsection{Acknowledgements}

Research of I.K. was conducted as part of the bachelor's scientific practice and is reflected in Sections \ref{33:sec:cl:3}, \ref{33:sec:cl:6}, and \ref{33:sec:appl:3}. The rest of the work (P.A. and A.I.) is supported by the Ministry of Science and Higher Education of the Russian
Federation in the framework of a scientific project under agreement № 075-15-2025-013. Also, A.I. is the winner of the "Young Mathematics of Russia" contest conducted by the Theoretical Physics and Mathematics Advancement Foundation "BASIS" and would like to thank its sponsors and jury.

\vspace{2mm}
\noindent With special warmth, A.V.Ivanov thanks N.V.Kharuk for helpful comments, K.A.Ivanov for creating a supportive work atmosphere, as well as his parents for the wonderful berry jam, which facilitated the lengthy writing process.

\vspace{2mm}
\textbf{Data availability statement.} Data sharing not applicable to this article as no datasets were generated or analysed during the current study.

\vspace{2mm}
\textbf{Conflict of interest statement.} The authors state that there is no conflict of interest.

\section{Appendix}
\label{33:sec:appl}
\subsection{Proof of Lemma \ref{33-lem-7}}
\label{33:sec:appl:2}
Relation \eqref{33-r-25} follows from the explicit calculation. Let us do it. From \eqref{33-r-23} we get
\begin{equation}\label{33-r-23-1}
	G_{2}^{\nu\mu}(0)=-\frac{\delta_{\mu\nu}}{2}
	\int_{\mathrm{B}_{1/\sigma}}\mathrm{d}^2y\int_{\mathrm{B}_{1/\sigma}}\mathrm{d}^2z\,
	G^{\Lambda,\mathbf{f}}(y)
	G^{\Lambda,\mathbf{f}}(y-z)
	A_0(z)G^{\Lambda,\mathbf{f}}(z)+\frac{\delta_{\mu\nu}}{16\pi\sigma^2}.
\end{equation}
Let us write the middle deformed function as  $G^{\Lambda,\mathbf{f}}=G^{\Lambda,\mathbf{f}}\pm G$, then the first part of the integral is
\begin{align*}
	\int_{\mathrm{B}_{1/\sigma}}\mathrm{d}^2y\int_{\mathrm{B}_{1/\sigma}}\mathrm{d}^2z\,
	G^{\Lambda,\mathbf{f}}(y)&
	\Big(G^{\Lambda,\mathbf{f}}(y-z)-G(y-z)\Big)
	A_0(z)G^{\Lambda,\mathbf{f}}(z)\\&=\frac{L}{2\pi\Lambda^2}
	\int_{\mathbb{R}^2}\mathrm{d}^2y\int_{\mathrm{B}_{1}}\mathrm{d}^2z\,
	\Big(G^{1,\mathbf{f}}(y-z)-G(y-z)\Big)
	A_0(z)G^{1,\mathbf{f}}(z)+\mathcal{O}(1/\Lambda^2)
	\\&=\frac{L}{2\pi\Lambda^2}
	\int_{\mathrm{B}_{1}}\mathrm{d}^2z\,
	\Big(G^{1,\mathbf{f}}(y)-G(y)\Big)+\mathcal{O}(1/\Lambda^2)
	\\&=\frac{L}{8\pi\Lambda^2}\bigg(
	\int_0^1\mathrm{d}^2s\,\mathbf{f}(s)-1\Bigg)+\mathcal{O}(1/\Lambda^2).
\end{align*}
In turn, the following chain of relations is true for the second part
\begin{align*}
	\int_{\mathrm{B}_{1/\sigma}}\mathrm{d}^2y\int_{\mathrm{B}_{1/\sigma}}\mathrm{d}^2z\,
	G^{\Lambda,\mathbf{f}}(y)&G(y-z)
	A_0(z)G^{\Lambda,\mathbf{f}}(z)\\&=\int_{\mathrm{B}_{1/\sigma}}\mathrm{d}^2y\int_{\mathrm{B}_{1/\Lambda}}\mathrm{d}^2z\,
	G^{\Lambda,\mathbf{f}}(y)G^{1/|y|,\mathbf{0}}(z)
	A_0(z)G^{\Lambda,\mathbf{f}}(z)
	\\&=\frac{1}{8\pi\sigma^2}-\frac{L}{4\pi\Lambda^2}+\frac{L}{8\pi\Lambda^2}\int_0^1\mathrm{d}s\,\mathbf{f}(s)\\&+\frac{L}{8\pi\Lambda^2}
	\int_{\mathrm{B}_{1}}\mathrm{d}^2y\,A_0(y)G^{1,\mathbf{f}}(y)\big(1-|y|^2\big)
	+\mathcal{O}(1/\Lambda^2).
\end{align*}
Finally, by describing the last integral and substituting all the parts in \eqref{33-r-23-1}, taking into account Definition \eqref{33-r-24-4}, we obtain the declared equality \eqref{33-r-25}.

\subsection{Auxiliary lemmas}
\label{33:sec:appl:1}
\begin{lemma}\label{33-lem-2} Let $v_1(\cdot)$ and $v_2(\cdot)$ be continuous functions on $\mathbb{R}^2$ with the Hölder exponent\footnote{Note that the introduction of the Hölder exponent is related to the fact that the functions in the work are not smooth and may have logarithmic behavior. In the statements below, the Hölder index is introduced for the same reasons.} $\alpha>0$, see \cite{He-1}, decreasing well enough at infinity (such that the integrals below converge). Then, taking into account the explicit form of regularization \eqref{33-o-55}, the following asymptotic expansions are valid for large values of the regularizing parameter $\Lambda$
	\begin{equation*}
		\int_{\mathbb{R}^2}\mathrm{d}^2x\int_{\mathbb{R}^2}\mathrm{d}^2y\,v_1(x)v_2(y)\Big(\partial_{x^{\nu}}G^{\Lambda,\mathbf{f}}(x-y)\Big)\Big(\partial_{y^{\mu}}G^{\Lambda,\mathbf{f}}(x-y)\Big)\stackrel{\mathrm{s.p.}}{=} -\frac{L}{4\pi}\int_{\mathbb{R}^{2}}\mathrm{d}^2x\,v_1(x)v_2(x)\delta_{\nu\mu}+\mathcal{O}(1),
	\end{equation*}
	\begin{equation*}
		\int_{\mathbb{R}^2}\mathrm{d}^2x\int_{\mathbb{R}^2}\mathrm{d}^2y\,v_1(x)v_2(y)\Big(\partial_{x^{\nu}}\partial_{y^{\mu}}G^{\Lambda,\mathbf{f}}(x-y)\Big)G^{\Lambda,\mathbf{f}}(x-y) \stackrel{\mathrm{s.p.}}{=} \frac{L}{4\pi}\int_{\mathbb{R}^{2}}\mathrm{d}^2x\,v_1(x)v_2(x)\delta_{\nu\mu}+\mathcal{O}(1).
	\end{equation*}
\end{lemma}
\begin{proof} Note that the integration domain $\mathbb{R}^2\times\mathbb{R}^2$ can be replaced by $\mathrm{B}_{1/\sigma}(y)\times\mathbb{R}^2$, where $\sigma$ is the mentioned above parameter. It is a finite fixed value, which, neglecting the finite integral, does not contain singularities with respect to the regularizing parameter $\Lambda$. In the remaining integral, we decompose one more time, then replace $v_1(x)\to v_1(y)$, again discarding the final parts, and also use a chain of relations
	\begin{equation}\label{33-r-5}
		\int_{\mathrm{B}_{1/\sigma}(y)}\mathrm{d}^2y\,\Big(\partial_{x^\nu}G^{\Lambda,\mathbf{f}}(x-y)\Big)
		\Big(\partial_{y^\mu}G^{\Lambda,\mathbf{f}}(x-y)\Big)\stackrel{\mathrm{s.p.}}{=}
		\frac{-1}{4\pi^2}\int_{1/\Lambda\leqslant|x|\leqslant1/\sigma}
		\mathrm{d}^2x\,\frac{x^\nu x^\mu}{|x|^4}\stackrel{\mathrm{s.p.}}{=}-
		\frac{L}{4\pi}\delta_{\mu\nu},
	\end{equation}
	where the replacement of the variable $x\to x+y$ was additionally used. This leads to the stated result. The second integral is analyzed similarly.
\end{proof}

\begin{lemma}\label{33-lem-4} Let $v_1^{ab}(\cdot)$ and $v_2^{ab}(\cdot)$ be matrix-valued continuously differentiable functions on $\mathbb{R}^2$, decreasing well enough at infinity (such that the integrals below converge). Additionally, we require that the first partial derivatives have the Hölder exponent $\alpha>0$, see \cite{He-1}. Then, taking into account the explicit form of regularization \eqref{33-o-55}, the following asymptotic expansions are valid for large values of the regularizing parameter $\Lambda$
	\begin{multline}\label{33-r-8}
		\int_{\mathbb{R}^2}\mathrm{d}^2x\int_{\mathbb{R}^2}\mathrm{d}^2y\,
		v_1^{ab}(x)
		\Big(\partial_{y^\mu}G_\Lambda^{bc}(x,y)\Big)
		\Big(\partial_{x_\nu}D_\nu^{ad}(x)G_\Lambda^{de}(x,y)\Big)
		v_2^{ec}(y)
		\stackrel{\mathrm{s.p.}}{=}\\\stackrel{\mathrm{s.p.}}{=} \frac{L}{\pi}
		\int_{\mathbb{R}^2}\mathrm{d}^2y\,v_1^{ab}(y)B_\mu^{bc}(y)v_2^{ac}(y)+\mathcal{O}(1),
	\end{multline}
	\begin{multline}\label{33-r-9}
		\int_{\mathbb{R}^2}\mathrm{d}^2x\int_{\mathbb{R}^2}\mathrm{d}^2y\,
		v_1^{ab}(x)
		\Big(G_\Lambda^{bg}(x,y)B_\mu^{gc}(y)\Big)
		\Big(\partial_{x_\nu}D_\nu^{ad}(x)G_\Lambda^{de}(x,y)\Big)
		v_2^{ec}(y)
		\stackrel{\mathrm{s.p.}}{=}\\\stackrel{\mathrm{s.p.}}{=} -\frac{2L}{\pi}
		\int_{\mathbb{R}^2}\mathrm{d}^2y\,v_1^{ab}(y)B_\mu^{bc}(y)v_2^{ac}(y)+\mathcal{O}(1),
	\end{multline}
	\begin{multline}\label{33-r-10}
		\int_{\mathbb{R}^2}\mathrm{d}^2x\int_{\mathbb{R}^2}\mathrm{d}^2y\,
		v_1^{ab}(x)
		\Big(G_\Lambda^{bg}(x,y)\overleftarrow{D}_\mu^{gc}(y)\Big)
		\Big(\partial_{x_\nu}D_\nu^{ad}(x)G_\Lambda^{de}(x,y)\Big)
		v_2^{ec}(y)
		\stackrel{\mathrm{s.p.}}{=}\\\stackrel{\mathrm{s.p.}}{=} -\frac{L}{\pi}
		\int_{\mathbb{R}^2}\mathrm{d}^2y\,v_1^{ab}(y)B_\mu^{bc}(y)v_2^{ac}(y)+\mathcal{O}(1).
	\end{multline}
	In this case, the operator $\partial_{x_\nu}D_\nu^{ad}(x)$ in the left parts can be replaced by aither $D_\nu^{ad}(x)\partial_{x_\nu}$ or $-2A^{ad}(x)$.
\end{lemma}
\begin{proof} Consider the left hand side of the first relation. Taking into account the expansions near the diagonal \eqref{33-o-5} and \eqref{33-r-1}, only the following three terms can have a singular contribution
	\begin{equation*}
		4\int_{\mathbb{R}^2}\mathrm{d}^2x\int_{\mathbb{R}^2}\mathrm{d}^2y\,
		v_1^{ab}(x)
		\Big(\partial_{y^\mu}G^{\Lambda,\mathbf{f}}(x-y)\Big)
		\Big(\partial_{x_\nu}\partial_{x^\nu}G^{\Lambda,\mathbf{f}}(x-y)\Big)
		v_2^{ab}(y),
	\end{equation*}
	\begin{equation*}
		4\int_{\mathbb{R}^2}\mathrm{d}^2x\int_{\mathbb{R}^2}\mathrm{d}^2y\,
		v_1^{ab}(x)
		\Big(\partial_{y^\mu}G^{\Lambda,\mathbf{f}}(x-y)\Big)
		\Big(\partial_{x_\nu}\partial_{x^\nu}\eta_1^{ae}(x,y)-\partial_{x_\nu}B_\nu^{ae}(x)G^{\Lambda,\mathbf{f}}(x-y)\Big)
		v_2^{eb}(y),
	\end{equation*}
	\begin{equation}\label{33-r-7}
		4\int_{\mathbb{R}^2}\mathrm{d}^2x\int_{\mathbb{R}^2}\mathrm{d}^2y\,
		v_1^{ab}(x)
		\Big(\partial_{y^\mu}\eta_1^{bc}(x,y)\Big)
		\Big(\partial_{x_\nu}\partial_{x^\nu}G^{\Lambda,\mathbf{f}}(x-y)\Big)
		v_2^{ac}(y).
	\end{equation}
	The first two are convergent integrals when regularization is removed and do not contain a singular component, since the following relations are fulfilled
	\begin{equation*}
		\int_{\mathrm{B}_{1/\sigma}}\mathrm{d}^2x\,
		\Big(\partial_{x^\mu}G^{\Lambda,\mathbf{f}}(x)\Big)
		\Big(\partial_{x_\nu}\partial_{x^\nu}G^{\Lambda,\mathbf{f}}(x)\Big)=0,
	\end{equation*}
	\begin{equation}\label{33-r-6}
		\int_{\mathrm{B}_{1/\sigma}}\mathrm{d}^2x\,
		\Big(\partial_{x^\mu}G^{\Lambda,\mathbf{f}}(x)\Big)
		\bigg(-\partial_{x_\nu}\partial_{x^\nu}
		\int_{\mathrm{B}_{1/\sigma}}\mathrm{d}^2z\,
		G^{\Lambda,\mathbf{f}}(x-z)\partial_{z^\rho}G^{\Lambda,\mathbf{f}}(z)
		-\partial_{x^\rho}G^{\Lambda,\mathbf{f}}(x)\bigg)\stackrel{\mathrm{s.p.}}{=}0.
	\end{equation}
	Note that the latter equality is proved by integration by parts. Indeed, using spherical averaging and the fact that the support of the function $\partial_{x_\nu}\partial_{x^\nu}G^{\Lambda,\mathbf{f}}(x)$ is located in the ball $\mathrm{B}_{1/\Lambda}$, we get
	\begin{align*}
		\eqref{33-r-6}&\stackrel{\mathrm{s.p.}}{=}\frac{\delta_{\mu\rho}}{2}
		\int_{\mathrm{B}_{1/\Lambda}}\mathrm{d}^2x\,
		\Big(\partial_{x_\eta}\partial_{x^\eta}G^{\Lambda,\mathbf{f}}(x)\Big)
		\bigg(\partial_{x_\nu}\partial_{x^\nu}
		\int_{\mathrm{B}_{1/\sigma}}\mathrm{d}^2z\,
		G^{\Lambda,\mathbf{f}}(x-z)G^{\Lambda,\mathbf{f}}(z)
		+G^{\Lambda,\mathbf{f}}(x)\bigg)\\&\stackrel{\mathrm{s.p.}}{=}
		\frac{\delta_{\mu\rho}}{2}
		\int_{\mathrm{B}_{1/\Lambda}}\mathrm{d}^2x\,
		\Big(\partial_{x_\eta}\partial_{x^\eta}G^{\Lambda,\mathbf{f}}(x)\Big)
		\bigg(
		\int_{\mathrm{B}_{1/\Lambda}}\mathrm{d}^2z\,
		\Big(\partial_{z_\eta}\partial_{z^\eta}G^{\Lambda,\mathbf{f}}(z)\Big)G^{\Lambda,\mathbf{f}}(x-z)
		+\frac{\ln(\Lambda/\sigma)}{2\pi}\bigg)\\&\stackrel{\mathrm{s.p.}}{=}
		\frac{\delta_{\mu\rho}}{2}\frac{\ln(\Lambda/\sigma)}{2\pi}
		\int_{\mathrm{B}_{1}}\mathrm{d}^2x\,
		\Big(\partial_{x_\eta}\partial_{x^\eta}G^{1,\mathbf{f}}(x)\Big)
		\bigg(
		\int_{\mathrm{B}_{1}}\mathrm{d}^2z\,
		\Big(\partial_{z_\eta}\partial_{z^\eta}G^{1,\mathbf{f}}(z)\Big)
		+1\bigg)=0
		.
	\end{align*}
	Thus, only the third part of \eqref{33-r-7} is of interest. Let us rewrite it in more explicit equivalent terms, taking into account factorization after shifting the variable $x\to x+y$
	\begin{equation*}
		4\int_{\mathrm{B}_{1/\sigma}}\mathrm{d}^2x\,
		\big(\partial_{x^\mu}\int_{\mathrm{B}_{1/\sigma}}\mathrm{d}^2z\,
		G^{\Lambda,\mathbf{f}}(x-z)\partial_{z_\rho}G^{\Lambda,\mathbf{f}}(z)\bigg)
		\Big(\partial_{x_\nu}\partial_{x^\nu}G^{\Lambda,\mathbf{f}}(x)\Big)
		\bigg(\int_{\mathbb{R}^2}\mathrm{d}^2y\,v_1^{ab}(y)B_\rho^{bc}(y)v_2^{ac}(y)\bigg).
	\end{equation*}
	Next, integrating by parts, using spherical averaging, and again taking into account the size of the function support $\partial_{x_\nu}\partial_{x^\nu}G^{\Lambda,\mathbf{f}}(x)$, we obtain
	\begin{equation*}
		2\int_{\mathrm{B}_{1/\Lambda}}\mathrm{d}^2x\,
		\bigg(\int_{\mathrm{B}_{1/\Lambda}}\mathrm{d}^2z\,
		G^{\Lambda,\mathbf{f}}(x-z)\partial_{z_\rho}\partial_{z^\rho}G^{\Lambda,\mathbf{f}}(z)\bigg)
		\Big(\partial_{x_\nu}\partial_{x^\nu}G^{\Lambda,\mathbf{f}}(x)\Big)
		\bigg(\int_{\mathbb{R}^2}\mathrm{d}^2y\,v_1^{ab}(y)B_\mu^{bc}(y)v_2^{ac}(y)\bigg).
	\end{equation*}
	Then, by scaling $x\to x/\Lambda$ and $z\to z/\Lambda$ and using the representation from \eqref{33-o-55}, we arrive at the formula
	\begin{equation*}
		\frac{\ln(\Lambda/\sigma)}{\pi}
		\bigg(\int_{\mathrm{B}_{1}}\mathrm{d}^2z\,\partial_{z_\rho}\partial_{z^\rho}G^{1,\mathbf{f}}(z)\bigg)^2
		\bigg(\int_{\mathbb{R}^2}\mathrm{d}^2y\,v_1^{ab}(y)B_\mu^{bc}(y)v_2^{ac}(y)\bigg),
	\end{equation*}
hence the stated statement. Similarly, the validity of the second relation \eqref{33-r-9} can be shown. The equality \eqref{33-r-10} is the sum of the previous two. To conclude the proof, we add that in the left hand sides the operator $\partial_{x_\nu}D_\nu^{ad}(x)$ can be replaced by $-2A^{ad}(x)$ from \eqref{33-o-3} or by $D_\nu^{ad}(x)\partial_{x_\nu}$ by adding a term proportional to $-\partial_{x^\nu}B_\nu^{ad}(x)$. Such addition does not lead to additional singularities, since the auxiliary integral is convergent.
\end{proof}

\begin{lemma}\label{33-lem-5} Let $v_1^{ab}(\cdot)$ and $v_2^{ab}(\cdot)$ be matrix-valued doubly continuously differentiable functions on $\mathbb{R}^2$, decreasing well enough at infinity (such that the integrals below converge). Additionally, we require that the second partial derivatives have the Hölder exponent $\alpha>0$, see \cite{He-1}. Then, taking into account the explicit form of regularization \eqref{33-o-55}, the following asymptotic expansions are valid for large values of the regularizing parameter $\Lambda$
	\begin{multline}\label{33-r-11}
		\int_{\mathbb{R}^2}\mathrm{d}^2x\int_{\mathbb{R}^2}\mathrm{d}^2y\,
		v_1^{ab}(x)
		\Big(A^{cg}(y)G_\Lambda^{bg}(x,y)\Big)
		\Big(A^{ad}(x)G_\Lambda^{de}(x,y)\Big)
		v_2^{ec}(y)
		\stackrel{\mathrm{s.p.}}{=}\\\stackrel{\mathrm{s.p.}}{=}\Lambda^2
		\bigg(\int_{\mathbb{R}^2}\mathrm{d}^2x\Big(A_0(x)G^{1,\mathbf{f}}(x)\Big)^2\bigg)
		\int_{\mathbb{R}^2}\mathrm{d}^2y\,
		v_1^{ab}(y)
		v_2^{ab}(y)+\mathcal{O}(1).
	\end{multline}
\end{lemma}
\begin{proof} Let us introduce a convenient representation\footnote{It follows from formula \eqref{33-o-5}. It is possible to check by substitution.} for the Green's function of the form
	\begin{equation*}
		G_\Lambda^{de}(x,y)=2G^{\Lambda,\mathbf{f}}(x-y)\delta^{de}-
		\int_{\mathbb{R}^2}\mathrm{d}^2z\,\rho_1(x-z)V^{df}(z)G_\Lambda^{fe}(z,y)+
		\int_{\mathbb{R}^2}\mathrm{d}^2z\,\frac{\ln(|x-z|\sigma)}{2\pi}V^{df}(z)G_\Lambda^{fe}(z,y),
	\end{equation*}
	where
	\begin{equation*}
		\rho_1(x)=G^{\Lambda,\mathbf{f}}(x)+\frac{\ln(|x|\sigma)}{2\pi}.
	\end{equation*}
	Note that the support of $\rho_1(x)$ is located in the ball $\mathrm{B}_{1/\Lambda}$, and the function itself has a continuous first derivative in $\mathbb{R}^2\setminus\{0\}$. Next, we apply the operator $A^{ad}(x)=A_0(x)\delta^{ad}/2+V^{ad}(x)/2$ to the new representation, then we get
	\begin{equation*}
		A^{ad}(x)G_\Lambda^{de}(x,y)=A_0(x)G^{\Lambda,\mathbf{f}}(x-y)\delta^{ae}+\rho_2^{ae}(x,y).
	\end{equation*}
	where
	\begin{equation*}
		\rho_2^{ae}(x,y)=-\frac{1}{2}
		\int_{\mathbb{R}^2}\mathrm{d}^2z\,\Big(A_0(x)\rho_1(x-z)\Big)V^{af}(z)G_\Lambda^{fe}(z,y).
	\end{equation*}
	Let us show that the combinations
	\begin{equation}\label{33-r-12}
		\int_{\mathbb{R}^2}\mathrm{d}^2x\int_{\mathbb{R}^2}\mathrm{d}^2y\,v_1^{ab}(x)
		\Big(\rho_2^{cb}(y,x)\Big)
		\Big(\rho_2^{ae}(x,y)\Big)
		v_2^{ec}(y),
	\end{equation}
	\begin{equation}\label{33-r-13}
		\int_{\mathbb{R}^2}\mathrm{d}^2x\int_{\mathbb{R}^2}\mathrm{d}^2y\,v_1^{ab}(x)
		\Big(A_0(x)G^{\Lambda,\mathbf{f}}(x-y)\Big)
		\Big(\rho_2^{ae}(x,y)\Big)v_2^{eb}(y),
	\end{equation}
	do not contain singular contributions. To prove this fact, it is easiest to use the perturbative decomposition again and consider the most singular contribution. In the first case, there is only one. After shifting the variable and factoring, the interesting part of the integral can be represented as
	\begin{equation*}
		\int_{\mathrm{B}_{1/\sigma}}\mathrm{d}^2x
		\bigg(\int_{\mathrm{B}_{1/\sigma}}\mathrm{d}^2z\,A_0(x)\rho_1(x-z)\partial_{z^\mu}G^{\Lambda,\mathbf{f}}(z)\bigg)
		\bigg(\int_{\mathrm{B}_{1/\sigma}}\mathrm{d}^2y\,A_0(x)\rho_1(x-y)\partial_{y^\nu}G^{\Lambda,\mathbf{f}}(y)\bigg).
	\end{equation*}
	After integration by parts, as well as using spherical symmetry and constraints on the support, we obtain 
	\begin{equation*}
		\frac{\delta_{\mu\nu}}{2}\int_{\mathrm{B}_{2/\Lambda}}\mathrm{d}^2x
		\bigg(\int_{\mathrm{B}_{3/\Lambda}}\mathrm{d}^2z\,A_0(x)\rho_1(x-z)G^{\Lambda,\mathbf{f}}(z)\bigg)
		\bigg(\int_{\mathrm{B}_{1/\Lambda}}\mathrm{d}^2y\,A_0(x)\rho_1(x-y)A(y)G^{\Lambda,\mathbf{f}}(y)\bigg).
	\end{equation*}
	It is clear that the possible singular contribution is obtained after scaling the variables by substitution $G^{\Lambda,\mathbf{f}}(z)\to\ln(\Lambda/\sigma)/(2\pi)$ and is equal to
	\begin{equation*}
		\frac{\delta_{\mu\nu}\ln(\Lambda/\sigma)}{4\pi}\int_{\mathrm{B}_{2/\Lambda}}\mathrm{d}^2x
		\bigg(\int_{\mathrm{B}_{3/\Lambda}}\mathrm{d}^2z\,A_0(x)\rho_1(x-z)\bigg)
		\bigg(\int_{\mathrm{B}_{1/\Lambda}}\mathrm{d}^2y\,A_0(x)\rho_1(x-y)A(y)G^{\Lambda,\mathbf{f}}(y)\bigg).
	\end{equation*}
	The last integral is zero, since
	\begin{equation*}
		\int_{\mathrm{B}_{1/\Lambda}}\mathrm{d}^2z\,A_0(z)\rho_1(z)=0.
	\end{equation*}
	Thus, \eqref{33-r-12} does not contain singular contributions. Similar arguments are valid for the integral from \eqref{33-r-13}. Therefore, the main part of the asymptotic decomposition of the left side of \eqref{33-r-11} can be obtained by replacing $A^{ad}(x)G_\Lambda^{de}(x,y)$ and $A^{cg}(y)G_\Lambda^{bg}(x,y)$ by $A_0(x)G^{\Lambda,\mathbf{f}}(x-y)\delta^{ae}$ and $A_0(x)G^{\Lambda,\mathbf{f}}(x-y)\delta^{cg}$. This is what the stated statement follows from.
\end{proof}

\subsection{Integral estimates}
\label{33:sec:appl:3}
\noindent\textbf{The first integral.} Consider a value of the form
\begin{equation*}
	\mathrm{Q}_1=\int_{\mathrm{B}_1}\mathrm{d}^2x\int_{\mathrm{B}_1}\mathrm{d}^2y\int_{\mathrm{B}_1}\mathrm{d}^{2}z\,\Big(\partial_{x^{\mu}}\ln{|x-y|}\Big)\ln{|x-z|}\Big(\partial_{z^{\sigma}}\ln{|y-z|}\Big)\Big(\partial_{x^{\rho}}\ln{|x|}\Big)\ln{|z|}\Big(\partial_{y^{\nu}}\ln{|y|}\Big).
\end{equation*}
Let us define the auxiliary number
\begin{equation*}
	M_1=\max_{x,y\in\mathrm{B}_1}\bigg(
	\int_{\mathrm{B}_1}\mathrm{d}^{2}z\,\Big|\ln{|x-z|}\Big(\partial_{z^{\sigma}}\ln{|y-z|}\Big)\ln{|z|}\Big|\bigg)<+\infty.
\end{equation*}
Further, given the fact that $\big|\partial_{y^{\nu}}\ln{|y|}\big|\leqslant1/|y|$, the integral can be estimated from above by
\begin{equation*}
	\big|\mathrm{Q}_1\big|\leqslant M_1\int_{\mathrm{B}_1}\mathrm{d}^2x\int_{\mathrm{B}_1}\mathrm{d}^2y\,
	\frac{1}{|x|\cdot|x-y|\cdot|y|},
\end{equation*}
from which the finitness follows.\\

\noindent\textbf{The second integral.} Next, consider the value
\begin{equation*}
	\mathrm{Q}_2=\int_{\mathrm{B}_1}\mathrm{d}^2x\int_{\mathrm{B}_1}\mathrm{d}^2y\int_{\mathrm{B}_1}\mathrm{d}^{2}z\,\Big(\partial_{x^{\mu}}\ln{|x-y|}\Big)\ln{|x-z|}\Big(\partial_{z^{\sigma}}\ln{|y-z|}\Big)\Big(\partial_{x^{\rho}}\ln{|x|}\Big)\Big(\partial_{z^{\nu}}\ln{|z|}\Big)\ln{|y|}.
\end{equation*}
It is clear that its absolute value can be estimated by the integral
\begin{equation*}
|Q_2|\leqslant
\int_{\mathrm{B}_{1}}\mathrm{d}^2z\int_{\mathrm{B}_{1}}\mathrm{d}^2x\int_{\mathrm{B}_{1}}\mathrm{d}^2y\,\frac{|\ln|x-z||\cdot|\ln|y||}{|x|\cdot|z|\cdot|x-y|\cdot|y-z|}
\equiv Q_3.
\end{equation*}
Note that the integral function is invariant with respect to the replacement $x\longleftrightarrow z$, therefore, dividing the integration domain for the variable $x$
\begin{equation*}
\mathrm{B}_1=\mathrm{B}_{|z|}\cup\big(\mathrm{B}_1\setminus\mathrm{B}_{|z|}),
\end{equation*}
changing the order of integration in one of the terms
\begin{equation*}
\int_{\mathrm{B}_{1}}\mathrm{d}^2z\int_{\mathrm{B}_1\setminus\mathrm{B}_{|z|}}\mathrm{d}^2x
\longrightarrow
\int_{\mathrm{B}_{1}}\mathrm{d}^2x\int_{\mathrm{B}_{|x|}}\mathrm{d}^2z,
\end{equation*}
and in it, reassigning the variables $x\longleftrightarrow z$, we get the equality
\begin{equation*}
Q_3=2\int_{\mathrm{B}_{1}}\mathrm{d}^2z\int_{\mathrm{B}_{|z|}}\mathrm{d}^2x\int_{\mathrm{B}_{1}}\mathrm{d}^2y\,\frac{|\ln|x-z||\cdot|\ln|y||}{|x|\cdot|z|\cdot|x-y|\cdot|y-z|}.
\end{equation*}
Consider an auxiliary function of the form 
\begin{equation*}
\mathrm{I}(x,z)=\int_{\mathrm{B}_{1}}\mathrm{d}^2y\,\frac{|\ln|y||}{|x-y|\cdot|y-z|}.
\end{equation*}
Let $z\in\mathrm{B}_{1}$ and $x\in\mathrm{B}_{|z|}$, $M>1$, as well as $r(M,x)=|x|/M$. Then we can write out the relations 
\begin{align*}
\mathrm{I}(x,z)=& \int_{\mathrm{B}_{1}\setminus\mathrm{B}_{r(M,x)}}\mathrm{d}^2y\,(\ldots)
+\int_{\mathrm{B}_{r(M,x)}}\mathrm{d}^2y\,(\ldots)\\
\leqslant&
|\ln(r(M,x))|\int_{\mathrm{B}_{1}}\mathrm{d}^2y\,
\frac{1}{|x-y||y-z|}+\left(\frac{M}{M-1}\right)^{2}\frac{1}{|x|^{2}}\int_{\mathrm{B}_{r(M,x)}}\mathrm{d}^2y\,|\ln|y||.
\end{align*}
For the first integral in the expression above, an estimate can be obtained using Theorem 3 from paragraph $115$ (chapter IV) of the monograph \cite{smir}, while the second one is calculated explicitly. Thus, we come to an estimate of the form
\begin{equation*}
\mathrm{I}(x,z)\leqslant
A_1|\ln(r(M,x))|(1+|\ln|x-z||)+\frac{A_2}{(M-1)^{2}}(2|\ln(r(M,x))|+1),
\end{equation*}
where $A_1$ and $A_2$ are numbers independent of $M$. Finally, writing out the auxiliary integral in the form
\begin{equation*}
Q_3=2\int_{\mathrm{B}_{1}}\mathrm{d}^2z\int_{\mathrm{B}_{|z|}}\mathrm{d}^2x\,\frac{|\ln|x-z||}{|x|\cdot|z|}\,\mathrm{I}(x,z),
\end{equation*}
note that due to the logarithmic behavior of the function $\mathrm{I}(x,z)$, all integrals are convergent. Thus, the convergence of $Q_2$ is shown.

\subsection{Special relations}
\label{33:sec:sp}
\noindent\textbf{Property 1.} Let $G^{1,\mathbf{f}}(x)$ be a deformation of the free fundamental solution from \eqref{33-o-55-1}, as well as $x\in\mathrm{B}_1$, then the relation holds
\begin{equation}\label{33-sp-1}
z_1(x)=
\int_{\mathrm{B}_2}\mathrm{d}^2y\,\Big(G^{1,\mathbf{f}}(x+y)-G^{1,\mathbf{f}}(y)\Big)=-
\frac{|x|^2}{4}.
\end{equation}
Let us use addition and subtraction, then we can rewrite
\begin{equation}\label{33-sp-2}
z_1(x)=
\int_{\mathrm{B}_2}\mathrm{d}^2y\,\Big(G^{1,\mathbf{f}}(x+y)-G(x+y)+G(x+y)-G^{1,\mathbf{f}}(y)\Big),
\end{equation}
where $G(x)=-\ln(|x|\sigma)/(2\pi)$. Next, note that the sum of the first two functions is constant at $x\in\mathrm{B}_1$. Indeed, the support of the difference $G^{1,\mathbf{f}}(\cdot)-G(\cdot)$ is contained in the ball $\mathrm{B}_1$, hence
\begin{equation}\label{33-sp-3}
\partial_{x_\mu}\int_{\mathrm{B}_2}\mathrm{d}^2y\,\Big(G^{1,\mathbf{f}}(x+y)-G(x+y)\Big)=
\int_{\partial\mathrm{B}_2}\mathrm{d}\sigma(y)\,\mathbf{n}^\mu(y)\Big(G^{1,\mathbf{f}}(x+y)-G(x+y)\Big)
=0,
\end{equation}
where $\mathrm{d}\sigma(y)$ is a standard measure on a sphere, and $\mathbf{n}^\mu(y)$ is the external normal to the boundary. Thus, choosing $x=0$ in the first two terms, we get
\begin{equation}\label{33-sp-4}
z_1(x)=
\int_{\mathrm{B}_2}\mathrm{d}^2y\,\Big(G(x+y)-G(y)\Big)=-\frac{1}{4}|x|^2,
\end{equation}
where the relation was used in the calculation, see \cite{Ivanov-2022}, 
\begin{equation}\label{33-sp-5}
\int_{\partial\mathrm{B}_1}\mathrm{d}\sigma(\hat{y})\,G(x+|y|\hat{y})=-\frac{1}{2\pi}
\begin{cases}
\ln(|y|\sigma),&|y|\geqslant|x|;\\
\ln(|x|\sigma),&|y|<|x|,
\end{cases}
=G^{1/|y|,\mathbf{0}}(x).
\end{equation}

\noindent\textbf{Property 2.} Let the function $G_1^\mu(\cdot)$ be defined by equality \eqref{33-r-21}, and the relation $x\in\mathrm{B}_1$ is also fulfilled, then the equality is true 
\begin{equation}\label{33-sp-6}
\Lambda G_1^\mu(x/\Lambda)=\frac{x^\mu \ln(\sigma\Lambda/\hat{\sigma})}{4\pi}-
\partial_{x_{\mu}}
\int_{\mathrm{B}_{2}}\mathrm{d}^2y\,
G^{1,\mathbf{f}}(x-y)
G^{1,\mathbf{f}}(y),
\end{equation}
where the second term is a correction to the first in the sense of expansion by the regularizing parameter $\Lambda$. To prove this, we use the definition, scaling of the variable, and integration by parts, then we get
\begin{align*}
\Lambda G_1^\mu(x/\Lambda)=&-
\Lambda\int_{\mathrm{B}_{1/\hat{\sigma}}}\mathrm{d}^2y\,
G^{\Lambda,\mathbf{f}}(x/\Lambda-y)
\partial_{y_{\mu}}G^{\Lambda,\mathbf{f}}(y)\\&-\partial_{x_{\mu}}
\int_{\mathrm{B}_{\Lambda/\hat{\sigma}}}\mathrm{d}^2y\,
\bigg(G^{1,\mathbf{f}}(x-y)+\frac{\ln(\Lambda)}{2\pi}\bigg)
\bigg(G^{1,\mathbf{f}}(y)+\frac{\ln(\sigma\Lambda/\hat{\sigma})}{2\pi}\bigg).
\end{align*}
Next, note that $\ln(\Lambda)/(2\pi)$ can be removed from the first multiplier, since in this case the differentiation of the constant will result. It is also important to note that the function
\begin{equation}\label{33-sp-8}
\partial_{x_{\mu}}
\int_{\mathrm{B}_{\Lambda/\hat{\sigma}}}\mathrm{d}^2y\,
G^{1,\mathbf{f}}(x-y)
G^{1,\mathbf{f}}(y)=
\partial_{x_{\mu}}
\int_{\mathrm{B}_{2}}\mathrm{d}^2y\,
G^{1,\mathbf{f}}(x-y)
G^{1,\mathbf{f}}(y)
\end{equation}
is finite due to the support of the averaged logarithm and, thus, falls into the correction term. The remaining part is instantly calculated using the first property from this section. Additionally, we note several auxiliary integrals in which only the main part of the asymptotics is used:
\begin{equation*}\label{33-sp-9}
\int_{\mathrm{B}_{1/\Lambda}}\mathrm{d}^2x\,\Big(A_0(x)G^{\Lambda,\mathbf{f}}(x)\Big)^2x_\nu 
G_1^\nu(x)=\frac{L\alpha_8(\mathbf{f})}{4\pi^2}+\mathcal{O}(1),
\end{equation*}
\begin{equation*}
\int_{\mathrm{B}_{1/\Lambda}}\mathrm{d}^2x\,\Big(A_0(x)G^{\Lambda,\mathbf{f}}(x)\Big)
G_1^\nu(x)\partial_{x^\nu}
\bigg(G^{\Lambda,\mathbf{f}}(x)-
\int_{\mathrm{B}_{1/\Lambda}}\mathrm{d}^2y\,G^{\Lambda,\mathbf{f}}(x-y)A_0(y)G^{\Lambda,\mathbf{f}}(y)\bigg)
=\frac{L\alpha_9(\mathbf{f})}{4\pi^2}+\mathcal{O}(1),
\end{equation*}
\begin{equation*}
\int_{\mathrm{B}_{1/\sigma}}\mathrm{d}^2x\,\Big(\partial^\mu G^{\Lambda,\mathbf{f}}(x)\Big)
\Big(\partial_\mu G_1^\nu(x)\Big)
\partial_{x^\nu}
\bigg(G^{\Lambda,\mathbf{f}}(x)-
\int_{\mathrm{B}_{1/\Lambda}}\mathrm{d}^2y\,G^{\Lambda,\mathbf{f}}(x-y)A_0(y)G^{\Lambda,\mathbf{f}}(y)\bigg)
=-\frac{L\theta_2}{8\pi^2}+\mathcal{O}(1),
\end{equation*}
\begin{equation*}\label{33-sp-11}
\int_{\mathrm{B}_{1/\Lambda}}\mathrm{d}^2x\,\Big(A_0(x)G^{\Lambda,\mathbf{f}}(x)\Big)
\Big(x_\nu\partial_{x_\mu}G^{\Lambda,\mathbf{f}}(x)\Big)
\Big(\partial_{x^\mu}G_1^\nu(x)\Big)
=\frac{L\alpha_{11}(\mathbf{f})}{4\pi^2}+\mathcal{O}(1),
\end{equation*}
\begin{equation*}\label{33-sp-12}
\int_{\mathrm{B}_{1/\Lambda}}\mathrm{d}^2x\,\Big(A_0(x)G^{\Lambda,\mathbf{f}}(x)\Big)
\Big(\partial_{x_\mu}G^{\Lambda,\mathbf{f}}(x)\Big)
G_1^\mu(x)
=\frac{L\alpha_{11}(\mathbf{f})}{4\pi^2}+\mathcal{O}(1).
\end{equation*}
As well as two integrals that use the correction term:
\begin{equation*}
\int_{\mathrm{B}_{1/\Lambda}}\mathrm{d}^2x\,\Big(A_0(x)G^{\Lambda,\mathbf{f}}(x)\Big)
\Big(\partial_{x_\mu}G_1^\nu(x)\Big)\Big(\partial_{x^\mu}G_1^\nu(x)\Big)
=\frac{L_1^2+2L\theta_2}{8\pi^2}+\mathcal{O}(1),
\end{equation*}
\begin{equation*}
\int_{\mathrm{B}_{1/\Lambda}}\mathrm{d}^2x\,\Big(A_0(x)G^{\Lambda,\mathbf{f}}(x)\Big)
\Big(\partial_{x_\mu}G_1^\mu(x)\Big)G^{\Lambda,\mathbf{f}}(x)
=\frac{L_1^2+L(\theta_2+\theta_1)}{4\pi^2}+\mathcal{O}(1).
\end{equation*}

\noindent\textbf{Property 3.} Let the function $G_2^{\nu\mu}(\cdot)$ be defined according to formula \eqref{33-r-23}. Let also $x\in\mathrm{B}_1$, then the relation is true
\begin{equation}\label{33-sp-13}
\Lambda^2G_2^{\mu\mu}(x/\Lambda)-\Lambda^2G_2^{\mu\mu}(0)=\frac{|x|^2L}{8\pi}+\mathcal{O}(1).
\end{equation}
To prove this, we use the definition and scaling of the integration variables, then we get
\begin{multline*}
\int_{\mathrm{B}_{\Lambda/\sigma}}\mathrm{d}^2y\int_{\mathrm{B}_{1}}\mathrm{d}^2z\,
\bigg(G^{1,\mathbf{f}}(y)+\frac{\ln(\Lambda)}{2\pi}\bigg)
\bigg(G^{1,\mathbf{f}}(y-z)+\frac{\ln(\Lambda)}{2\pi}\bigg)
A_0(z)G^{1,\mathbf{f}}(z)-\\
-\int_{\mathrm{B}_{\Lambda/\sigma}}\mathrm{d}^2y\int_{\mathrm{B}_{1}}\mathrm{d}^2z\,
\bigg(G^{1,\mathbf{f}}(x-y)+\frac{\ln(\Lambda)}{2\pi}\bigg)
\bigg(G^{1,\mathbf{f}}(y-z)+\frac{\ln(\Lambda)}{2\pi}\bigg)
A_0(z)G^{1,\mathbf{f}}(z).
\end{multline*}
It is clear that the terms with $\ln(\Lambda)/(2\pi)$, which follow from the first multiplier, are reduced. Then, using information about the behavior of the averaged logarithm, we can conclude that the integral
\begin{equation*}
\int_{\mathrm{B}_{\Lambda/\sigma}}\mathrm{d}^2y\int_{\mathrm{B}_{1}}\mathrm{d}^2z\,
G^{1,\mathbf{f}}(y)
G^{1,\mathbf{f}}(y-z)
A_0(z)G^{1,\mathbf{f}}(z)
-\int_{\mathrm{B}_{\Lambda/\sigma}}\mathrm{d}^2y\int_{\mathrm{B}_{1}}\mathrm{d}^2z\,
G^{1,\mathbf{f}}(x-y)
G^{1,\mathbf{f}}(y-z)
A_0(z)G^{1,\mathbf{f}}(z)
\end{equation*}
is a finite function and therefore falls into the correction $\mathcal{O}(1)$. The remaining part, taking into account the first property, leads to the stated answer
\begin{equation*}
\frac{L}{2\pi}\bigg(\int_{\mathrm{B}_{\Lambda/\sigma}}\mathrm{d}^2y\int_{\mathrm{B}_{1}}\mathrm{d}^2z\,
G^{1,\mathbf{f}}(y)
	A_0(z)G^{1,\mathbf{f}}(z)
	-\int_{\mathrm{B}_{\Lambda/\sigma}}\mathrm{d}^2y\int_{\mathrm{B}_{1}}\mathrm{d}^2z\,
G^{1,\mathbf{f}}(x-y)
	A_0(z)G^{1,\mathbf{f}}(z)\bigg)=\frac{|x|^2L}{8\pi}.
\end{equation*}
Here are two examples of applying properties that appear in calculations:
\begin{equation*}
\int_{\mathrm{B}_{1/\Lambda}}\mathrm{d}^2x\,\Big(A_0(x)G^{\Lambda,\mathbf{f}}(x)\Big)^2G_2^{\mu\mu}(x)=
L\bigg(
2\alpha_5(\mathbf{f})\alpha_6(\mathbf{f})+\frac{\alpha_8(\mathbf{f})}{8\pi^2}
\bigg)+\mathcal{O}(1),
\end{equation*}
\begin{equation*}
\int_{\mathrm{B}_{1/\Lambda}}\mathrm{d}^2x\,\Big(A_0(x)G^{\Lambda,\mathbf{f}}(x)\Big)
\Big(\partial_{x_\nu}G^{\Lambda,\mathbf{f}}(x)\Big)
\Big(\partial_{x^\nu}G_2^{\mu\mu}(x)\Big)=
\frac{L\alpha_{11}(\mathbf{f})}{4\pi^2}+\mathcal{O}(1).
\end{equation*}
\noindent\textbf{Property 4.} Taking into account the definitions from \eqref{33-9-24-53} and \eqref{33-9-24-55} for two functionals, the following equalities are true
\begin{equation*}
\alpha_9(\mathbf{f})=0,\,\,\,\alpha_{11}(\mathbf{f})=-\frac{1}{4}.
\end{equation*}
To prove this, let us first consider the second integral
\begin{equation*}
\alpha_{11}(\mathbf{f})=\pi\int_{\mathbb{R}^2}\mathrm{d}^2x\,\Big(A_0(x)G^{1,\mathbf{f}}(x)\Big)
x^\nu\partial_{x^\nu}G^{1,\mathbf{f}}(x).
\end{equation*}
Note that the functions depend on the absolute value of the argument. Therefore, the operators can be rewritten in the polar coordinates
\begin{equation*}
A_0(x)=-r^{-1}\partial_rr\partial_r,\,\,\,x^\nu\partial_{x^\nu}=r\partial_r,
\end{equation*}
where $r=|x|$. Then, substituting and integrating by parts, we obtain the following chain of equalities
\begin{equation*}
\alpha_{11}(\mathbf{f})=-2\pi^2\int_{0}^{+\infty}\mathrm{d}r\,\Big(\partial_rr\partial_rG^{1,\mathbf{f}}(x)\Big)
r\partial_rG^{1,\mathbf{f}}(x)=-\pi^2\Big(r\partial_rG^{1,\mathbf{f}}(x)\Big)^2\Big|_{r\to+\infty}
=-\frac{1}{4}.
\end{equation*}
Next, note that the functional $\alpha_{11}(\mathbf{f})$ is part of $\alpha_{9}(\mathbf{f})$, so it is enough to calculate only the second part
\begin{equation*}
\alpha_{9}(\mathbf{f})-\alpha_{11}(\mathbf{f})=
-\pi\int_{\mathbb{R}^2}\mathrm{d}^2x\,\Big(A_0(x)G^{1,\mathbf{f}}(x)\Big)
x^\nu\partial_{x^\nu}z_2(x),
\end{equation*}
where the auxiliary function has the form
\begin{equation*}
z_2(x)=\int_{\mathrm{B}_1}\mathrm{d}^2y\,G^{1,\mathbf{f}}(x-y)\Big(A_0(y)G^{1,\mathbf{f}}(y)\Big).
\end{equation*}
First of all, we note that the function $z_2(x)$ depends on the absolute value of $|x|$, therefore, the differential operators can be rewritten in polar coordinates without the angular part. Next, we integrate it again by parts, then we get
\begin{equation*}
\alpha_{9}(\mathbf{f})-\alpha_{11}(\mathbf{f})=\pi^2\Big(r\partial_rG^{1,\mathbf{f}}(x)\Big)\Big(r\partial_rz_2(x)\Big)\Big|_{r\to+\infty}=\frac{1}{4},
\end{equation*}
where the relations were used
\begin{equation*}
z_2(x)=\int_{\mathrm{B}_1}\mathrm{d}^2y\,G^{1,\mathbf{f}}(x-y)\Big(A_0(y)G^{1,\mathbf{f}}(y)\Big)=
-\frac{\ln(|x|\sigma)}{2\pi}\int_{\mathrm{B}_1}\mathrm{d}^2y\,A_0(y)G^{1,\mathbf{f}}(y)=G^{1,\mathbf{f}}(x)
\end{equation*}
for $|x|>2$, as well as
\begin{equation*}
-2\pi^2\int_{0}^{+\infty}\mathrm{d}r\,\Big(r\partial_rG^{1,\mathbf{f}}(x)\Big)
\partial_rr\partial_rz_2(x)=
\pi\int_{\mathbb{R}^2}\mathrm{d}^2x\,\Big(x^\nu\partial_{x^\nu}G^{1,\mathbf{f}}(x)\Big)
A_0(x)z_2(x)=-\alpha_{9}(\mathbf{f})+\alpha_{11}(\mathbf{f}).
\end{equation*}
Finally, from the equality $\alpha_{9}(\mathbf{f})-\alpha_{11}(\mathbf{f})=1/4$ we obtain $\alpha_{9}(\mathbf{f})=0$.\\

\noindent\textbf{Property 5.} In addition, we consider another auxiliary integral, the asymptotics of which must be calculated using a less trivial method. It has the form
\begin{equation*}
\int_{\mathrm{B}_{1/\sigma_1}}\mathrm{d}^2x\,
\Big(\partial_{x_\nu}G^{\Lambda,\mathbf{f}}(x)\Big)
\Big(\partial_{x^\nu}G^{\Lambda,\mathbf{f}}(x)\Big)
\partial_{x^\mu}G_1^\mu(x)=
\frac{L_1^2}{8\pi^2}+\mathcal{O}(1).
\end{equation*}
Let us show its calculation. We make a few comments. First, the singular part of the asymptotics does not depend on which final value of $\sigma_1>0$ will be fixed. For further information, it is convenient to choose $1/\sigma_1=1/\hat{\sigma}-1/\Lambda$. Secondly, the application of the second property for the function $\partial_{x^\mu}G_1^\mu(x)$ is impossible, since it was derived with a variable restriction. Therefore, we need to use the definition from \eqref{33-r-21} followed by integration by parts. Then, after scaling the variables, the integral takes the form
\begin{equation*}
\int_{\mathrm{B}_{\Lambda/\hat{\sigma}-1}}\mathrm{d}^2x\,
\Big(\partial_{x_\nu}G^{1,\mathbf{f}}(x)\Big)
\Big(\partial_{x^\nu}G^{1,\mathbf{f}}(x)\Big)
\bigg(
\frac{\ln(\sigma\Lambda/\hat{\sigma})}{2\pi}+
\int_{x+\mathrm{B}_1}\mathrm{d}^2y\,
A_0(x)G^{1,\mathbf{f}}(x-y)
G^{1,\mathbf{f}}(y)
\bigg).
\end{equation*}
The convenient choice of $\sigma_1$ led to the possibility of making a shift of $y\to y+x$ by centering the multiplier with the Laplace operator. Breaking it down into terms, we get
\begin{multline*}
\frac{\ln(\sigma\Lambda/\hat{\sigma})}{2\pi}
\bigg(\int_{\mathrm{B}_{1}}\mathrm{d}^2x\,
\Big(\partial_{x_\nu}G^{1,\mathbf{f}}(x)\Big)
\Big(\partial_{x^\nu}G^{1,\mathbf{f}}(x)\Big)
+\frac{\ln(\Lambda/\hat{\sigma}-1)}{2\pi}\bigg)
+\\+
\int_{\mathrm{B}_{\Lambda/\hat{\sigma}-1}}\mathrm{d}^2x\,
\Big(\partial_{x_\nu}G^{\Lambda,\mathbf{f}}(x)\Big)
\Big(\partial_{x^\nu}G^{\Lambda,\mathbf{f}}(x)\Big)
\int_{\mathrm{B}_1}\mathrm{d}^2y\,
G^{1,\mathbf{f}}(x-y)
A_0(y)G^{1,\mathbf{f}}(y).
\end{multline*}
Next, we note that it is possible to make a substitution in the last integral
\begin{equation*}
\int_{\mathrm{B}_1}\mathrm{d}^2y\,
G^{1,\mathbf{f}}(x-y)
A_0(y)G^{1,\mathbf{f}}(y)\longrightarrow G^{1,\mathbf{f}}(x),
\end{equation*}
since the support of their difference is contained in $\mathrm{B}_2$. In addition, the integral in the first term can be calculated explicitly
\begin{equation*}
\int_{\mathrm{B}_{1}}\mathrm{d}^2x\,
\Big(\partial_{x_\nu}G^{1,\mathbf{f}}(x)\Big)
\Big(\partial_{x^\nu}G^{1,\mathbf{f}}(x)\Big)=
\int_{\mathrm{B}_{1}}\mathrm{d}^2x\,
\Big(A_0(x)G^{1,\mathbf{f}}(x)\Big)G^{1,\mathbf{f}}(x)\big|_{\sigma=1}=\frac{\theta_1+\mathbf{f}(0)/2}{2\pi}.
\end{equation*}
Then the original integral, after removing some finite parts, is rewritten as
\begin{equation*}
\frac{\ln(\sigma\Lambda/\hat{\sigma})}{4\pi^2}
\big(\theta_1+\mathbf{f}(0)/2+\ln(\Lambda/\hat{\sigma})\big)
+\frac{1}{4}\bigg(\mathrm{I}_1(\Lambda,\hat{\sigma})-\frac{\ln(\sigma\Lambda/\hat{\sigma})}{\pi^2}\big(\theta_1+\mathbf{f}(0)/2+\ln(\Lambda/\hat{\sigma})\big)\bigg),
\end{equation*}
where the definition from \eqref{33-w-7} was used. Remembering that $\theta_1=-\ln(\sigma/\hat{\sigma})$, we obtain
\begin{equation*}
\frac{(L_1+\ln(\sigma/\hat{\sigma}))^2+2L\theta_1}{8\pi^2},
\end{equation*}
from which final answer follows.


\begin{thebibliography}{99}
\bibitem{3}
L. D. Faddeev, A. A. Slavnov, \textit{Gauge Fields: An Introduction to Quantum Theory}, Frontiers in Physics \textbf{83}, Addison-Wesley, 1--236 (1991)
\bibitem{9}
C. Itzykson, J. B. Zuber, \textit{Quantum Field Theory}, Mcgraw-hill, New York, 1--705 (1980)
\bibitem{10}
M. E. Peskin, D. V. Schroeder, \textit{An Introduction to Quantum Field Theory}, Addison-Wesley, 1--868 (1995)

\bibitem{sk-1}
R. P. Feynman, A. R. Hibbs, D. F. Styer, \textit{Quantum Mechanics and Path Integrals}, Mineola, NY: Dover Publications, 1--371 (2010)
\bibitem{34-c-m}
P. Cartier, C. DeWitt-Morette, \textit{A Rigorous Mathematical Foundation of Functional Integration}, In: DeWitt-Morette, C., Cartier, P., Folacci, A. (eds) Functional Integration. NATO ASI Series, vol \textbf{361}, Springer, Boston (1997) doi:10.1007/978-1-4899-0319-8\_1
\bibitem{sk-2}
J. Zinn Justin, \textit{Path Integrals in Quantum Mechanics}, Oxford University Press, 1--334 (2004)
\bibitem{34-6}
E. T. Shavgulidze, O.G. Smolyanov, \textit{Functional integrals}, Moscow, URSS, 1--328 (2015)

\bibitem{6}
J. C. Collins, \textit{Renormalization: An Introduction to Renormalization, the Renormalization Group and the Operator-Product Expansion}, Cambridge University Press, 1--392 (1984)
\bibitem{7}
O. I. Zavialov, \textit{Renormalized quantum field theory}, Kluwer Academic Publishers, Dodrecht, Boston, 1--524 (1990)
\bibitem{105}
D. I. Kazakov, \textit{Radiative Corrections, Divergences, Regularization, Renormalization, Renormalization Group and All That in Examples in Quantum Field Theory}, arXiv:0901.2208 [hep-ph] (2009)


\bibitem{sk-3}
M. Atiyah, \textit{Topological quantum field theories}, Publications Mathématiques de l’Institut des Hautes Scientifiques \textbf{68}, 175--186 (1988) doi:10.1007/BF02698547
\bibitem{sk-4}
G. Segal, \textit{The definition of conformal field theory}, Topology, geometry and quantum field theory, London Math. Soc. Lect. Note Ser. \textbf{308}, Cambridge University Press, 421--577 (2004) doi:10.1017/CBO9780511526398.019
\bibitem{sk-5}
N. Reshetikhin, \textit{Lectures on Quantization of Gauge Systems}, New Paths Towards Quantum
Gravity, Springer Berlin Heidelberg, 125--190 (2010) doi:10.1007/978-3-642-11897-5$\_$3
\bibitem{sk-6}
S. Stolz, P. Teichner, \textit{Supersymmetric field theories and generalized cohomology}, Mathematical foundations of quantum field theory and perturbative string theory, Proc. Sympos.
Pure Math. \textbf{83}, 279--340 (2011) arXiv:1108.0189
\bibitem{sk-7}
A. S. Cattaneo, P. Mnev, N. Reshetikhin, \textit{Perturbative Quantum Gauge Theories on Manifolds with Boundary}, Commun. Math. Phys. \textbf{357}, 631--730 (2018) doi:10.1007/s00220-017-3031-6
\bibitem{sk-14}
S. Kandel, \textit{Functorial quantum field theory in the Riemannian setting}, Adv. Theor. Math.
Phys. \textbf{20}(6), 1443--1471 (2016) doi:10.4310/ATMP.2016.v20.n6.a5
\bibitem{sk-16}
S. Kandel, P. Mnev, K. Wernli, \textit{Two-dimensional perturbative scalar QFT and Atiyah-Segal gluing}, Adv. Theor. Math.
Phys. \textbf{25}(7), 1847--1952 (2021) doi:10.4310/ATMP.2021.v25.n7.a5
\bibitem{sksk}
A. V. Ivanov, \textit{Effective actions, cutoff regularization, quasi-locality, and gluing of partition functions}, J. Phys. A: Math. Theor., \textbf{58}, 135401 (2025) doi:10.1088/1751-8121/adc3de, arXiv:2411.13857, https://www.pdmi.ras.ru/preprint/2024/24-11.html

\bibitem{19}
C. G. Bollini, J. J. Giambiagi, \textit{Dimensional Renormalization: The Number of Dimensions as a Regularizing Parameter}, Nuovo Cim. B, \textbf{12}, 20--26 (1972)
\bibitem{555}
G. 't Hooft, M. Veltman, \textit{Regularization and renormalization of gauge fields}, Nucl. Phys. B \textbf{44}, 189--213 (1972)


\bibitem{AA-1}
A. A. Slavnov, \textit{Invariant regularization of non-linear chiral theories}, Nucl. Phys. B, \textbf{31}(2), 301--315 (1971) doi:10.1016/0550-3213(71)90234-3 
\bibitem{AA-2}
A. A. Slavnov, \textit{Invariant regularization of gauge theories}, Theoret. and Math. Phys., \textbf{13}:2 (1972) 1064--1066 doi:10.1007/BF01035526
\bibitem{Bakeyev-Slavnov}
T. Bakeyev, A. Slavnov, \textit{Higher covariant derivative regularization revisited}, Mod. Phys. Lett. A \textbf{11}(19), 1539--1554 (1996)
\bibitem{29-st}
K. V. Stepanyantz, \textit{The Higher Covariant Derivative Regularization as a Tool for Revealing the Structure of Quantum Corrections in Supersymmetric Gauge Theories},  Proc. Steklov Inst. Math. \textbf{309}, 284--298 (2020)

\bibitem{chi-0}
A. Brizola, O. Battistel, Marcos Sampaio, M. C. Nemes, \textit{Implicit Regularisation Technique: Calculation of the Two-loop $\phi^4_4$-theory $\beta$-function}, Mod. Phys. Lett. A, \textbf{14}, 1509--1518 (1999)
\bibitem{chi-1}
A. L. Cherchiglia, M. Sampaio, M. C. Nemes, \textit{Systematic Implementation of Implicit Regularization for Multi-Loop Feynman Diagrams}, Int. J. Mod. Phys. A, \textbf{26}, 2591--2635 (2011)
\bibitem{chi-2}
A. Cherchiglia, D. C. Arias-Perdomo, A. R. Vieira, M. Sampaio, B. Hiller, \textit{Two-loop renormalisation of gauge theories in 4D Implicit Regularisation and connections to dimensional methods}, Eur. Phys. J. C, \textbf{81}, 468 (2021)

\bibitem{FF-1}
R. P. Feynman, \textit{Relativistic Cut-Off for Quantum Electrodynamics}, Phys. Rev. \textbf{74}, 1430--1438 (1948) doi:10.1103/PhysRev.74.1430
\bibitem{Bog-R}
N. N. Bogolyubov, D. V. Shirkov, \textit{Introduction to the Theory of Quantized Fields}, Willey, New York, 1--620 (1980)
\bibitem{Pauli-Villars}
W. Pauli, F. Villars, \textit{On the Invariant Regularization in Relativistic Quantum Theory}, Rev. Mod. Phys. \textbf{21}(3): 434--444 (1949)

\bibitem{Sh}
S. L. Shatashvili, \textit{Two-loop approximation in the background field formalism}, Theor Math Phys \textbf{58}, 144--150 (1984) doi:10.1007/BF01017919
\bibitem{w6}
M. Oleszczuk, \textit{A symmetry-preserving cut-off regularization},
Z. Phys. C, \textbf{64}, 533--538 (1994)
\bibitem{w7}
Sen-Ben Liao, \textit{Operator Cutoff Regularization and Renormalization Group in 
	Yang-Mills Theory}, Phys. Rev. D, \textbf{56}, 5008--5033 (1997)
\bibitem{w8}
G. Cynolter, E. Lendvai, \textit{Cutoff Regularization Method in Gauge Theories},
[arXiv:1509.07407 [hep-ph]] (2015)
\bibitem{Khar-2020}
N. V. Kharuk, \textit{Mixed type regularizations and nonlogarithmic singularities}, Questions of quantum field theory and statistical physics. Part 27, Zap. Nauchn. Sem. POMI, \textbf{494}, POMI, St. Petersburg, 2020, 242--249; J. Math. Sci. (N. Y.), \textbf{264}, 362--367 (2022) 10.1007/s10958-022-06003-7
\bibitem{sk-b-19}
A. A. Bagaev, \textit{Two-loop calculations of the matrix $\sigma$-model effective action in the background field formalism}, Theor Math Phys, \textbf{154}:2, 303--310 (2008) doi:10.1007/s11232-008-0028-5

\bibitem{34}
A. V. Ivanov, N. V. Kharuk, \textit{Quantum equation of motion and two-loop cutoff renormalization for $\phi^3$ model}, Zap. Nauchn. Sem. POMI, \textbf{487}, POMI, St. Petersburg, 2019, 151--166; J. Math. Sci. (N. Y.), \textbf{257}:4 (2021), 526--536, arXiv:2203.04562, doi:10.1007/s10958-021-05500-5
\bibitem{Ivanov-Kharuk-2020}
A. V. Ivanov, N. V. Kharuk, \textit{Two-Loop Cutoff Renormalization of 4-D Yang--Mills Effective Action}, 2020 J. Phys. G: Nucl. Part. Phys. \textbf{48}, 015002, arXiv:2004.05999, 
doi:10.1088/1361-6471/abb939
\bibitem{Ivanov-Kharuk-20222}
A. V. Ivanov, N. V. Kharuk, \textit{Formula for two-loop divergent part of 4-D Yang--Mills effective action},  Eur. Phys. J. C \textbf{82}, 997 (2022), arXiv:2203.07131, doi:10.1140/epjc/s10052-022-10921-w
\bibitem{Ivanov-Akac}
P. V. Akacevich, A. V. Ivanov, \textit{On Two-Loop Effective Action of 2D Sigma Model},  Eur. Phys. J. C \textbf{83}, 653 (2023), arXiv:2304.02374, doi:10.1140/epjc/s10052-023-11797-0
\bibitem{Ivanov-Kharuk-2023}
A. V. Ivanov, N. V. Kharuk, \textit{Three-loop divergences in effective action of
	4-dimensional Yang--Mills theory with cutoff regularization: $\Gamma_4^2$-contribution}, Zap. Nauchn. Sem. POMI, \textbf{520}, POMI, St. Petersburg, 2023, 162--188, J Math Sci \textbf{284}, 681--699 (2024) doi:10.1007/s10958-024-07379-4
\bibitem{Iv-2024-1}
A. V. Ivanov, \textit{Three-loop renormalization of the quantum action for a four-dimensional scalar model with quartic interaction with the usage of the background field method and a cutoff regularization}, Nucl. Phys. B, \textbf{1006}, 116647 (2024), doi:10.1016/j.nuclphysb.2024.116647, arXiv:2402.14549, https://www.pdmi.ras.ru/preprint/2024/24-02.html
\bibitem{Iv-Kh-2024}
A. V. Ivanov, N. V. Kharuk, \textit{Three-loop renormalization of the quantum action for a five-dimensional scalar cubic model with the usage of the background field method and a cutoff regularization}, Eur. Phys. J. Plus \textbf{139}, 849 (2024) doi:10.1140/epjp/s13360-024-05648-4,
arXiv:2404.07513, https://www.pdmi.ras.ru/preprint/2024/24-05.html
\bibitem{Kh-2024}
N. V. Kharuk, \textit{Three-loop renormalization with a cutoff in a sextic model}, Questions of quantum field theory and statistical physics. Part 30, Zap. Nauchn. Sem. POMI, \textbf{532}, POMI, St. Petersburg, 2024, 273--286 https://mathscinet.ams.org/mathscinet/relay-station?mr=4811943
https://www.mathnet.ru/eng/znsl7462
\bibitem{Kh-25}
N. V. Kharuk, \textit{Four-loop renormalization with a cutoff in a sextic model}, (2025) 	arXiv:2504.07688, http://www.pdmi.ras.ru/preprint/2025/25-02.html

\bibitem{HB1}
A. M. Polyakov, \textit{Interaction of goldstone particles in two dimensions. Applications to ferromagnets and massive Yang-Mills fields}, Phys. Lett. B, \textbf{59}(1), 79--81 (1975)
\bibitem{HB2}
A. A. Migdal, \textit{Phase transitions in gauge and spin-lattice systems}, Sov. Phys. JETP, \textbf{42}(4), 743--746 (1976)
\bibitem{HB3}
E. Brezin, J. Zinn-Justin, \textit{Renormalization of the Nonlinear $\sigma$ Model in $2+\varepsilon$ Dimensions -- Application to the Heisenberg Ferromagnets}, Phys. Rev. Lett., \textbf{36}, 691--694 (1976)
\bibitem{HB7}
E. Brezin, J. Zinn-Justin, \textit{Spontaneous breakdown of continuous symmetries near two dimensions}, Phys. Rev. B, \textbf{14}, 3110--3120 (1976)
\bibitem{HB4}
E. Brezin, J. Zinn-Justin, J. C. Le Guillou, \textit{Renormalization of the nonlinear $\sigma$ model in $2+\varepsilon$ dimensions}, Phys. Rev. D, \textbf{14}(10), 2615--2621 (1976)
\bibitem{HB}
S. Hikami, E. Brezin, \textit{Three-loop calculations in the two-dimensional non-linear $\sigma$ model}, J. Phys. A: Math. Gen. \textbf{11}, 1141 (1978)
\bibitem{sig2}
D. Friedan, \textit{Nonlinear models in $2+\varepsilon$ dimensions}, Ann. Phys. \textbf{163}, 318--419 (1985)
\bibitem{sig1}
A. M. Polyakov, \textit{Gauge Fields and Strings}, London, Taylor and Francis Group, 1--312 (1987)
\bibitem{q-2}
I. Bakas, \textit{Renormalization group equations and geometric flows}, (2007) doi:10.48550/arXiv.hep-th/0702034
\bibitem{q-1}
K. Anagnostopoulos, K. Farakos, P.Pasipoularides, A. Tsapalis, \textit{Non-Linear Sigma Model and asymptotic freedom at the Lifshitz point}, (2010) doi:10.48550/arXiv.1007.0355
\bibitem{q-3}
M. Carfora, \textit{Renormalization Group and the Ricci Flow}, Milan Journal of Mathematics, \textbf{78}(1), 319--353 (2010) doi:10.1007/s00032-010-0110-y 
\bibitem{q-6}
A. N. Efremov, A. Rançon, \textit{Non-linear sigma models on constant curvature target manifolds: a functional renormalization group approach}, Phys. Rev. D \textbf{104}, 105003 (2021) doi:10.1103/PhysRevD.104.105003
\bibitem{q-14}
N. Levine, \textit{Universal 1-loop divergences for integrable sigma models}, J. High Energ. Phys. \textbf{2023}, 3 (2023) doi:10.1007/JHEP03(2023)003


\bibitem{Ivanov-2022}
A. V. Ivanov, \textit{Explicit Cutoff Regularization in Coordinate Representation}, 2022 J. Phys. A: Math. Theor. \textbf{55}, 495401, arXiv:2209.01783, doi:10.1088/1751-8121/aca8dc
\bibitem{Iv-2024}
A. V. Ivanov, \textit{An applicability condition of a cutoff regularization in the coordinate representation}, Funct Anal Its Appl \textbf{59}, 1--10 (2025) arXiv:2403.09218, https://www.pdmi.ras.ru/preprint/2024/24-04.html, https://doi.org/10.1134/S123456782501001X;
Translated from: Funktsional'nyi Analiz i ego Prilozheniya, \textbf{59}:1, 5--17 (2025) https://doi.org/10.4213/faa4221
\bibitem{sk-b-20}
V. P. Zastavnyi, \textit{The continuation of the radial function from the exterior of the ball to a function positively defined on the entire space}, Bulletin of Donetsk National University: Series A. Natural Sciences, 2/2024, 14--28 (2024) doi:10.5281/zenodo.13752079

\bibitem{HB5}
M. T. Grisaru, A. E. M. Van De Ven, D. Zanon, \textit{Two-dimensional supersymmetric sigma-models on Ricci-flat Kähler manifolds are not finite}. Nucl. Phys. B, \textbf{277}, 388--408 (1986)
\bibitem{HB6}
M. T. Grisaru, A. E. M. Van De Ven, D. Zanon, \textit{Four-loop divergences for the N=1 supersymmetric non-linear sigma-model in two dimensions}. Nucl. Phys. B, \textbf{277}, 409--428 (1986)


\bibitem{q-15}
A. A. Slavnov, L. D. Faddeev, \textit{Invariant perturbation theory for nonlinear chiral Lagrangians}, Theor Math Phys \textbf{8}, 843--850 (1971) doi:10.1007/BF01029338
\bibitem{q-17}
I. Ya. Aref'eva, \textit{Phase transition in the three-dimensional chiral field}, Annals Phys., \textbf{117}:2, 393--406 (1979) doi:10.1016/0003-4916(79)90361-0
\bibitem{q-16}
M. A. Semenov-Tian-Shansky, L. D. Faddeev, \textit{On the theory of nonlinear chiral fields}, Vestnik Leningrad Univ. Math., \textbf{10}, 319--327 (1982)
\bibitem{q-12}
L. D Faddeev, N. Yu Reshetikhin, \textit{Integrability of the principal chiral field model in 1+1 dimension}, Ann. Phys. \textbf{167}(2), 227--256 (1986) doi:10.1016/0003-4916(86)90201-0
\bibitem{q-13}
G. Arutyunov, S. Frolov, \textit{Foundations of the $AdS_5\times S^5$ Superstring: Part I}, J. Phys. A: Math. Theor. \textbf{42}, 254003 (2009) doi:10.1088/1751-8113/42/25/254003
\bibitem{q-10}
T. Nguyen, \textit{Quantization of the nonlinear sigma model revisited}, J. Math. Phys. \textbf{57}, 082301 (2016) doi:10.1063/1.4961153
\bibitem{q-11}
J. Evslin, B. Zhang, \textit{Mass-gap in the compactified principal chiral model}, Phys. Rev. D \textbf{98}, 085016 (2018) doi:10.1103/PhysRevD.98.085016
\bibitem{q-9}
C. D. Batista, M. Shifman, Z. Wang, S.-S. Zhang, \textit{Principal Chiral Model in Correlated Electron Systems}, Phys. Rev. Lett. \textbf{121}, 227201 (2018) doi:10.1103/PhysRevLett.121.227201
\bibitem{q-5}
J. Kluson, \textit{Nonrelativistic String Theory Sigma Model and Its Canonical Formulation}, Eur. Phys. J. C \textbf{79}, 108 (2019) doi:10.1140/epjc/s10052-019-6623-9
\bibitem{q-7}
F. Bascone, F. Pezzella, \textit{Principal Chiral Model without and with WZ term: Symmetries and Poisson-Lie T-Duality}, PoS, \textbf{376}, CORFU2019 134 (2020) doi:10.22323/1.376.0134
\bibitem{q-8}
B. Hoare, \textit{Integrable deformations of sigma models}, J. Phys. A: Math. Theor. \textbf{55}, 093001 (2022) doi:10.1088/1751-8121/ac4a1e
\bibitem{q-4}
G. Oling, Z. Yan, \textit{Aspects of Nonrelativistic Strings}, Front. Phys., \textbf{10}:832271 (2022)
doi:10.3389/fphy.2022.832271
\bibitem{33-ma-1}
Y. Liu, M. Mariño, \textit{Trans-series from condensates in the non-linear sigma model}, (2025) arXiv:2507.02605

\bibitem{29-1-0}
L. D. Faddeev, \textit{Scenario for the renormalization in the 4D Yang--Mills theory}, Int. J. Mod. Phys. A, \textbf{31}, 1630001 (2016)
\bibitem{29-1-1}
C. E. Derkachev, A. V. Ivanov, L. D. Faddeev, \textit{Renormalization scenario for the quantum Yang--Mills theory in four-dimensional space-time}, TMF, \textbf{192}:2 (2017), 227--234; Theoret. and Math. Phys., \textbf{192}:2 (2017), 1134--1140 10.1134/S0040577917080049

\bibitem{2}
M. Nakahara, \textit{Geometry, topology and physics}, Second Edition, CRC Press, 1--573 (2003)

\bibitem{102}
B. S. DeWitt, \textit{Quantum Theory of Gravity. 2. The Manifestly Covariant
	Theory}, Phys. Rev. \textbf{162}, 1195--1239 (1967)
\bibitem{103}
B. S. DeWitt, \textit{Quantum Theory of Gravity. 3. Applications of the Covariant Theory}, Phys. Rev. \textbf{162}, 1239--1256 (1967)
\bibitem{24}
G. ’t Hooft, \textit{The background field method in gauge field theories}, (Karpacz, 1975), Proceedings, Acta Universitatis Wratislaviensis, \textbf{1}, Wroclaw, 345--369 (1976)
\bibitem{25}
L. F. Abbott, \textit{Introduction to the background field method}, Acta Phys. Polon. B, \textbf{13}:1--2, 33--50 (1982)
\bibitem{26}
I. Ya. Aref'eva, A. A. Slavnov, L. D. Faddeev, \textit{Generating functional for the S-matrix in gauge-invariant theories}, TMF, \textbf{21}:3, 311--321 (1974)
\bibitem{33-d-1}
P. S. Howe, G. Papadopoulos, K. S. Stelle, \textit{The background field method and the non-linear $\sigma$-model}, Nucl. Phys. B, \textbf{296}(1), 26--48 (1988) doi:10.1016/0550-3213(88)90379-3 

\bibitem{Fa-1}
L. D. Faddeev, \textit{A couple of methodological comments on the quantum Yang-Mills theory}, Theor Math Phys \textbf{181}, 1638--1642 (2014) doi:10.1007/s11232-014-0240-4
\bibitem{Ba-1}
A. A. Bagaev, \textit{Renormalization of the quantum equation of motion for Yang–Mills fields in background formalism}, J Math Sci \textbf{138}, 5631--5635 (2006) doi:10.1007/s10958-006-0331-3
\bibitem{Ba-2}
A. A. Bagaev, \textit{A remark on renormalization of the quantum equation of motion for the matrix sigma model}, J Math Sci \textbf{151}, 2813--2815 (2008) doi:10.1007/s10958-008-9007-5

\bibitem{Vas-98}
A. N. Vasiliev, \textit{Functional Methods in Quantum Field Theory and Statistical Physics}, CRC Press, 1--320 (1998)

\bibitem{I-R}
A. V. Ivanov, M. A. Russkikh, \textit{Quantum Field Theory on the Example of the Simplest Cubic Model}, J Math Sci \textbf{275}, 306--325 (2023) https://doi.org/10.1007/s10958-023-06683-9
\bibitem{Iv-244}
A. V. Ivanov, \textit{Applicability condition of a cutoff in two-dimensional models}, Questions of quantum field theory and statistical physics. Part 30, Zap. Nauchn. Sem. POMI, 532, POMI, St. Petersburg, 153--168 (2024) https://www.mathnet.ru/eng/znsl7457

\bibitem{Gelfand-1964}
I. M. Gel'fand, G. E. Shilov, 
\textit{Generalized Functions, Volume 1: Properties and Operations},
AMS Chelsea Publishing \textbf{377}, 1--423 (1964)
\bibitem{Vladimirov-2002}
V. S. Vladimirov, \textit{Methods of the theory of generalized functions}, London,
CRC Press, 1--328 (2002)

\bibitem{13}
J. P. Bornsen, A. E. M. van de Ven, \textit{Three-loop Yang--Mills $\beta$-function via the covariant background field method}, Nucl. Phys. B, \textbf{657}, 257--303 (2003) doi:10.1016/S0550-3213(03)00118-4

\bibitem{29}
M. Lüscher, \textit{Dimensional regularisation in the presence of large background fields}, Annals of Physics \textbf{142}, 359--392 (1982)
\bibitem{vas1}
D. V. Vassilevich, \textit{Heat kernel expansion: user's manual}, Phys. Rept. \textbf{388}, 279--360 (2003) doi:10.1016/j.physrep.2003.09.002
\bibitem{vas2}
D. Fursaev, D. Vassilevich, \textit{Operators, Geometry and Quanta: Methods of Spectral Geometry in Quantum Field Theory}, Springer, 1--304 (2011)
\bibitem{33}
A. V. Ivanov, N. V. Kharuk, \textit{Heat kernel: Proper-time method, Fock--Schwinger gauge, path integral, and Wilson line}, TMF, \textbf{205}:2, 242--261, (2020); Theoret. and Math. Phys., \textbf{205}:2, 1456--1472 (2020) 10.1134/S0040577920110057 [arXiv:1906.04019]
\bibitem{30-1-1}
A. V. Ivanov, N. V. Kharuk, \textit{Special Functions for Heat Kernel Expansion}, Eur. Phys. J. Plus \textbf{137}, 1060 (2022), arXiv:2106.00294v2, 10.1140/epjp/s13360-022-03176-7

\bibitem{Kharuk-2021}
N. V. Kharuk, \textit{Zero modes of the Laplace operator in two-loop calculations in the Yang--Mills theory}, Questions of quantum field theory and statistical physics. Part 28, Zap. Nauchn. Sem. POMI, \textbf{509}, POMI, St. Petersburg, 216--226 (2021), J Math Sci \textbf{275}, 370--377 (2023) doi:10.1007/s10958-023-06687-5

\bibitem{He-1}
D. Gilbarg, N. Trudinger, \textit{Elliptic Partial Differential Equations of Second Order}, New York: Springer, 1--517 (1983)
\bibitem{smir}
V. I. Smirnov, \textit{A course of higher mathematics: volume 5}, Pergamon Press, 1--644 (1964)


\end{thebibliography}
\end{document}